\def\final{1}
\def\stocsubmit{0}

\documentclass[11pt]{article}

\usepackage{algorithm, algorithmic}
\usepackage{amsmath, amssymb, amsthm}
\usepackage{framed}
\usepackage{microtype}
\usepackage{fullpage}
\usepackage{graphicx}
\usepackage{verbatim}
\usepackage{color}
\definecolor{DarkGreen}{rgb}{0.1,0.5,0.1}
\definecolor{DarkRed}{rgb}{0.5,0.1,0.1}
\definecolor{DarkBlue}{rgb}{0.1,0.1,0.5}
\usepackage{enumitem}
\usepackage[pdftex]{hyperref}
\hypersetup{
    unicode=false,				% non-Latin characters in Acrobat¿s bookmarks
    pdftoolbar=true,				% show Acrobat toolbar?
    pdfmenubar=true,			% show Acrobat menu?
    pdffitwindow=false, 		% page fit to window when opened
    pdftitle={},    					% title
    pdfauthor={}
    pdfsubject={},				% subject of the document
    pdfnewwindow=true,		% links in new window
    pdfkeywords={},			% list of keywords
    colorlinks=true,				% false: boxed links; true: colored links
    linkcolor=DarkBlue,		% color of internal links
    citecolor=DarkBlue,	% color of links to bibliography
    filecolor=DarkBlue,		% color of file links
    urlcolor=DarkBlue,		% color of external links
}

\ifnum\final=0
\newcommand{\mynote}[1]{\marginpar{\tiny \sf #1}}
\else
\newcommand{\mynote}[1]{}
\fi
\newcommand{\mnote}[1]{\mynote{Mark: {#1}}}
\newcommand{\jnote}[1]{\mynote{Jon: {#1}}}

\newcommand\N{\mathbb{N}}
\newcommand\R{\mathbb{R}}

\newcommand{\cA}{\mathcal{A}}
\newcommand{\cB}{\mathcal{B}}

\newcommand{\cD}{\mathcal{D}}
\newcommand{\cE}{\mathcal{E}}

\newcommand{\calG}{\mathcal{G}}

\newcommand{\cM}{\mathcal{M}}

\newcommand{\cQ}{\mathcal{Q}}
\newcommand{\cR}{\mathcal{R}}

\newcommand{\cX}{\mathcal{X}}

\renewcommand{\dots}{\ldots}
\newcommand{\poly}{\mathrm{poly}}
\newcommand{\polylog}{\mathrm{polylog}}
\newcommand{\bits}{\{0,1\}}
\newcommand{\getsr}{\gets_{\mbox{\tiny R}}}
 
\newcommand{\set}[1]{\left\{#1\right\}} 
\newcommand{\from}{:}
\newcommand{\negl}{\mathrm{negl}}
\newcommand{\corr}[2]{\langle #1, #2 \rangle}

\newcommand{\ex}[1]{\mathbb{E}\left[#1\right]}
\DeclareMathOperator*{\Expectation}{\mathbb{E}}
\newcommand{\Ex}[2]{\Expectation_{#1}\left[#2\right]}
\DeclareMathOperator*{\Probability}{\mathrm{Pr}}
\newcommand{\prob}[1]{\mathrm{Pr}\left[#1\right]}
\newcommand{\Prob}[2]{\Probability_{#1}\left[#2\right]}
\DeclareMathOperator*{\argmin}{arg\,min}

\newcommand{\INDSTATE}[1][1]{\STATE\hspace{#1\algorithmicindent}}

\newtheorem{theorem}{Theorem}[section]
\newtheorem{lemma}[theorem]{Lemma}
\newtheorem{fact}[theorem]{Fact}
\newtheorem{claim}[theorem]{Claim}
\newtheorem{remark}[theorem]{Remark}
\newtheorem{corollary}[theorem]{Corollary}

\theoremstyle{definition}
\newtheorem{definition}[theorem]{Definition}

\newcommand{\db}{D}
\newcommand{\univ}{\cX}
\newcommand{\dbset}{\univ^{\rows}}
\newcommand{\san}{\cA}
\newcommand{\query}{q}
\newcommand{\queryj}[1]{\query_{#1}}
\newcommand{\queryset}{\cQ}
\newcommand{\answer}{a}
\newcommand{\answerq}[1]{\answer_{#1}}
\newcommand{\eps}{\varepsilon}
\newcommand{\rows}{n}
\newcommand{\row}{x}
\newcommand{\rowi}[1]{\row_{#1}}

\newcommand{\cols}{d}
\newcommand{\wcsc}{n^*} %Initial stab at the notation for the wc error sample complexity
\newcommand{\vcdim}{\mathit{VC}}
\newcommand{\adv}{\cB}
\newcommand{\privadv}{\adv'}
\newcommand{\recadv}{\adv}
\newcommand{\conj}{\cM}
\newcommand{\kdconj}[1]{\conj_{#1, d}}
\newcommand{\alphsize}{m}

\newcommand{\users}{\rows}
\newcommand{\len}{\cols}
\newcommand{\gen}{\mathit{Gen}}
\newcommand{\trace}{\mathit{Trace}}
\newcommand{\codebook}{C}
\newcommand{\codebookS}[1]{\codebook_{#1}}
\newcommand{\codeword}{c}
\newcommand{\codewordi}[1]{\codeword_{#1}}
\newcommand{\codewordij}[2]{\codeword_{#1#2}}
\newcommand{\pirateword}{c'}
\newcommand{\piratewordj}[1]{\pirateword_{#1}}
\renewcommand{\sec}{\xi}
\newcommand{\fpadv}{\san_{\mathit{FP}}}
\newcommand{\rob}{\beta}

\newcommand{\recqueryset}{\queryset}
\newcommand{\recquery}{\query}
\newcommand{\recuniv}{\univ}
\newcommand{\recdb}{D}
\newcommand{\recrow}{x}
\newcommand{\recrows}{n}

\newcommand{\privqueryset}{\queryset'}
\newcommand{\goodqueryset}{\queryset'_\mathit{good}}
\newcommand{\privquery}{\query'}
\newcommand{\privuniv}{\univ'}
\newcommand{\privdist}{\cD'}
\newcommand{\privdb}{D'}
\newcommand{\privrow}{x'}
\newcommand{\privrows}{n'}
\newcommand{\reidadv}{\adv'}

\newcommand{\compadv}{\adv^*}
\newcommand{\compdist}{\cD^*}
\newcommand{\compdb}{\db^*}

\title{Fingerprinting Codes and the \\ Price of Approximate Differential Privacy\thanks{A preliminary version of this work appeared in the Symposium on the Theory of Computing 2014.}}
\author{Mark Bun\thanks{Supported by an NDSEG Fellowship and NSF grant CNS-1237235.} \qquad \qquad Jonathan Ullman\thanks{Supported by NSF grant CNS-1237235.} \qquad \qquad Salil Vadhan\thanks{Supported by NSF grant CNS-1237235, a gift from Google, and a Simons Investigator Award.} \\
\ \\
School of Engineering and Applied Sciences \&\\
Center for Research on Computation and Society\\
Harvard University, Cambridge, MA \\
\texttt{\{mbun,jullman,salil\}@seas.harvard.edu} }

\begin{document}
\maketitle

\pagenumbering{gobble}
\begin{abstract}
We show new information-theoretic lower bounds on the sample complexity of $(\eps, \delta)$-differentially private algorithms that accurately answer large sets of counting queries.  A counting query on a database $D \in (\{0,1\}^d)^n$ has the form ``What fraction of the individual records in the database satisfy the property $q$?''  We show that in order to answer an arbitrary set $\cQ$ of $\gg d/\alpha^2$ counting queries on $D$ to within error $\pm \alpha$ it is necessary that
$$
n \geq \tilde{\Omega}\Bigg( \frac{\sqrt{d} \log |\cQ|}{\alpha^2 \eps} \Bigg).
$$
This bound is optimal up to poly-logarithmic factors, as demonstrated by the Private Multiplicative Weights algorithm (Hardt and Rothblum, FOCS'10).  In particular, our lower bound is the first to show that the sample complexity required for accuracy and $(\eps, \delta)$-differential privacy is asymptotically larger than what is required merely for accuracy, which is $O(\log |\cQ| / \alpha^2)$.  In addition, we show that our lower bound holds for the specific case of $k$-way marginal queries (where $|\cQ| = 2^k \binom{d}{k}$) when $\alpha$ is not too small compared to $d$ (e.g.~when $\alpha$ is any fixed constant).

Our results rely on the existence of short \emph{fingerprinting codes} (Boneh and Shaw, CRYPTO'95; Tardos, STOC'03), which we show are closely connected to the sample complexity of differentially private data release.  We also give a new method for combining certain types of sample complexity lower bounds into stronger lower bounds.
\end{abstract}
\thispagestyle{empty}

\newpage
\tableofcontents

\newpage
\pagenumbering{arabic}
\section{Introduction}

Consider a database $D \in \univ^n$, in which each of the $n$ rows corresponds to an individual's record, and each record is an element of some data universe $\univ$ (e.g.~$\univ = \bits^d$, corresponding to $d$ binary attributes per record).  The goal of privacy-preserving data analysis is to enable rich statistical analyses on such a database while protecting the privacy of the individuals.  It is especially desirable to achieve \emph{$(\eps, \delta)$-differential privacy}~\cite{DworkMcNiSm06,DworkKeMcMiNa06}, which (for suitable choices of $\eps$ and $\delta$) guarantees that no individual's data has a significant influence on the information released about the database.  A natural way to measure the tradeoff between these two goals is via \emph{sample complexity}---the minimum number of records $n$ such that there exists a (possibly computationally unbounded) algorithm that achieves both differential privacy and statistical accuracy.

Some of the most basic statistics are \emph{counting queries}, which are queries of the form ``What fraction of individual records in $D$ satisfy some property $q$?''  In particular, we would like to design an algorithm that takes as input a database $D$ and, for some family of counting queries $\cQ$, outputs an approximate answer to each of the queries in $\cQ$ that is accurate to within, say, $\pm .01$.  Suppose we are given a bound on the number of queries $|\cQ|$ and the dimensionality of the database records $d$, but otherwise allow the family $\cQ$ to be arbitrary.  What is the sample complexity required to achieve $(\eps, \delta)$-differential privacy and statistical accuracy for $\cQ$?

Of course, if we drop the requirement of privacy, then we could achieve perfect accuracy when $D$ contains any number of records.  However, in many interesting settings the database $D$ consists of random samples from some larger population, and an analyst is actually interested in answering the queries on the population.  Thus, even without a privacy constraint, $D$ would need to contain enough records to ensure that (with high probability) for every query $q \in \cQ$, the answer to $q$ on $D$ is close to the answer to $q$ on the whole population, say within $\pm .01$.  To achieve this form of \emph{statistical accuracy}, it is well-known that it is necessary and sufficient for $D$ to contain $\Theta(\log |\cQ|)$ samples.\footnote{For a specific family of queries $\cQ$, the necessary and sufficient number of samples is proportional to the \emph{VC-dimension} of $\cQ$, which can be as large as $\log |\cQ|$.}  In this work we consider whether there is an additional ``price of differential privacy'' if we require both statistical accuracy and $(\eps, \delta)$-differential privacy (for, say, $\eps = O(1)$, $\delta = o(1/n)$).  This benchmark has often been used to evaluate the utility of differentially private algorithms, beginning with the seminal work of Dinur and Nissim~\cite{DinurNi03}.

Some of the earliest work in differential privacy~\cite{DinurNi03,DworkNi04,BlumDwMcNi05,DworkMcNiSm06} gave an algorithm---the so-called \emph{Laplace mechanism}---whose sample complexity is $\tilde{\Theta}(|\cQ|^{1/2})$, and thus incurs a large price of differential privacy.  Fortunately, a remarkable result of Blum, Ligett, and Roth~\cite{BlumLiRo08} showed that the dependence on $|\cQ|$ can be improved exponentially to $O(d \log|\cQ|)$ where $d$ is the dimensionality of the data.  Their work was improved on in several important aspects~\cite{DworkNaReRoVa09,DworkRoVa10,RothRo10,HardtRo10,GuptaRoUl12,HardtLiMc12}.  The current best upper bound on the sample complexity is $O(\sqrt{d} \log|\cQ|)$, which is obtained via the private multiplicative weights mechanism of Hardt and Rothblum~\cite{HardtRo10}.

These results show that the price of privacy is small for datasets with few attributes, but may be large for high-dimensional datasets.  For example, if we simply want to estimate the mean of each of the $d$ attributes without a privacy guarantee, then $\Theta(\log d)$ samples are necessary and sufficient to get statistical accuracy.  However, the best known $(\eps, \delta)$-differentially private algorithm requires $\Omega(\sqrt{d})$ samples---an exponential gap.  In the special case of \emph{pure} $(\eps, 0)$-differential privacy, a lower bound of $\Omega(d)$ is known~\cite{HardtTa10}.  However, for the general case of \emph{approximate} $(\eps,\delta)$-differential privacy the best known lower bound is $\Omega(\log d)$~\cite{DinurNi03}.  More generally, there are no known lower bounds that separate the sample complexity of $(\eps,\delta)$-differential privacy from the sample complexity required for statistical accuracy alone.

In this work we close this gap almost completely, and show that there is indeed a ``price of approximate differential privacy'' for high-dimensional datasets.  
\begin{theorem}[Informal] \label{thm:main0}
Any algorithm that takes as input a database $D \in (\bits^d)^n$, satisfies approximate differential privacy, and estimates the mean of each of the $d$ attributes to within error $\pm 1/3$ requires $n \geq \tilde{\Omega}(\sqrt{d})$ samples.
\end{theorem}

We establish this lower bound using a combinatorial object called a \emph{fingerprinting code}, which was originally introduced by Boneh and Shaw~\cite{BonehSh98} for the problem of watermarking copyrighted content.  Specifically, we use Tardos' construction of optimal fingerprinting codes~\cite{Tardos08}.  The use of ``secure content distribution schemes'' to prove lower bounds for differential privacy originates with the work of Dwork et al.~\cite{DworkNaReRoVa09}, who used ``traitor-tracing schemes,'' which are a cryptographic analogue of information-theoretic fingerprinting codes, to prove computational hardness results for differential privacy.  Extending this connection, Ullman~\cite{Ullman13} used fingerprinting codes to construct a novel traitor-tracing scheme and obtain a strong computational hardness result for differential privacy.\footnote{In fact, one way to prove Theorem \ref{thm:main0} is by replacing the one-way functions in~\cite{Ullman13} with a random oracle, and thereby obtain an information-theoretically secure traitor-tracing scheme.}   Here we show that a \emph{direct} use of fingerprinting codes yields information-theoretic lower bounds on sample complexity.  

Using the additional structure of Tardos' fingerprinting code, we are able to prove \emph{statistical minimax lower bounds} for inferring the marginals of a product distribution from samples while guaranteeing differential privacy for the sample.  Specifically, suppose the database $D \in (\bits^d)^n$ consists of $n$ independent samples from a product distribution over $\bits^d$ such that the $i$-th coordinate of each sample is set to $1$ with probability $p_i$, for some unknown $p = (p_1,\dots,p_d) \in [0,1]^d$.  We show that if there exists a differentially private algorithm that takes such a database as input, satisfies approximate differential privacy, and outputs $\hat{p}$ such that $\| \hat{p} - p \|_\infty \leq 1/3$, then $n \geq \tilde{\Omega}(\sqrt{d}).$  Statistical minimax bounds of this type for differentially private inference problems were first studied by Duchi, Jordan, and Wainwright~\cite{DuchiJoWa13}, who proved minimax bounds for algorithms that satisfy the stronger constraint of \emph{local pure $(\eps, 0)$-differential privacy}.

\medskip

Next, we consider the sample complexity of answering an arbitrary set $\cQ$ of counting queries to within error $\pm \alpha$.  As above, if we assume the database contains samples from a population, and require only that the answers to queries on the sampled database and the population are close, to within $\pm \alpha$, then $\Theta(\log |\cQ| / \alpha^2)$ samples are necessary and sufficient for just statistical accuracy.  When $|\cQ|$ is large (relative to $d$ and $1/\alpha$), the best sample complexity for differential privacy is again achieved by the private multiplicative weights algorithm, and is $O(\sqrt{d}  \log|\cQ| / \alpha^2)$.  For pure differential privacy, a lower bound of $\Omega(d \log|\cQ| / \alpha^2)$ is known~\cite{HardtThesis}.  On the other hand, the best known lower bound for approximate differential privacy is $\Omega(\max\{ \log|\cQ| / \alpha, 1/\alpha^2 \})$, which follows from the techniques of~\cite{DinurNi03}.  To resolve this gap, we give a \emph{composition theorem} that allows us to obtain a nearly optimal lower bound by combining Theorem~\ref{thm:main0} with (variants of) the existing sample complexity lower bounds.  The result shows that the private multiplicative weights algorithm achieves nearly-optimal sample-complexity as a function of $|\cQ|, d$, and $\alpha$.
\begin{theorem}[Informal] \label{thm:main1}
For every sufficiently small $\alpha > 0$, $d \geq 6\log(1/\alpha)$, and $s \geq d/\alpha^2$, there exists a family of queries $\cQ$ of size $s$ such that any algorithm that takes as input a database $D \in (\bits^d)^n$, satisfies approximate differential privacy, and outputs an approximate answer to each query in $\cQ$ to within $\pm \alpha$ requires $n \geq \tilde{\Omega}(\sqrt{d} \log |\cQ| / \alpha^2)$.
\end{theorem} 
We remark that the condition that $d \geq 6\log(1/\alpha)$ is both necessary (up to the constant factor) and fairly mild.  Necessary because the \emph{noisy histogram algorithm} (see, e.g.~\cite{Vadhan16}) requires $n = O(2^{d/2} \sqrt{\log |\cQ|} / \alpha)$ samples, which is better than the conclusion of the lower bound when $d < 2\log(1/\alpha)$.  Mild because differential privacy cannot be satisfied for large query sets unless $\alpha \gtrsim 1/\sqrt{n}$, so the condition is no stronger than assuming $n \lesssim 2^{d/3}$, in which case the number of samples is exponential in the dimension.  Similarly, the condition $s \geq d/\alpha^2$ is also necessary, since adding independent noise to each query requires only $n \gtrsim |\cQ|^{1/2}/\alpha$ samples.

\medskip
Finally, we consider the sample complexity of the natural and well studied class of \emph{$k$-way marginal queries}, also known as \emph{$k$-way conjunction queries} (see e.g.~\cite{BarakChDwKaMcTa07,KasiviswanathanRuSmUl10,GuptaHaRoUl11,ThalerUlVa12,ChandrasekaranThUlWa13,DworkNiTa13}).  A $k$-way marginal query on a database $D \in (\{0,1\}^d)^n$ is specified by a set $S \subseteq [d]$, $|S| \leq k$, and a pattern $t \in \bits^{|S|}$ and asks ``What fraction of records in $D$ has each attribute $j$ in $S$ set to $t_j$?''  The number of $k$-way marginal queries on $\{0,1\}^d$ is about $2^{k} \binom{d}{k}$.  For the special case of $k=1$, the queries simply ask for the mean of each attribute, which was discussed above.  We prove that the lower bound of Theorem~\ref{thm:main1}, which applies to worst-case queries, also holds for the special case of $k$-way marginal queries when $\alpha$ is not too small.

\begin{theorem}[Informal] \label{thm:main2}
Any algorithm that takes a database $D \in (\bits^d)^n$, satisfies approximate differential privacy, and outputs an approximate answer to each of the $k$-way marginal queries to within $\pm \alpha$ (for $\alpha$ smaller than some universal constant and larger than an inverse polynomial in $d$) requires $n \geq \tilde{\Omega}(k \sqrt{d} / \alpha^2)$.
\end{theorem}

We remark that, since the number of $k$-way marginal queries is about $2^{k} \binom{d}{k}$, the sample complexity lower bound in Theoem~\ref{thm:main2} essentially matches that of Theorem~\ref{thm:main1}.  The two theorems are incomparable, since Theorem~\ref{thm:main1} applies even when $\alpha$ is exponentially small in $d$, but only applies for a worst-case family of queries.

\subsection{Our Techniques}

We now describe the main technical ingredients used to prove these results.  For concreteness, we will describe the main ideas for the case of $k$-way marginal queries.

\paragraph{Fingerprinting Codes.}
Fingerprinting codes, introduced by Boneh and Shaw~\cite{BonehSh98}, were originally designed to address the problem of watermarking copyrighted content.  Roughly speaking, a (fully-collusion-resilient) fingerprinting code is a way of generating codewords for $n$ users in such a way that any codeword can be uniquely traced back to a user.  Each legitimate copy of a piece of digital content has such a codeword hidden in it, and thus any illegal copy can be traced back to the user who copied it. Moreover, even if an arbitrary subset of the users collude to produce a copy of the content, then under a certain \emph{marking assumption}, the codeword appearing in the copy can still be traced back to one of the users who contributed to it.  The standard marking assumption is that if every colluder has the same bit $b$ in the $j$-th bit of their codeword, then the $j$-th bit of the ``combined'' codeword in the copy they produce must be also $b$.  We refer the reader to the original paper of Boneh and Shaw~\cite{BonehSh98} for the motivation behind the marking assumption and an explanation of how fingerprinting codes can be used to watermark digital content.

We show that the existence of short fingerprinting codes implies sample complexity lower bounds for $1$-way marginal queries.  Recall that a $1$-way marginal query $q_j$ is specified by an integer $j \in [d]$ and asks simply ``What fraction of records in $D$ have a $1$ in the $j$-th bit?''  Suppose a coalition of users takes their codewords and builds a database $D \in (\{0,1\}^d)^n$ where each record contains one of their codewords, and $d$ is the length of the codewords.  Consider the $1$-way marginal query $q_j(D)$.  If every user in $S$ has a bit $b$ in the $j$-th bit of their codeword, then $q_j(D) = b$.  Thus, if an algorithm answers $1$-way marginal queries on $D$ with non-trivial accuracy, its output can be used to obtain a combined codeword that satisfies the marking assumption.  By the tracing property of fingerprinting codes, we can use the combined codeword to identify one of the users in the database.  However, if we can identify one of the users from the answers, then the algorithm is not differentially private.

This argument can be formalized to show that if there is a fingerprinting code for $n$ users with codewords of length $d$, then the sample complexity of answering $1$-way marginals must be at least $n$.  The nearly-optimal construction of fingerprinting codes due to Tardos~\cite{Tardos08}, gives fingerprinting codes with codewords of length $d = \tilde{O}(n^2)$, which implies a lower bound of $n \geq \tilde{\Omega}(\sqrt{d})$ on the sample complexity required to answer $1$-way marginals queries.

\paragraph{Composition of Sample Complexity Lower Bounds.}
Suppose we want to prove a lower bound of $\tilde{\Omega}(k \sqrt{d})$ for answering $k$-way marginals up to accuracy $\pm .01$ (a special case of Theorem~\ref{thm:main2}).  Given our lower bound of $\tilde{\Omega}(\sqrt{d})$ for $1$-way marginals, and the known lower bound of $\Omega(k)$ for answering $k$-way marginals implicit in~\cite{DinurNi03,Roth10}, a natural approach is to somehow compose the two lower bounds to obtain a nearly-optimal lower bound of $\tilde{\Omega}(k \sqrt{d})$.  Our composition technique uses the idea of the $\Omega(k)$ lower bound from~\cite{DinurNi03,Roth10} to show that if we can answer $k$-way marginal queries on a large database $D$ with $n$ rows, then we can obtain the answers to the $1$-way marginal queries on a ``subdatabase'' of roughly $n/k$ rows.  Our lower bound for $1$-way marginals tell us that $n/k = \tilde{\Omega}(\sqrt{d})$, so we deduce $n = \tilde{\Omega}(k \sqrt{d})$.  

Actually, this reduction only gives accurate answers to \emph{most} of the $1$-way marginals on the subdatabase, so we need an extension of our lower bound for $1$-way marginals to differentially private algorithms that are allowed to answer a small fraction of the queries with arbitrarily large error. Proving a sample complexity lower bound for this problem requires a ``robust'' fingerprinting code whose tracing algorithm can trace codewords that have errors introduced into a small fraction of the bits.  We show how to construct such a robust fingerprinting code of length $d = \tilde{O}(n^2)$, and thus obtain the desired lower bound.  Fingerprinting codes satisfying a weaker notion of robustness were introduced by Boneh and Naor~\cite{BonehNa08,BonehKiMo10}.\footnote{In the fingerprinting codes of~\cite{BonehNa08,BonehKiMo10} the adversary is allowed to \emph{erase} a large fraction of the coordinates of the combined codeword, and must reveal which coordinates are erased.}

Theorems~\ref{thm:main1} and~\ref{thm:main2} are proven by using this composition technique repeatedly to combine our lower bound for $1$-way marginals with (variants of) several known lower bounds that capture the optimal dependence on $\log |\cQ|$ and $1/\alpha^2$.

\paragraph{Are Fingerprinting Codes Necessary to Prove Differential Privacy Lower Bounds?}
The connection between fingerprinting codes and differential privacy lower bounds extends to arbitrary families $\cQ$ of counting queries. We introduce the notion of a generalized fingerprinting code with respect to $\cQ$, where each codeword corresponds to a data universe element $x \in \univ$ and the bits of the codeword are given by $q(x)$ for each $q \in \queryset$, but is the same as an ordinary fingerprinting code otherwise. The existence of a generalized fingerprinting code with respect to $\cQ$, for $n$ users, implies a sample complexity lower bound of $n$ for privately releasing answers to $\cQ$. We also show a partial converse to the above result, which states that some sort of ``fingerprinting-code-like object'' is necessary to prove sample complexity lower bounds for answering counting queries under differential privacy.  This object has similar semantics to a generalized fingerprinting code, however the marking assumption required for tracing is slightly stronger and the probability that tracing succeeds can be significantly smaller than what is required by the standard definition of fingerprinting codes.  Our partial converse parallels the result of Dwork et al.~\cite{DworkNaReRoVa09} that shows computational hardness results for differential privacy imply a ``traitor-tracing-like object.''  We leave it as an open question to pin down precisely the relationship between fingerprinting codes and information-theoretic lower bounds in differential privacy (and also between traitor-tracing schemes and computational hardness results for differential privacy).

\subsection{Other Related Work}
\subsubsection{Previous Work}
We have mostly focused on the sample complexity as a function of the number of queries, the number of attributes $d$, and the accuracy parameter $\alpha$.  There have been several works focused on the sample complexity as a function of the specific family $\cQ$ of queries. For $(\eps, 0)$-differential privacy, Hardt and Talwar~\cite{HardtTa10} showed how to approximately characterize the sample complexity of a family $\cQ$ when the accuracy parameter $\alpha$ is sufficiently small.  Nikolov, Talwar, and Zhang~\cite{NikolovTaZh13} extended their results to give an approximate characterization for $(\eps,\delta)$-differential privacy and for the full range of accuracy parameters.  Specifically,~\cite{NikolovTaZh13} give an $(\eps, \delta)$-differentially private algorithm that answers any family of queries $\cQ$ on $\bits^{d}$ with error $\alpha$ using a number of samples that is optimal up to a factor of $\poly(d, \log |\cQ|)$ that is independent of $\alpha$.  Thus, their algorithm has sample complexity that depends optimally on $\alpha$.  However, their characterization may be loose by a factor of $\poly(d, \log |\cQ|)$. In fact, when $\alpha$ is a constant, the lower bound on the sample complexity given by their characterization is always $O(1)$, whereas their algorithm requires $\poly(d, \log |\cQ|)$ samples to give non-trivially accurate answers.  In contrast, our lower bounds are tight to within $\poly(\log d, \log \log |\cQ|, \log(1/\alpha))$ factors, and thus give meaningful lower bounds even when $\alpha$ is constant, but apply only to certain families of queries.

There have been attempts to prove optimal sample complexity lower bounds for $k$-way marginals.  In particular, when $k$ is a constant, Kasiviswanathan et al.~\cite{KasiviswanathanRuSmUl10} and De~\cite{De12} prove a lower bound of $\min\{|\cQ|^{1/2} / \alpha, 1/\alpha^2\}$ on the sample complexity. Note that when $\alpha$ is a constant, these lower bounds are $O(1)$.

There have also been attempts to explicitly and precisely determine the sample complexity of even simpler query families than $k$-way conjunctions, such as point functions and threshold functions~\cite{BeimelKaNi10,BeimelNiSt13a,BeimelNiSt13b,BunNiStVa15}.  These works show that these families can have sample complexity lower than $\tilde{O}(\sqrt{d} \log |\cQ| / \alpha^2)$.

In addition to the general computational hardness results referenced above, there are several results that show stronger hardness results for restricted types of efficient algorithms~\cite{UllmanVa11, GuptaHaRoUl11,DworkNaVa12}.

\subsubsection{Subsequent Work}
Subsequent to our work, Steinke and Ullman~\cite{SteinkeUl15a} refined our use of fingerprinting codes to prove a lower bound of $\Omega(\sqrt{d \log(1/\delta)}/\eps)$ on the number of samples required to release the mean of
each of the $d$ attributes under $(\eps, \delta)$-differential privacy when $\delta \ll 1/n$.  This lower bound is optimal up to constant factors, and improves on Theorem~\ref{thm:main0} by a factor of roughly $\sqrt{\log(1/\delta)} \cdot \log d$.  They also improve and simplify our analysis of robust fingerprinting codes.

Our fingerprinting code technique has also been used to prove lower bounds for other types of differentially private data analyses.  Namely, Dwork et al.~\cite{DworkTaThZh14} prove lower bounds for differentially private principal component analysis and Bassily, Smith, and Thakurta~\cite{BassilySmTh14} prove lower bounds for differentially private empirical risk minimization.  In order to establish lower bounds for privately releasing threshold functions, Bun et al.~\cite{BunNiStVa15} construct a fingerprinting-code-like object that yields a lower bound for the problem of releasing a value between the minimum and maximum of a dataset.

Dwork et al.~\cite{DworkSmStUlVa15} observe that the privacy attack implicit in our negative results is closely related to the influential attacks that were employed by Homer et al.~\cite{Homer+08} (and further studied in~\cite{SankararamonObJoHa09}) to violate privacy of public genetic datasets.  Using this connection, they show how to make Homer et al.'s attack robust to very general models of noise and how to make the attack work without detailed knowledge of the population the dataset represents.

A pair of works~\cite{HardtUl14, SteinkeUl15b} show that fingerprinting codes and the related traitor-tracing schemes imply both information-theoretic lower bounds and computational hardness results for the ``false discovery'' problem in adaptive data analysis.  Specifically, they show lower bounds for answering an online sequence of adaptively chosen counting queries where the database is a sample from some unknown distribution and the answers must be accurate with respect to that distribution.  These works~\cite{HardtUl14, SteinkeUl15b} effectively reverse a connection established in~\cite{DworkFeHaPiReRo15, BassilySmStUl15}, which used differentially private algorithms to obtain positive results for this problem.

Our technique for composing lower bounds in differential privacy has also found applications outside of privacy.  Specifically, Liberty et al.~\cite{LibertyMiThUl14} used this technique to prove nearly optimal lower bounds on the space required to ``sketch'' a database while approximately preserving answers to $k$-way marginal queries (called ``frequent itemset queries'' in their work).

%Finally, Bun et al.~\cite{BunNiStVa15} prove that sample complexity lower bounds for threshold functions imply a fingerprinting-code-like %object.  This object is analogous to the one we have shown is implied by sample complexity lower bounds for $1$-way marginals.

\section{Preliminaries}

\subsection{Differential Privacy}

We define a \emph{database} $\db \in \dbset$ to be an ordered tuple of $\rows$ rows $(\rowi{1},\dots,\rowi{\rows}) \in \univ$ chosen from a \emph{data universe} $\univ$.  We say that two databases $\db, \db' \in \univ^\rows$ are \emph{adjacent} if they differ only by a single row, and we denote this by $\db \sim \db'$. In particular, we can replace the $i$th row of a database $\db$ with some fixed ``junk'' element of $\univ$ to obtain another database $\db_{-i} \sim \db$. We emphasize that if $\db$ is a database of size $\rows$, then $\db_{-i}$ is also a database of size $\rows$.
\begin{definition}[Differential Privacy~\cite{DworkMcNiSm06}]\label{def:dp} Let $\san \from \univ^n \to \cR$ be a randomized algorithm (where $n$ is a varying parameter). $\san$ is \emph{$(\eps, \delta)$-differentially private} if for every two adjacent databases $\db \sim \db'$ and every subset $S \subseteq \cR$,
$$
\prob{\san(D) \in S} \leq e^{\eps} \Pr[\san(D') \in S] + \delta.
$$
\end{definition}

\begin{lemma} \label{lem:dp-exchange}
Let $\san \from \univ^n \to \cR$ be a randomized algorithm such that for every $\db \in \univ^n$, every $i, j \in [n]$, and every subset $S \subseteq \cR$,
\[\prob{\san(\db_{-i}) \in S} \leq e^{\eps} \prob{\san(\db_{-j}) \in S} + \delta.\]
Let $\bot$ denote the fixed junk element of $\univ$. Then $\san': \univ^{n-1} \to \cR$ defined by $\san'(x_1, \dots, x_{n-1}) = \san(x_1, \dots, x_{n-1}, \bot)$ is $(2\eps, (e^\eps + 1)\delta)$-differentially private.
\end{lemma}

\begin{proof}
Let $D = (x_1, \dots, x_{n-1})$ and $D' = (x_1, \dots, x_i', \dots, x_{n-1})$ be adjacent databases. Then for any $S \subseteq \cR$, we have
\begin{align*}
\prob{\san'(\db) \in S} &= \prob{\san(x_1, \dots, x_{n-1}, \bot) \in S} \\
&\le e^{\eps} \prob{\san(x_1, \dots, x_{i-1}, \bot, x_{i+1}, \dots, x_{n-1}, \bot) \in S} + \delta \\
&\le e^{2\eps} \prob{\san(x_1, \dots, x_{i-1}, x_i', x_{i+1}, \dots, x_{n-1}, \bot) \in S} + (e^\eps + 1)\delta \\
&= e^{2\eps} \prob{\san'(\db') \in S}  + (e^\eps + 1) \delta.
\end{align*}
\end{proof}

\subsection{Counting Queries and Accuracy}
In this paper we study algorithms that answer \emph{counting queries}.  A counting query on $\univ$ is defined by a predicate $\query \from \univ \to \bits$.  Abusing notation, we define the evaluation of the query $\query$ on a database $\db = (\rowi{1},\dots,\rowi{\rows}) \in \dbset$ to be its average value over the rows,
$$
\query(\db) = \frac{1}{\rows} \sum_{i=1}^{\rows} \query(\rowi{i}).
$$

\begin{definition}[Accuracy for Counting Queries]
Let $\queryset$ be a set of counting queries on $\univ$ and $\alpha, \beta \in [0,1]$ be parameters.  For a database $\db \in \dbset$, a sequence of answers $\answer = (\answerq{\query})_{\query \in \queryset} \in \R^{|\queryset|}$ is \emph{$(\alpha, \beta)$-accurate for $\queryset$} if
$
\left|\query(\db) - \answerq{\query}\right| \leq \alpha
$
for at least a $1-\beta$ fraction of queries $\query \in \queryset$.  

Let $\san \from \dbset \to \R^{|\queryset|}$ be a randomized algorithm.  $\san$ is \emph{$(\alpha,\beta)$-accurate for $\queryset$} if for every $\db \in \dbset$,
$$
\prob{\textrm{$\san(\db)$ is $(\alpha,\beta)$-accurate for $\queryset$}} \geq 2/3.
$$
When $\beta = 0$ we may simply write that $\answer$ or $\san$ is \emph{$\alpha$-accurate for $\queryset$}.
\end{definition}

In the definition of accuracy, we have assumed that $\san$ outputs a sequence of $|\queryset|$ real-valued answers, with $\answerq{\query}$ representing the answer to $\query$.  Since we are not concerned with the running time of the algorithm, this assumption is without loss of generality.\footnote{In certain settings, $\san$ is allowed to output a ``summary'' $z \in \cR$ for some range $\cR$.  In this case, we would also require that there exists an ``evaluator'' $\cE \from \cR \times \queryset \to \R$ that takes a summary and a query and returns an answer $\cE(z,\query) = \answerq{\query}$ that approximates $\query(\db)$.  The extra generality is used to allow $\san$ to run in less time than the number of queries it is answering.  However, since we do not bound the running time of $\san$ we can convert any such sanitizer to one that outputs a sequence of $|\queryset|$ real-valued answers simply by running the evaluator for every $\query \in \queryset$.}

An important example of a collection of counting queries is the set of \emph{$k$-way marginals}.  For all of our results it will be sufficient to consider only the set of \emph{monotone $k$-way marginals}.

\begin{definition}[Monotone $k$-way Marginals]
A \emph{(monotone) $k$-way marginal} $q_S$ over $\bits^d$ is specified by a subset $S \subseteq [d]$ of size $|S| \le k$. It takes the value $q_S(x) = 1$ if and only if $x_i = 1$ for every index $i \in S$. The collection of all (monotone) $k$-way marginals is denoted by $\kdconj{k}$.
\end{definition}

\subsection{Sample Complexity}
In this work we prove lower bounds on the sample complexity required to simultaneously achieve differential privacy and accuracy.
\begin{definition}[Sample Complexity]
Let $\queryset$ be a set of counting queries on $\univ$ and let $\alpha,\beta > 0$ be parameters, and let $\eps, \delta$ be functions of $n$.  We say that \emph{$(\queryset, \univ)$ has sample complexity $\wcsc$ for $(\alpha,\beta)$-accuracy and $(\eps, \delta)$-differential privacy} if $\wcsc$ is the least $\rows \in \N$ such that there exists an $(\eps, \delta)$-differentially private algorithm $\san \from \univ^n \to \R^{|\queryset|}$ that is $(\alpha,\beta)$-accurate for $\queryset$.
\end{definition}

We will focus on the case where $\eps = O(1)$ and $\delta = o(1/\rows)$.   This setting of the parameters is essentially the most-permissive for which $(\eps,\delta)$-differential privacy is still a meaningful privacy definition.  However, pinning down the exact dependence on $\eps$ and $\delta$ is still of interest. Regarding $\eps$, this can be done via the following standard lemma, which allows us to take $\eps = 1$ without loss of generality.
\begin{lemma} \label{lem:epsforfree}
For every set of counting queries $\queryset$, universe $\univ$, $\alpha, \beta \in [0,1], \eps \leq 1$.  $(\queryset, \univ)$ has sample complexity $\wcsc$ for $(\alpha,\beta)$-accuracy and $(1, o(1/\rows))$-differential privacy if and only if it has sample complexity $\Theta(\wcsc / \eps)$ for $(\alpha, \beta)$-accuracy and $(\eps, o(1/\rows))$-differential privacy.
\end{lemma}
\ifnum\stocsubmit=1
\fi

\subsection{Re-identifiable Distributions}
All of our eventual lower bounds will take the form of a ``re-identification'' attack, in which we possess data from a large number of individuals, and identify one such individual who was included in the database.  In this attack, we choose a distribution on databases and give an adversary 1) a database $D$ drawn from that distribution and 2) either $\san(D)$ or $\san(D_{-i})$ for some row $i$, where $\san$ is an alleged sanitizer.  The adversary's goal is to identify a row of $D$ that was given to the sanitizer.  We say that the distribution is re-identifiable if there is an adversary who can identify such a row with sufficiently high confidence whenever $\san$ outputs accurate answers.  If the adversary can do so, it means that there must be a pair of adjacent databases $D \sim D_{-i}$ such that the adversary can distinguish $\san(D)$ from $\san(D_{-i})$, which means $\san$ cannot be differentially private.

\begin{comment}
In more detail, we consider two experiments: in the first, the adversary is given a database $D$ and the output of the sanitizer on the whole database. As long as the adversary fails (i.e. outputs $\bot$) with low probability, then it will output the identity of some user $i^* \in [n]$ with probability $\Omega(1/n)$. In the second experiment, the adversary is given a database $D$ and the mechanism's output on $D$ with a random row removed. A successful adversary outputs the identity of the removed individual with very low probability $o(1/n)$. Thus the adversary's behavior on adjacent databases is noticeably different, giving the desired violation of differential privacy.
\end{comment}

\begin{definition}[Re-identifiable Distribution] \label{def:reidentifiabledist}
For a data universe $\univ$ and $\rows \in \N$, let $\cD$ be a distribution on $\rows$-row databases $\db \in \univ^\rows$.
Let $\queryset$ be a family of counting queries on $\univ$ and let $\gamma, \sec, \alpha, \beta \in [0,1]$ be parameters.  The distribution $\cD$ is \emph{$(\gamma,\sec)$-re-identifiable from $(\alpha,\beta)$-accurate answers to $\cQ$} if there exists a (possibly randomized) adversary $\adv \from \univ^\users \times \R^{|\cQ|} \to [n] \cup \{\bot\}$ such that for every randomized algorithm $\san \from \univ^\rows \to \R^{|\cQ|}$, the following both hold:
\begin{enumerate}
\item 
$
\Prob{\db \getsr \cD}{(\adv(\db, \san(\db)) = \bot) \land (\textrm{$\san(\db)$ is $(\alpha,\beta)$-accurate for $\queryset$)}} \le \gamma.
$
\item
For every $i \in [\rows]$,
$
\Prob{\db \getsr \cD}{\adv(\db,  \san(\db_{-i})) = i} \leq \sec.
$
\end{enumerate}
Here the probability is taken over the choice of $D$ and $i$ as well as the coins of $\san$ and $\adv$.  We allow $\cD$ and $\adv$ to share a common state.
\end{definition}

Note that, when row $i$ is not in the dataset, then it would be an error for $\adv$ to declare that row $i$ is in the dataset, and condition 2 requires that the probability of this error occurring is at most $\sec$.

The common state between $\cD$ and $\adv$ should be thought of as auxiliary information about the realization of $\db$ that may help $\adv$ identify a user $i$. Formally, we could model this shared state by having $\cD$ output an additional string $aux$ that is given to $\adv$ but not to $\san$.  However, we make the shared state implicit to reduce notational clutter. The need for this shared state will become apparent when we use fingerprinting codes to construct re-identifiable distributions; in the context of fingerprinting codes, the shared state represents auxiliary information about a codebook that helps the $\trace$ algorithm accuse a guilty pirate.

If $\san$ is an $(\alpha, \beta)$-accurate algorithm, then its output $\san(\db)$ will be $(\alpha, \beta)$-accurate with probability at least $2/3$.  Therefore, if $\gamma < 2/3$, we can conclude that $\prob{\adv(\db, \san(\db)) \in [\rows]} \geq 1- \gamma - 1/3 = \Omega(1)$.  In particular, there exists some $i^* \in [\rows]$ for which $\prob{\adv(\db, \san(\db)) = i^*} \geq \Omega(1/\rows)$.  However, if $\sec = o(1/\rows)$, then $\prob{\adv(\db, \san(\db_{-i^*})) = i^*} \leq \sec = o(1/\rows)$.  Thus, for this choice of $\gamma$ and $\sec$ we will obtain a contradiction to $(\eps, \delta)$-differential privacy of the post-processed algorithm $\adv(\db, \san(\cdot))$ for any $\eps = O(1)$ and $\delta = o(1/\rows)$.  Note that this conclusion holds even if $\cD$ and $\adv$ share a common state.
\ifnum\stocsubmit=1
\begin{comment}
\fi
We summarize this argument with the following lemma.
\begin{lemma} \label{lem:reidenttodp}
Let $\queryset$ be a family of counting queries on $\univ$, $\rows \in \N$ and $\sec \in [0,1]$.  Suppose there exists a distribution on $\rows$-row databases $\db \in \univ^n$ that is $(\gamma, \sec)$-re-identifiable from $(\alpha, \beta)$-accurate answers to $\queryset$.  Then there is no $(\eps, \delta)$-differentially private algorithm $\san \from \univ^\rows \to \R^{|\queryset|}$ that is $(\alpha, \beta)$-accurate for $\queryset$ for any $\eps, \delta$ such that
$
e^{-\eps} (1-\gamma-1/3) / n - \delta \geq \sec.
$

In particular, if there exists a distribution that is $(\gamma, o(1/n))$-re-identifiable from $(\alpha, \beta)$-accurate answers to $\queryset$ for $\gamma = 1/3$, then no algorithm $\san \from \univ^\rows \to \R^{|\queryset|}$ that is $(\alpha, \beta)$-accurate for $\queryset$ can satisfy $(O(1), o(1/\rows))$-differential privacy.
\end{lemma}
\ifnum\stocsubmit=1
\end{comment}
\fi

\section{Lower Bounds via Fingerprinting Codes} \label{sec:1way}

In this section we prove that there exists a simple family of $\cols$ queries that requires $n \geq \tilde{\Omega}(\sqrt{d})$ samples for both accuracy and privacy.  Specifically, we prove that for the family of $1$-way marginals on $\cols$ bits, sample complexity $\tilde{\Omega}(\sqrt{d})$ is required to produce differentially private answers that are accurate even just to within $\pm 1/3$.  In contrast, without a privacy guarantee, $\Theta(\log d)$ samples from the population are necessary and sufficient to ensure that the answers to these queries on the database and the population are approximately the same.  The best previous lower bound for $(\eps,\delta)$-differential privacy is also $O(\log d)$, which follows from the techniques of~\cite{DinurNi03, Roth10}.

In Section~\ref{sec:fpcs} we give the relevant background on fingerprinting codes and in Section~\ref{sec:lb1way} we prove our lower bounds for $1$-way marginals.

\subsection{Fingerprinting Codes} \label{sec:fpcs}
Fingerprinting codes were introduced by Boneh and Shaw~\cite{BonehSh98} to address the problem of watermarking digital content.  A \emph{fingerprinting code} is a pair of randomized algorithms $(\gen, \trace)$.  The code generator $\gen$ outputs a \emph{codebook} $\codebook \in \bits^{\users \times \len}$.  Each row $\codewordi{i}$ of $C$ is the \emph{codeword} of user $i$.  For a subset of users $S \subseteq [\users]$, we use $\codebookS{S} \in \bits^{|S| \times \len}$ to denote the set of codewords of users in $S$.  The parameter $d$ is called the \emph{length} of the fingerprinting code.

The security property of fingerprinting codes asserts that any codeword can be ``traced'' to a user $i \in [\users]$.  Moreover, we require that the fingerprinting code is ``fully-collusion-resilient''---even if any ``coalition'' of users $S \subseteq [\users]$ gets together and ``combines'' their codewords in any way that respects certain constraints known as a \emph{marking assumption}, then the combined codeword $c'$ can be traced to a user $i \in S$.  That is, there is a tracing algorithm $\trace$ that takes as inputs the codebook and combined codeword $c'$ and outputs either a user $i \in [\users]$ or $\bot$, and we require that if $\pirateword$ satisfies the constraints, then $\trace(\codebook, \pirateword) \in S$ with high probability.  Moreover, $\trace$ should accuse an innocent user, i.e. $\trace(\codebook, \pirateword) \in [\users] \setminus S$, with very low probability.  Analogous to the definition of re-identifiable distributions (Definition~\ref{def:reidentifiabledist}), we allow $\gen$ and $\trace$ to share a common state.\footnote{As in Definition~\ref{def:reidentifiabledist}, we could model this by having $\gen$ output an additional string $aux$ that is given to $\trace$.  However, we make the shared state implicit to reduce notational clutter.}  When designing fingerprinting codes, one tries to make the marking assumption on the combined codeword as weak as possible.

The basic marking assumption is that each bit of the combined word $\pirateword$ must match the corresponding bit for some user in $S$.  Formally, for a codebook $\codebook \in \bits^{\users \times \len}$, and a coalition $S \subseteq [\users]$, we define the set of \emph{feasible codewords for $\codebookS{S}$} to be
$$
F(\codebookS{S}) = \set{\pirateword \in \bits^{\len} \mid \forall j \in [\len], \exists i \in S, \piratewordj{j} = \codewordij{i}{j}}.
$$
Observe that the combined codeword is only constrained on coordinates $j$ where all users in $S$ agree on the $j$-th bit.

We are now ready to formally define a fingerprinting code.
\begin{definition}[Fingerprinting Codes] \label{def:fpc}
For any $\users, \len \in \N$, $\sec \in (0,1]$, a pair of algorithms $(\gen, \trace)$ is an \emph{$(\users,\len)$-fingerprinting code with security $\sec$} if $\gen$ outputs a codebook $\codebook \in \bits^{\users \times \len}$ and for every (possibly randomized) adversary $\fpadv$, and every coalition $S \subseteq [\users]$, if we set $\pirateword \getsr \fpadv(\codebookS{S})$, then
\begin{enumerate}
\item 
$
\prob{\pirateword \in F(\codebookS{S}) \land \trace(\codebook, \pirateword) = \bot} \leq \sec,
$
\item 
$
\prob{\trace(\codebook, \pirateword) \in [\users] \setminus S} \leq \sec,
$
\end{enumerate}
where the probability is taken over the coins of $\gen, \trace$, and $\fpadv$.  The algorithms $\gen$ and $\trace$ may share a common state.
\end{definition}

We remark that our proof of Theorem~\ref{thm:fpctolb}, showing how to construct re-identifiable distributions from a fingerprinting codes, will only require collusion resilience against coalitions $S$ of size $|S| \ge n-1$. Our choice to state Definition~\ref{def:fpc} using resilience against arbitrary coalitions is more consistent with the literature on fingerprinting codes.

Tardos~\cite{Tardos08} constructed a family of fingerprinting codes with a nearly optimal number of users $\users$ for a given length $d$.
\begin{theorem}[\cite{Tardos08}] \label{thm:fpc}
For every $\len \in \N$, and $\sec \in [0,1]$, there exists an $(\users, \len)$-fingerprinting code with security $\sec$ for
$$
\users = \users(\len, \sec) = \tilde{\Omega}(\sqrt{\len / \log(1/\sec)}).
$$
\end{theorem}

As we will see in the next subsection, fingerprinting codes satisfying Definition~\ref{def:fpc} will imply lower bounds on the sample complexity for releasing $1$-way marginals with $(\alpha, 0)$-accuracy (accuracy for every query).  In order to prove sample-complexity lower bounds for $(\alpha, \beta)$-accuracy with $\beta > 0$, we will need fingerprinting codes satisfying a stronger security property.  Specifically, we will expand the feasible set $F(\codebookS{S})$ to include all codewords that satisfy most feasibility constraints, and require that even codewords in this expanded set can be traced.  Formally, for any $\rob \in [0,1]$, we define
$$
F_{\rob}(\codebookS{S}) = \set{\pirateword \in \bits^{\len} \mid \Prob{j \getsr [d]}{\exists i \in S, \piratewordj{j} = \codewordij{i}{j}} \ge 1 - \rob}.
$$
Observe that $F_{0}(\codebookS{S}) = F(\codebookS{S})$.

\begin{definition}[Error-Robust Fingerprinting Codes] \label{def:rfpc}
For any $\users, \len \in \N$, $\sec, \rob \in [0,1]$, a pair of algorithms $(\gen, \trace)$ is an \emph{$(\users,\len)$-fingerprinting code with security $\sec$ robust to a $\rob$ fraction of errors} if $\gen$ outputs a codebook $\codebook \in \bits^{\users \times \len}$ and for every (possibly randomized) adversary $\fpadv$, and every coalition $S \subseteq [\users]$, if we set $\pirateword \getsr \fpadv(\codebookS{S})$, then
\begin{enumerate}
\item 
$
\prob{\pirateword \in F_{\rob}(\codebookS{S}) \land \trace(\codebook, \pirateword) = \bot} \leq \sec,
$
\item 
$
\prob{\trace(\codebook, \pirateword) \in [\users] \setminus S} \leq \sec,
$
\end{enumerate}
where the probability is taken over the coins of $\gen, \trace$, and $\fpadv$.  The algorithms $\gen$ and $\trace$ may share a common state.
\end{definition}

In Section~\ref{chap:rfpc} we show how to construct error-robust fingerprinting codes with a nearly-optimal number of users that are tolerant to a constant fraction of errors.
\begin{theorem} \label{thm:rfpc}
For every $\len \in \N$, and $\sec \in (0,1]$, there exists an $(\users, \len)$-fingerprinting code with security $\sec$ robust to a $1/75$ fraction of errors for
$$
\users = \users(\len, \sec) = \tilde{\Omega}(\sqrt{\len / \log(1/\sec)}).
$$
\end{theorem}

Boneh and Naor~\cite{BonehNa08} introduced a different notion of fingerprinting codes robust to adversarial ``erasures''.  In their definition, the adversary is allowed to output a string in $\{0,1,?\}^d$, and in order to trace they require that the fraction of ?~symbols is bounded away from $1$ and that any non-? symbols respect the basic feasibility constraint.  For this definition, constructions with nearly-optimal length $d = \tilde{O}(n^2)$, robust to a $1 - o(1)$ fraction of erasures are known~\cite{BonehKiMo10}.  In contrast, our codes are robust to adversarial ``errors.''  Robustness to a $\rob$ fraction of errors can be seen to imply robustness to nearly a $2\rob$ fraction of erasures but the converse is false.  Thus for corresponding levels of robustness our definition is strictly more stringent.  Unfortunately we don't currently know how to design a code tolerant to a $1/2-o(1)$ fraction of errors, so our Theorem~\ref{thm:rfpc} does not subsume prior results on robust fingerprinting codes.
\subsection{Lower Bounds for $1$-Way Marginals} \label{sec:lb1way}

We are now ready to state and prove the main result of this section, namely that there is a distribution on databases $D \in (\bits^{\cols})^\rows$, for $\rows = \tilde{\Omega}(\sqrt{d})$, that is re-identifiable from accurate answers to $1$-way marginals.

\begin{theorem} \label{thm:fpctolb}
For every $\users, \len \in \N$, and $\sec \in [0,1]$ if there exists an $(\users,\len)$-fingerprinting code with security $\sec$, robust to a $\rob$ fraction of errors, then there exists a distribution on $\users$-row databases $D \in (\bits^\cols)^{\users}$ that is $(\sec, \sec)$-re-identifiable from $(1/3, \beta)$-accurate answers to $\kdconj{1}$.

In particular, if $\sec = o(1/\users)$, then  \ifnum\stocsubmit=1
\begin{comment}
\fi
by Lemma~\ref{lem:reidenttodp} 
\ifnum\stocsubmit=1
\end{comment}
\fi
there is no algorithm $\san \from (\bits^{\cols})^\users \to \R^{|\kdconj{1}|}$ that is $(O(1), o(1/\users))$-differentially private and $(1/3, \rob)$-accurate for $\kdconj{1}$.
\end{theorem}

By combining Theorem~\ref{thm:fpctolb} with Theorem~\ref{thm:fpc} we obtain a sample complexity lower bound for $1$-way marginals, and thereby establish Theorem~\ref{thm:main0} in the introduction.
\begin{corollary} \label{cor:lb1way}
For every $\cols \in \N$, the family of $1$-way marginals on $\bits^\cols$ has sample complexity at least 
$
\tilde{\Omega}(\sqrt{\cols})
$
for $(1/3,1/75)$-accuracy and $(O(1), o(1/\rows))$-differential privacy.
\end{corollary}

\begin{proof}[Proof of Theorem~\ref{thm:fpctolb}]
Let $(\gen, \trace)$ be the promised fingerprinting code.
\begin{comment}  
\begin{figure}[ht]
\begin{framed}
\begin{algorithmic}
\STATE{$\adv(\db, i, \san(\db))$:}
\INDSTATE[1]{Let $\answer = \san(\db) \in [0,1]^{|\kdconj{1}|}$ and let $\overline{\answer}$ be $\answer$ with each coordinate rounded to $\bits$.}
\INDSTATE[1]{Let $i' \getsr \trace(\db, \overline{\answer})$.}
\INDSTATE[1]{If $i' = i$, output $1$.  Otherwise, output $0$.}
\end{algorithmic}
\end{framed}
\vspace{-6mm}
\caption{The privacy adversary $\adv$.} \label{fig:reidfrom1way}
\end{figure}
\end{comment}
We define the re-identifiable distribution $\cD$ to simply be the output distribution of the code generator, $\gen$.  And we define the privacy adversary $\adv$ to take the answers $\answer = \san(\db) \in [0,1]^{|\kdconj{1}|}$, obtain $\overline{\answer} \in \bits^{|\kdconj{1}|}$ by rounding each entry of $\answer$ to $\{0,1\}$, run the tracing algorithm $\trace$ on the rounded answers $\overline{\answer}$, and return its output.  The shared state of $\cD$ and $\adv$ will be the shared state of $\gen$ and $\trace$.

Now we will verify that $\cD$ is $(\sec,\sec)$-re-identifiable.  First, suppose that $\san(\db)$ outputs answers $\answer = (\answerq{\queryj{j}})_{j \in [\cols]}$ that are $(1/3,\rob)$-accurate for $1$-way marginals.  That is, there is a set $G \subseteq [\cols]$ such that $|G| \geq (1-\rob)d$ and for every $j \in G$, the answer $\answerq{\queryj{j}}$ estimates the fraction of rows having a $1$ in column $j$ to within $1/3$.  Let $\overline{\answer}_{\queryj{j}}$ be $\answerq{\queryj{j}}$ rounded to the nearest value in $\bits$.  Let $j$ be a column in $G$.  If column $j$ has all $1$'s, then $\answerq{\queryj{j}} \geq 2/3$, and $\overline{\answer}_{\queryj{j}} = 1$.  Similarly, if column $j$ has all $0$'s,  then $\answerq{\queryj{j}} \leq 1/3$, and $\overline{\answer}_{\queryj{j}} = 0$.  Therefore, we have
\begin{equation} \label{eq:3-1}
\textrm{$\answer$ is $(1/3,\rob)$-accurate} \Longrightarrow \overline{\answer} \in F_{\rob}(\db).
\end{equation}
By security of the fingerprinting code (Definition~\ref{def:rfpc}), we have
\begin{equation} \label{eq:3-2}
\prob{\overline{\answer} \in F_{\rob}(\db) \land \trace(\db, \overline{\answer}) = \bot} \leq \sec.
\end{equation}
Combining~\eqref{eq:3-1} and~\eqref{eq:3-2} implies that
$$
\prob{\textrm{$\san(\db)$ is $(1/3,\rob)$-accurate} \land \trace(\db, \overline{\answer}) = \bot } \leq \sec.
$$
But the event $\trace(\db, \overline{\answer}) = \bot$ is exactly the same as $\adv(\db, \san(\db)) = \bot$, and thus we have established the first condition necessary for $\cD$ to be $(\sec, \sec)$-re-identifiable.

The second condition for re-identifiability follows directly from the soundness of the fingerprinting code, which asserts that for every adversary $\fpadv$, in particular for $\san$, it holds that
$$
\prob{\trace(\db, \fpadv(\db_{-i})) = i} \leq \sec.
$$
This completes the proof.
\end{proof}

\begin{remark}
Corollary~\ref{cor:lb1way} implies a lower bound of $\tilde{\Omega}(\sqrt{d})$ for any family $\queryset$ on a data universe $\univ$ in which we can ``embed'' the $1$-way marginals on $\bits^{\cols}$ in the sense that there exists $\query_1,\dots,\query_d \in \queryset$ such that for every string $x \in \bits^d$ there is an $x' \in \bits^d$ such that $(\query_1(x'),\dots,\query_d(x')) = x$.  (The maximum such $d$ is actually the VC dimension of $\univ$ when we view each element $x \in \univ$ as defining a mapping $\query \mapsto \query(x)$.  See Definition~\ref{def:vcdim}.)
\end{remark}

Our proof technique does not directly yield a lower bound with any meaningful dependence on the accuracy $\alpha$. Since the privacy adversary $\adv$ simply runs the tracing algorithm on the rounded answers it is given, it is not able to leverage subconstant accuracy to gain an advantage in re-identification. However, Lemma~\ref{lem:alphaforfree} lets us generically translate our lower bound for constant accuracy into a lower bound depending linearly on $1/\alpha$. For 1-way marginals, we get an essentially tight sample complexity lower bound of $\tilde{\Omega}(\sqrt{d}/\alpha)$ for $(\alpha, \beta)$-accuracy.

\begin{corollary} \label{cor:lb1wayalpha}
For every $\cols \in \N$, the family of $1$-way marginals on $\bits^\cols$ has sample complexity at least 
$
\tilde{\Omega}(\sqrt{\cols}/\alpha)
$
for $(\alpha,1/75)$-accuracy and $(O(1), o(1/\rows))$-differential privacy.
\end{corollary}

\subsubsection{Minimax Lower Bounds for Statistical Inference}

Using the additional structure of Tardos' fingerprinting code, and our robust fingerprinting codes, we can prove minimax lower bounds for an ``inference version'' of the problem computing the $1$-way marginals of a product distribution.

For any $d \in \N$, and any \emph{marginals} $p = (p_1,\dots,p_d) \in [0,1]^d,$ let $\cD_{p}$ denote the product distribution over strings $x \in \bits^d$ where each coordinate $x_i$ is an independent draw from a Bernoulli random variable with mean $p_i$ (i.e.~$x_i$ is set to $1$ with probability $p_i$ and set to $0$ otherwise).  We use $\cD_{p}^{\otimes n}$ to denote $n$ independent draws from $\cD_{p}.$  We say that a vector $q \in [0,1]^d$ is $(\alpha, \beta)$-accurate for $p$ if
$$
\Prob{i \getsr [d]}{| q_i - p_i | \leq \alpha} \geq 1 - \beta.
$$
We can now formally define the problem of inferring the marginals $p$ as follows.
\begin{definition}
Let $\alpha, \beta \in [0,1]$ be parameters.  An algorithm $\san \from (\bits^d)^n \to \R^d$ \emph{$(\alpha, \beta)$-accurately infers the marginals of a product distribution} if for every vector of marginals $p \in [0,1]^d$,
$$
\Prob{\db \getsr \cD_{p}^{\otimes n},\, \textrm{$\san$'s coins}}{\textrm{$\san(\db)$ is $(\alpha, \beta)$-accurate for $p$}} \geq 2/3.
$$
\end{definition}

Our lower bound can thus be stated as follows,
\begin{theorem}
Suppose there is a function $n = n(d)$ such that for every $d \in \N$, there exists an algorithm $\san \from (\bits^d)^n \to \R^d$ that satisfies $(O(1), o(1/n))$-differential privacy and $(1/3, 1/75)$-accurately infers the marginals of a product distribution.  Then $n = \tilde{\Omega}(\sqrt{d}).$
\end{theorem}

\begin{proof}[Proof Sketch]
The proof has the same general structure that we used to prove Theorem~\ref{thm:fpctolb}. Here, we describe additional observations about the structure of the fingerprinting codes used in that proof (see Section~\ref{chap:rfpc} for a description of Tardos' fingerprinting code) that allow it to carry over to the inference version of computing $1$-way marginals.

First, in Tardos' (non-robust) fingerprinting code, the codebook $\db$ is chosen by first sampling marginals $p \in [0,1]^d$ from an appropriate distribution and then sampling $\db$ from $\cD_{p}^{\otimes n}.$  The robust fingerprinting codes we construct in Section~\ref{chap:rfpc} also have this property.\footnote{To generate a codebook $\db'$ for our robust fingerprinting code, we sample a codebook $\db$ from Tardos' fingerprinting code and then insert additional columns of all $1$'s or all $0$'s to $\db$ in random locations.  Equivalently, we can obtain a codebook $\db'$ by appending $1$'s and $0$'s in random locations of $p$ to obtain a vector $p'$ and then sampling $\db'$ from $\cD_{p'}^{\otimes n}.$}  Thus the instances used to prove Theorem~\ref{thm:fpctolb} indeed consist of independent samples from a product distribution, which is what the inference problem assumes.

Next, recall that the proof of Theorem~\ref{thm:fpctolb} shows that any string that is $(\alpha, \beta)$-accurate for the $1$-way marginals of $\db$ can be traced successfully.  It is moreover the case that any string that is $(\alpha, \beta)$-accurate for the marginals $p$ can also be traced successfully.  This is because the rows of $\db$ are sampled independently from $\cD_{p}$, so accuracy for the $1$-way marginals of $\db$ and accuracy for $p$ coincide with high probability, at least when $n = \omega(\log d)$:

\begin{claim}
Let $p \in [0, 1]^d$ and let $\db \getsr \cD_{p}^{\otimes n}$. Let $a \in [0, 1]^d$ denote the exact $1$-way marginals of $\db$. Then for every $\alpha, \eta > 0$, and $n = \Omega(\log (d/\eta)/\alpha^2)$, we have $\|a - p\|_\infty \le \alpha$ with probability at least $1 - \eta$ over the choice of $\db$.
\end{claim}
We remark that Steinke and Ullman~\cite{SteinkeUl15a} showed that accuracy with respect to the marginals $p$ actually suffices to trace regardless of the value of $n$.

These two observations suffice to show that, when $n$ is too small, a differentially private algorithm cannot be accurate for $p$ with high probability over the choices of both $p$ and $\db$.  Thus, for every differentially private algorithm, there exists some $p$ such that the algorithm is not accurate with high probability over the choice of $\db$, which means that the algorithm does not accurately infer the marginals of an arbitrary product distribution.
\end{proof}

\subsection{Fingerprinting Codes for General Query Families}

In this section, we generalize the connection between fingerprinting codes and sample complexity lower bounds for arbitrary sets of queries.  We show that a generalized fingerprinting code with respect to any family of counting queries $\cQ$ yields a sample complexity lower bound for $\cQ$, which is analogous to our lower bound for $1$-way marginals (Theorem~\ref{thm:fpctolb}). We then argue that some type of fingerprinting code is necessary to prove any sample complexity lower bound by exhibiting a tight connection between such lower bounds and a weak variant of our generalized fingerprinting codes.

We begin by defining our generalization of fingerprinting codes. Fix a finite data universe $\cX$ and a set of counting queries $\cQ$ over $\cX$. A generalized fingerprinting code with respect to the family $\cQ$ consists of a pair of randomized algorithms $(\gen, \trace)$.  The code generation algorithm $\gen$ produces a codebook $C \in \cX^n$. Each row $c_i$ of $C$ is the codeword corresponding to user $i$. A coalition $S \subseteq [n]$ of pirates receives the subset $C_S = \{c_i : i\in S\}$ of codewords, and produces an \emph{answer vector} $a \in [0, 1]^{|\cQ|}$. We replace the traditional marking condition on the pirates with the generalized constraint that they output a \emph{feasible answer vector}. A natural way to define feasibility for answer vectors is to require a condition similar to $(\alpha, \beta)$-accuracy, i.e.~an answer vector $a$ is feasible if $|a_q - q(C_S)| \le \alpha$ for all but a $\beta$ fraction of queries $q \in \cQ$. We thus define a generalized set of feasible answer vectors by
$$
F_{\alpha, \rob}(\codebookS{S}) = \set{a \in [0, 1]^{|\cQ|} \mid \Prob{q \getsr \cQ}{|a_q - q(C_S)| \le \alpha} \ge 1 - \rob}.
$$
When $\alpha = 1-1/n$, the generalized set of feasible answer vectors captures the traditional marking assumption by rounding each entry of a feasible answer vector to $0$ or $1$.\footnote{An equivalent way to view a codebook is as a set of $n$ codewords $C \in (\{0, 1\}^{|\cQ|})^{n}$, where each user's codeword is $c_i = (q(x))_{q \in \cQ}$ for some $x \in \cX$. Notice that the case where $\cQ$ is the class of $1$-way marginals places no constraints on the structure of a codeword, i.e.~a codeword can be any binary string. With this viewpoint, the goal of the pirates is to output an answer vector $a \in [0, 1]^{|\cQ|}$ with $|a_q - \frac{1}{|S|} \sum_{i \in S}(c_i)_q| \le \alpha$ for all but a $\beta$ fraction of the queries $q \in \cQ$.}

%Note that when $\alpha = 1-1/n$, this definition captures the traditional marking assumption. That is,
%$$
%F_{1-1/n, \rob}(\codebookS{S}) = \set{a \in \{0, 1\}^{|\cQ|} \mid \Prob{q \getsr \cQ}{\exists i \in S, a_q = q(c_i)} \ge 1 - \rob}.
%$$
%In particular, when $\cQ$ is the query class of $1$-way marginals, we have $F_{1-1/n, \rob} = F_\rob$ from Section \ref{sec:fpcs}.

\begin{definition} \label{def:general-fpc}
A pair of algorithms $(\gen, \trace)$ is an $(n, \cQ)$-\emph{fingerprinting code} for $(\alpha, \beta)$-accuracy with security $(\gamma, \sec)$ if $\gen$ outputs a codebook $C \in \cX^n$ and for every (possibly randomized) adversary $\fpadv$, and every coalition $S \subseteq [n]$ with $|S| \ge n-1$, if we set $a \getsr \fpadv(C_S)$, then
\begin{enumerate}
\item 
$
\prob{a \in F_{\alpha, \rob}(\codebookS{S}) \land \trace(C, a) = \bot} \leq \gamma,
$
\item 
$
\prob{\trace(C, a) \in [\users] \setminus S} \leq \sec,
$
\end{enumerate}
where the probability is taken over the coins of $\gen, \trace$, and $\fpadv$.  The algorithms $\gen$ and $\trace$ may share a common state.
\end{definition}

The security properties of Definition \ref{def:general-fpc} differ from those of an ordinary fingerprinting code in two ways so as to enable a clean statement of a composition theorem for generalized fingerprinting codes (Theorem \ref{thm:composition-fpc}). First, we use two separate security parameters $\gamma, \sec$ for the different types of tracing errors, as in the definition of re-identifiable distributions. Second, security only needs to hold for coalitions of size $n-1$ or $n$.  However, this condition implies security for coalitions of arbitrary size with an increased false accusation probability of $n \sec$.

As in Theorem \ref{thm:fpctolb}, the existence of a generalized $(n, \cQ)$-fingerprinting code implies a sample complexity lower bound of $n$ for privately releasing answers to $\cQ$, with essentially the same proof.

\begin{theorem} \label{thm:general-fpctolb}
For every $n \in \N$ and $\gamma, \sec \in [0,1)$, if there exists an $(n, \cQ)$-fingerprinting code for $(\alpha,\beta)$-accuracy with security $(\gamma, \sec)$, then there exists a distribution on $n$-row databases $D \in \cX^n$ that is $(\gamma, \sec)$-re-identifiable from $(\alpha, \beta)$-accurate answers to $\cQ$.

In particular, if $\gamma \le 1/3$ and $\sec = o(1/\users)$, then there is no algorithm $\san \from \cX^n \to [0, 1]^{|\cQ|}$ that is $(O(1), o(1/\users))$-differentially private and $(\alpha, \rob)$-accurate for $\cQ$.
\end{theorem}

We now turn to investigate whether a converse to Theorem \ref{thm:general-fpctolb} holds. We show that a sample complexity lower bound for a family of queries $\cQ$ is essentially equivalent to the existence of a weak type of fingerprinting code, where the tracing procedure depends on the family $\cQ$ and the tracing error probabilities satisfy certain affine constraint. It remains an interesting open question to determine the precise relationship between privacy lower bounds and our notion of generalized fingerprinting codes.

\begin{definition} \label{def:weak-fpc}
A pair of algorithms $(\gen, \trace)$ is an $(n, \cQ)$-\emph{weak fingerprinting code} for $(\alpha, \beta)$-accuracy with security $(\eps, \delta)$ if $\gen$ outputs a codebook $C \in \cX^n$ and for every (possibly randomized) adversary $\fpadv$ that outputs a feasible answer vector with probability $2/3$, and every coalition $S \subseteq [n]$ with $|S| \ge n-1$, if we set $a \getsr \fpadv(C_S)$, then
\[\Pr[\trace(C, a) \ne \bot] > e^{\eps} n \cdot \Pr[\trace(C, a) \in [n] \setminus S] + \delta,\]
where the probabilities are taken over the coins of $\gen$, $\trace$, and $\fpadv$. The algorithms $\gen$ and $\trace$ may share a common state.
\end{definition}

That is, we require the false accusation probability $\Pr[\trace(C, a) \in [n] \setminus S]$ to be much smaller than the total probability of accusing any user. Note that a tracing algorithm that accuses a random user with probability $p$ will falsely accuse a user with probability $p/n$ when $|S| = n-1$; however, this does not satisfy Definition \ref{def:weak-fpc} because we require the gap between the two probabilities to be at least a factor of $e^\eps n$.

Observe that taking $\sec < (1-\delta)/2e^\eps n$ in Definition \ref{def:general-fpc} yields an $(n, \cQ)$-weak fingerprinting code with security $(\eps, \delta)$. However, Definition \ref{def:weak-fpc} is weaker than Definition \ref{def:general-fpc} in a few important ways. First, security only holds against pirates with a failure probability of at most $1/3$. Second, while Definition \ref{def:general-fpc} requires completeness error $\Pr[\trace(C, a) = \bot] < \sec$, a weak fingerprinting code allows $\Pr[\trace(C, a) = \bot] = 1 - o(1)$ as long as $\Pr[\trace(C, a) \in [n]\setminus S]$ is sufficiently small. 

The following theorem shows that the existence of an $(n, \cQ)$-weak fingerprinting code is essentially equivalent to a sample complexity lower bound of $n$ against $\cQ$.

\begin{theorem} \label{thm:weak-fpc}
For every $n \in \N$, if there exists an $(n, \cQ)$-weak fingerprinting code for $(\alpha, \beta)$-accuracy with security $(\eps, \delta)$, then there exists a distribution on $n$-row databases $D \in \cX^n$ such that no $(\eps/2, \delta/(2e^{\eps/2}n))$-differentially private algorithm $\cA : \cX^n \to \R^{|\cQ|}$ outputs $(\alpha, \beta)$-accurate answers to $\cQ$.

Conversely, let $\eps \le 3$ and suppose there is no $(\eps, \delta)$-differentially private $\cA: \cX^n \to \R^{|\cQ|}$ that gives $(\alpha, \beta)$-accurate answers to $\cQ$ with probability at least $1/2$. Then there exists an $(m = \lceil n / \eps \rceil, \cQ)$-weak fingerprinting code for $(\alpha - \alpha', \beta)$-accuracy with security $(\eps/6, \delta/(e^{\eps/3} + e^{5\eps/6}))$, for $\alpha' = \tilde{O}(\sqrt{\eps \vcdim(\cQ)/ n})$.
\end{theorem}

\begin{proof}

The forward direction follows the ideas of Lemma \ref{lem:reidenttodp} and Theorem \ref{thm:fpctolb}. Suppose for the sake of contradiction that there exists an $(\eps', \delta')$-differentially private $\cA: \cX^n \to \R^{|\cQ|}$ that is $(\alpha, \beta)$-accurate for $\cQ$. Define a pirate strategy $\fpadv$ for coalitions of size $|S| \ge n-1$ by running $\cA$ on its input $C_S$ (possibly padded to size $n$ by a junk row). Since $\cA$ is $(\alpha, \beta)$-accurate, with probability at least $2/3$ it produces an answer vector $a$ such that $|a - q(C_S)| \le \alpha$ for all but a $\beta$ fraction of $q \in \cQ$. Hence, $\fpadv$ outputs a feasible answer vector with probability $2/3$. Define
\[p = \Pr_{\substack{C \getsr \gen \\ \text{coins}(\fpadv), \text{coins}(\trace)}}[\trace(C, \fpadv(C)) \ne \bot].\]
Then there exists an $i^*$ such that $\Pr[\trace(C, \fpadv(C)) = i^*] \ge p/n$. By differential privacy,
\[\Pr[\trace(C, \fpadv(C_{-i^*})) = i^*] \ge e^{-\eps'} \cdot \left(\frac{p}{n} - \delta'\right).\]
 On the other hand, by the security of the weak fingerprinting code and differential privacy,
\begin{align*}
e^\eps \cdot n \cdot \Pr[\trace(C, \fpadv(C_{-i^*}) = i^*] &< \Pr[\trace(C, \fpadv(C_{-i^*}) \ne \bot] - \delta\\
&\le e^{\eps'} p + \delta' - \delta.
\end{align*}
This yields a contradiction whenever $\eps' \le \eps/2$ and $\delta' \le \delta / (1 + e^{\eps/2} n)$.

We now show the converse direction, i.e. that the high sample complexity of $(\cQ, \univ)$ implies the existence of a weak fingerprinting code. We begin with a technical lemma which shows that the high sample complexity of $\cQ$ also rules out mechanisms that satisfy only a one-sided constraint on the probability of any event under the replacement of one row:

\begin{lemma} \label{lem:one-sided}
Let $\eps \le 1/2$. Let $\cA$ be an $(\alpha, \beta)$-accurate algorithm for $\cQ$ on databases $\db \in \univ^m$. Suppose we have that for all databases $\db \in \univ^m$, all $i \in [m]$, and all measurable $T \subseteq \text{Range}(\cA)$ that
\[\Prob{\substack{j \getsr [m] \\ \text{coins}(\cA)}}{\cA(\db_{-j}) \in T} \le e^{\eps}\Prob{\text{coins}(\cA)}{\cA(\db_{-i}) \in T} + \delta.\]
Let $d = \vcdim(\cQ)$ be the VC-dimension of $\cQ$ and let
\[\alpha' = \left(\frac{8}{m} \cdot \left(\ln 24 + d \cdot \ln\left(\frac{2em}{d}\right) \right)\right)^{1/2} + \frac{\eps}{m}.\]
Then there exists a $(6\eps, (e^{2\eps} + e^{5\eps})\delta)$-differentially private algorithm $\cB$ on databases of size $n = \lceil m / \eps \rceil$ that gives $(\alpha + \alpha', \beta)$-accurate answers to $\cQ$ on any database $\db' \in \univ^n$ with probability at least $1/2$ .
\end{lemma}

\begin{proof}
On input a database $\db' \in \univ^{n}$, consider the algorithm $\cB'$ that samples a random subset $\db$ consisting of $m$ rows from $\db'$ (without replacement) and returns $\cA(\db)$. Then by our hypothesis on $\cA$, for every $i \in [n]$ and every measurable $T \subseteq \text{Range}(\cB) = \text{Range}(\cA)$ we have
\begin{align*}
\Prob{\substack{j \getsr [n] \\ \text{coins}(\cB')}}{\cB'(\db'_{-j}) \in T}
\le e^{\eps}\Prob{\text{coins}(\cB')}{\cB'(\db'_{-i}) \in T} + \delta. \\
\end{align*}
On the other hand, a ``secrecy-of-the-sample'' argument \cite{KasiviswanathanLeNiRaSm07} enables us to obtain the reverse inequality. For a row $k \in [n]$, consider the following two experiments:
\begin{description}
\item{Experiment 1:} Sample a random subset $\db$ of $m$ rows from $\db'_{-k}$.
\item{Experiment 2:} Sample $j \getsr [n]$, and then sample a random subset $\db$ of $m$ rows from $\db'_{-j}$.
\end{description}
Any database $\db$ sampleable under Experiment 1 appears with probability $1/{n \choose m}$, but appears with probability at least
\[\frac{n - m}{n} \cdot \frac{1}{{n \choose m}} \ge (1 - \eps) \cdot \frac{1}{{n  \choose m}}\]
 under Experiment 2. Therefore,
\[\Prob{\substack{j \getsr [n] \\ \text{coins}(\cB)}}{\cB(\db'_{-j}) \in T} \ge e^{-2\eps}\Prob{\text{coins}(\cB)}{\cB(\db'_{-k}) \in T}.\]
 Combining the two inequalities shows that for every database $\db' \in \univ^n$ and every $i, k \in [n]$, 
\[\Prob{\text{coins}(\cB')}{\cB'(\db'_{-k}) \in T} \le e^{3\eps}\Prob{\text{coins}(\cB')}{\cB'(\db'_{-i}) \in T} + e^{2\eps}\delta.\]
By Lemma \ref{lem:dp-exchange}, the algorithm $\cB(\db'_1, \dots, \db'_{n-1}) = \cB'(\db'_1, \dots, \db'_{n-1}, \bot)$ is $(6\eps, (e^{2\eps} + e^{5\eps})\delta)$-differentially private.

Finally, uniform convergence of the sampling error of $\cB'$ implies that it remains an accurate algorithm, and hence so is $\cB$. In particular, when $\db$ is a random sample of $m$ rows from $\db'$ and $d$ is the VC-dimension of $\cQ$, we have~\cite{AnthonyBa09}:
\[\Pr[\exists q \in \cQ : |q(\db) - q(\db')| > \alpha'] \le 4 \cdot \left(\frac{2em}{d}\right)^d \cdot \exp\left(-\frac{(\alpha')^2 m}{8}\right).\]
Taking $\alpha'$ as in the theorem statement makes the total failure probability of $\cB$ at most $1/2$.
\end{proof}

Now we proceed to complete the proof of Theorem \ref{thm:weak-fpc}. Suppose $(\cQ, \cX)$ has sample complexity greater than $n$ for $(\alpha + \alpha', \beta)$-accuracy (with failure probability $1/2$) and $(6\eps, (e^{2\eps} + e^{5\eps})\delta)$-differential privacy. By Lemma \ref{lem:one-sided}, for every $(\alpha, \beta)$-accurate mechanism $\cA$ for $\cQ$ there exists a database $\db \in \univ^{m}$ with $m = \lfloor n\eps \rfloor$, a set $T$, and an index $i$ such that
\begin{equation} \label{eqn:one-sided}
\Prob{\substack{j \getsr [m] \\ \text{coins}(\cA)}}{\cA(\db_{-j}) \in T} > e^{\eps}\Prob{\text{coins}(\cA)}{\cA(\db_{-i}) \in T} + \delta.
\end{equation}
We now argue that it is without loss of generality to restrict our attention to mechanisms $\cA$ whose range is the finite set $I_m^{|\cQ|} = \{0, \frac{1}{2m}, \frac{1}{m}, \dots, 1-\frac{1}{2m}, 1\}^{|\cQ|}$. To see this, note that the exact answer to any counting query $q$ on a database $\db \in \univ^m$ is in the set $\{0, \frac{1}{m}, \frac{2}{m}, \dots, 1 - \frac{1}{m}, 1\}$. Therefore, if an answer $a \in [0, 1]$ satisfies $|a - q(\db)| \le \alpha$, then the value 
\[\bar{a} = \frac{1}{2m} \cdot \left(\lceil (a - \alpha)m \rceil + \lfloor (a + \alpha)m \rfloor\right)\]
is a point in $I_m$ that also satisfies $|\bar{a} - q(\db)| \le \alpha$.  Thus, we will henceforth assume that the mechanism's output lies in this finite range.

We now apply the min-max theorem from game theory (or equivalently, linear programming duality), to exhibit a fixed distribution on $(\db, T, i)$ for which Inequality (\ref{eqn:one-sided}) holds. Specifically, consider a two-player zero-sum game in which Player 1 chooses a triple $(\db, T, i)$, where $\db \in \univ^m$, $T \subseteq I_m^{|\cQ|}$, and $i \in [m]$, and Player 2 chooses a randomized function $\cA : \univ^m \to I_m^{|\cQ|}$ that is $(\alpha, \beta)$-accurate for $\cQ$. Let the payoff to Player 1 be
\[\Pr_{j \getsr [m]}[\cA(\db_{-j}) \in T] - e^{\eps} \mathbb{ I} (\cA(\db_{-i}) \in T).\]
By inequality (\ref{eqn:one-sided}), the value of this game is greater than $\delta$. So by the min-max theorem there exists a mixed strategy for Player 1 that achieves a payoff greater than $\delta$ against any mixed strategy for Player 2. (Note that we can apply the min-max theorem because we have assumed that the mechanism's output lies in a finite range.)  That is, there exists a distribution $\cD$ over triples $(\db, T, i)$ such that for any randomized algorithm $\cA : \univ^m \to I_m^{|\cQ|}$ that takes any $\db$ to a feasible vector in $F_{\alpha, \beta}(\db)$ with probability at least $2/3$,

\begin{equation} \label{eqn:weakfpc-dist}
\Prob{\substack{j \getsr [m] \\ \text{coins}(\cA) \\ (\db, T, i) \getsr \cD}}{\cA(\db_{-j}) \in T} > e^{\eps} \cdot \Prob{\substack{\text{coins}(\cA) \\ (\db, T, i) \getsr \cD}}{\cA(\db_{-i}) \in T} + \delta.
\end{equation}

Now consider the following code: $\gen$ samples a database $\db$, a set $T$, and an index $i$ according to the promised distribution $\cD$. The codebook $C$ is $(\db_{\pi(1)}, \dots, \db_{\pi(m)})$ where $\pi: [m] \to [m]$ is a random permutation. On input an answer vector $a$, the algorithm $\trace$ checks whether $a \in T$. If it is, then $\trace$ outputs $\pi(i)$, and otherwise outputs $\bot$.

To analyze the security of this code, fix a coalition $S$ of $m-1$ users using a pirate strategy $\fpadv$. Because the codebook is a random permutation of the rows of $\db$, it is equivalent to analyze the original database $\db$ and a random coalition of $m-1$ users. Thus the part of the codebook $C_S$ given to the pirates is a random set of $m - 1$ rows from $\db$, i.e. $\db_{-j}$ for a random $j \in [m]$ with the junk row at index $j$ removed. The condition that $\fpadv$ outputs a feasible answer vector is equivalent to $a = \fpadv(C_S)$ being an $(\alpha, \beta)$-accurate answer vector. Therefore, letting $\cA : \univ^m \to I_m^{|\cQ|}$ be the algorithm that runs $\fpadv$ on its input with the junk row removed, we have
\[\Pr_{\gen, \trace, \fpadv}[\trace(C, a) \ne \bot] = \Pr_{\substack{\text{coins}(\fpadv) \\ (\db, T, i) \getsr \cD, \pi}}[\fpadv(C_S) \in T] = \Pr_{\substack{j \getsr [m], \text{coins}(\cA) \\ (\db, T, i) \getsr \cD}}[\cA(\db_{-j}) \in T].\]
On the other hand, the probability that $\trace$ outputs the user $j$ not in the coalition is
\begin{align*}
\Pr_{\gen, \trace, \fpadv}[\trace(C, a) = i] &= \Pr_{\substack{j \getsr [m], \text{coins}(\fpadv) \\ (\db, T, i) \getsr \cD, \pi}}[\trace(C, a) = i \land j = i] \\
&= \frac{1}{m} \cdot \Pr_{\text{coins}(\cA), (\db, T, i) \getsr \cD}[\cA(\db_{-i}) \in T],
\end{align*}
because the events $\{j = i\}$ and $\{\trace(C, a) = i\}$ are independent. Thus by (\ref{eqn:weakfpc-dist}),
\[\Pr[\trace(a) \ne \bot] > e^{\eps} m \cdot \Pr[\trace(a) \in [m] \setminus S] + \delta,\]
where both probabilities are taken over the coins of $\gen, \trace$, and $\fpadv$.
\end{proof}

\section{A Composition Theorem for Sample Complexity} \label{sec:comp}

\newcommand{\multsofm}{\set{0,1/\alphsize,\dots,(\alphsize - 1)/\alphsize,1}}

In this section we state and prove a composition theorem for sample complexity lower bounds.  At a high-level the composition theorem starts with two pairs, $(\queryset, \univ)$ and $(\queryset', \univ')$, for which we know sample-complexity lower bounds of $\recrows$ and $\privrows$ respectively, and attempts to prove a sample-complexity lower bound of $\recrows \cdot \privrows$ for a related family of queries on a related data universe.  

Specifically, our sample-complexity lower bound will apply to the ``product'' of $\queryset$ and $\queryset'$, defined on $\univ \times \univ'$.  We define the product $\queryset \land \queryset'$ to be
$$
\queryset \land \queryset' = \{q \land q' \from (x, x') \mapsto q(x) \land q'(x') \mid q \in \queryset, q \in \queryset'\}.
$$
Since $q,q'$ are boolean-valued, their conjunction can also be written $q(x) q'(x')$.  

We now begin to describe how we can prove a sample complexity lower bound for $\queryset \land \queryset'$.  First, we describe a certain product operation on databases.  Let $\recdb \in \univ^{\recrows}$, $\recdb = (\recrow_1,\dots,\recrow_{\recrows})$, be a database.  Let $\privdb_1,\dots,\privdb_\recrows \in (\univ')^{\privrows}$ where $\privdb_i = (\privrow_{i1},\dots,\privrow_{i\privrows})$ be $\recrows$ databases.  We define the product database $\compdb = \recdb \times (\privdb_1,\dots,\privdb_{\recrows}) \in (\recuniv \times \privuniv)^{\recrows \cdot \privrows}$ as follows:  For every $i = 1,\dots,\recrows, j = 1,\dots,\privrows$, let the $(i,j)$-th row of $\compdb$ be $x^*_{(i,j)} = (\recrow_i, \privrow_{ij})$.  Note that we index the rows of $\compdb$ by $(i,j)$.  We will sometimes refer to $\privdb_1,\dots,\privdb_{\recrows}$ as the ``subdatabases'' of $\compdb$.  

The key property of these databases is that we can use a query $q \land q' \in \queryset \land \queryset'$ to compute a ``subset-sum'' of the vector $s_{q'} = (q'(D'_1),\dots,q'(D'_n))$ consisting of the answers to $q'$ on each of the $n$ subdatabases.  That is, for every $q \in \queryset$ and $q' \in \queryset'$,
\begin{equation} \label{eqn:subsetsum}
(q \land q')(D^*) 
= \frac{1}{\recrows \cdot \privrows} \sum_{i=1}^{\recrows} \sum_{j=1}^{\privrows} (q \land q')(x^*_{(i,j)})
= \frac{1}{n} \sum_{i=1}^{n} q(x_i) q'(D'_i).
\end{equation}

Thus, every approximate answer $\answer_{\recquery \land \privquery}$ to a query $\recquery \land \privquery$ places a subset-sum constraint on the vector $s_{\privquery}$.  (Namely, $\answer_{\recquery \land \privquery} \approx \frac{1}{\recrows} \sum_{i=1}^{\recrows} \recquery(\recrow_i) \privquery(\privdb_i)$)
If the database $D$ and family $\queryset$ are chosen appropriately, and the answers are sufficiently accurate, then we will be able to reconstruct a good approximation to $s_{q'}$.  Indeed, this sort of ``reconstruction attack'' is the core of many lower bounds for differential privacy, starting with the work of Dinur and Nissim~\cite{DinurNi03}.  The setting they consider is essentially the special case of what we have just described where $D'_1,\dots,D'_n$ are each just a single bit ($\univ' = \bits$, and $\queryset'$ contains only the identity query).  In Section~\ref{sec:applications} we will discuss choices of $D$ and $\queryset$ that allow for this reconstruction.  

We now state the formal notion of reconstruction attack that we want $D$ and $\queryset$ to satisfy.
\begin{comment}
We also consider reconstruction attacks.  Here each record of the database contains a piece of ``sensitive information,'' and the adversary will attempt to reconstruct almost all of the sensitive information.  Attacks of this nature were first considered in \cite{DinurNi03}. In that work, and subsequent work on reconstruction attacks, the sensitive information was a single bit.  The aim of a reconstruction attack was to recover almost all of these sensitive bits from accurate answers to queries.

We generalize this notion slightly by aiming to reconstruct databases that are vectors of fractional entries.  Instead of aiming to recover most of these entries exactly, we aim to reconstruct entries that are accurate on average (in the $\ell_1$-norm). That is, each record contains a sensitive value in $\{0,1/m,\dots,1\}$ for some choice of $m$ and we will approximate these sensitive values with small average error.
\end{comment}
\begin{definition} [Reconstruction Attacks] \label{def:recattack}
Let $\queryset$ be a family of counting queries over a data universe $\univ$.  Let $\rows \in \N$ and $\alpha', \alpha, \beta \in [0,1]$ be parameters.  Let $\db = (\row_1,\dots,\row_{\rows}) \in \univ^\rows$ be a database.  Suppose there is an adversary $\recadv_{D} : \R^{|\queryset|} \to [0,1]^\rows$ with the following property:  For every vector $s \in [0,1]^\rows$ and every sequence $\answer = (\answerq{\query})_{\query \in \queryset} \in \R^{|\queryset|}$ such that
$$
\left|\answerq{\query} - \frac{1}{n}\sum_{i = 1}^\rows q(\row_i)s_i\right| < \alpha
$$
for at least a $1-\beta$ fraction of queries $\query \in \queryset$,
$\recadv_{\recdb}(\answer)$ outputs a vector $t \in [0,1]^{\rows}$ such that
$$
\frac{1}{n} \sum_{i=1}^{\rows} | t_i - s_i | \leq \alpha'.
$$
Then we say that $D \in \univ^\rows$ \emph{enables an $\alpha'$-reconstruction attack from $(\alpha, \beta)$-accurate answers to $\queryset$}.
\mnote{We kept the phrasing as ``from accurate answers'' instead of changing it to ``from correlations''}
%If there exists such a database, then we say that \emph{$(\queryset, \univ)$ enables an $\rows$-row database that is $(\alpha', \alphsize)$-reconstructible database from $(\alpha, \beta)$-accurate answers}.
%When defining reconstruction attacks, we will sometimes abuse notation by writing $\query(s_i)$ for $\query(\row_i)s_i$.
\end{definition}

A reconstruction attack itself implies a sample-complexity lower bound, as in~\cite{DinurNi03}.  However, we show how to obtain stronger sample complexity lower bounds from the reconstruction attack by applying it to a product database $D^*$ to obtain accurate answers to queries on its subdatabases.  For each query $q' \in \queryset'$, we run the adversary promised by the reconstruction attack on the approximate answers given to queries of the form $(q \land q') \in \queryset \land \{q'\}$.  As discussed above, answers to these queries will approximate subset sums of the vector $s_{q'} = (q'(\privdb_1),\dots,q'(\privdb_\recrows))$.  When the reconstruction attack is given these approximate answers, it returns a vector $t_{q'} = (t_{q',1},\dots,t_{q', \recrows})$ such that $t_{q',i} \approx s_{q',i} = q'(\privdb_i)$ on average over $i$.  Running the reconstruction attack for every query $\privquery$ gives us a collection $t = (t_{q',i})_{q' \in \queryset', i \in [\recrows]}$ where $t_{q',i} \approx q'(\privdb_i)$ on average over both $q'$ and $i$.  By an application of Markov's inequality, for most of the subdatabases $\privdb_i$, we have that $t_{q',i} \approx q'(\privdb_i)$ on average over the choice of $q' \in \queryset'$.  For each $i$ such that this guarantee holds, another application of Markov's inequality shows that for most queries $q' \in \queryset'$ we have $t_{q', i} \approx q'(\privdb_i)$, which is our definition of $(\alpha, \beta)$-accuracy (later enabling us to apply a re-identification adversary for $\privqueryset$).

The algorithm we have described for obtaining accurate answers on the subdatabases is formalized in Figure~\ref{fig:compadv0}.
\begin{figure}[ht]
\begin{framed}
\begin{algorithmic}
\STATE{Let $a = (a_{q \land q'})_{q \in \queryset, q'\in\queryset'}$ be an answer vector.}
\STATE{Let $\recadv_{\recdb} \from \R^{|\recqueryset|} \to [0,1]^{\recrows}$ be a reconstruction attack.}
\STATE{For each $q' \in \queryset'$}
\INDSTATE[1]{Let $(t_{q', 1}, \dots, t_{q', n}) = \recadv_{\recdb}((a_{q \land q'})_{q \in \queryset})$}
\STATE{Output $(t_{\privquery, i})_{\privquery \in \privqueryset, i \in [n]}$.}
\end{algorithmic}
\end{framed}
\vspace{-6mm}
\caption{The reconstruction $\cR_{\recdb}^*(a)$.}
\label{fig:compadv0}
\end{figure}

We are now in a position to state the main lemma that enables our composition technique.  The lemma says that if we are given accurate answers to $\queryset \land \queryset'$ on $\compdb$ and the database $\recdb \in \recuniv^{\recrows}$ enables a reconstruction attack from accurate answers to $\recqueryset$, then we can obtain accurate answers to $\privqueryset$ on most of the subdatabases $\privdb_1,\dots,\privdb_{\recrows} \in (\univ')^{\privrows}$.
\begin{lemma} \label{lem:comp0}
Let $\recdb \in \recuniv^{\recrows}$ and  $\privdb_1,\dots,\privdb_n \in (\univ')^{\privrows}$ be databases and $D^* \in (\univ \times \univ')^{\recrows \cdot \privrows}$ be as above.  Let $\answer = (\answer_{\recquery \land \privquery})_{\recquery \in \recqueryset, \privquery \in \privqueryset} \in \R^{|\recqueryset \land \privqueryset|}$. Let $\alpha', \alpha, \beta \in [0,1]$ be parameters.  Suppose that for some parameter $c > 1$, the database $D$ enables an $\alpha'$-reconstruction attack from $(\alpha, c\beta)$-accurate answers to $\queryset$.  Then if $(t_{q', i})_{q' \in \queryset', i \in [\recrows]} = \cR_{\recdb}^*(a)$ (Figure~\ref{fig:compadv0}),
\begin{align*}
&\textrm{$a$ is $(\alpha, \beta)$-accurate for $\queryset \land \queryset'$ on $D^*$} \\
&\Longrightarrow{} \Prob{i \getsr [n]}{\textrm{$(t_{q', i})_{q' \in \queryset'}$ is $(6c\alpha', 2/c)$-accurate for $\queryset'$ on $D_i$}} \geq 5/6.
\end{align*}
\end{lemma}

The additional bookkeeping in the proof is to handle the case where $a$ is only accurate for most queries.  In this case the reconstruction attack may fail completely for certain queries $q' \in \queryset'$ and we need to account for this additional source of error.

\begin{proof}[Proof of Lemma~\ref{lem:comp0}]
Assume the answer vector $\answer = (\answer_{q \land q'})_{q \in \queryset, q' \in \queryset'}$ is $(\alpha,\beta)$-accurate for $\recqueryset \land \privqueryset$ on $\compdb = \recdb \times (\privdb_1,\dots,\privdb_{\recrows})$.  By assumption, $\recdb$ enables a reconstruction attack $\recadv_{\recdb}$ that succeeds in reconstructing an approximation to $s_{\query'} = (\query'(\privdb_1),\dots,q'(\privdb_{\recrows}))$ when given $(\alpha, c\beta)$-accurate answers for the family of queries $\queryset \land \{q'\}$.  Consider the set of $q'$ on which the reconstruction attack succeeds, i.e.
$$
\goodqueryset = \set{q' \mid \textrm{$(\answer_{q \land q'})_{q \in \queryset}$ is $(\alpha, c\beta)$-accurate for $\queryset \land \{q'\}$}}.
$$
Since $a$ is $(\alpha, \beta)$-accurate, an application of Markov's inequality shows that $$\prob{q' \in \goodqueryset} \geq 1-1/c.$$  Thus, $|\goodqueryset| \geq (1-1/c)|\queryset'|$.

Recall that, by~\eqref{eqn:subsetsum}, we can interpret answers to $\queryset \land \queryset'$ as subset sums of answers to the subdatabases, so for every $q' \in \goodqueryset$,
$$
\left|\answer_{\recquery \land \privquery} - \frac{1}{\recrows}\sum_{i=1}^\recrows \recquery(\recrow_i) q'(\privdb_i) \right| < \alpha
$$
for at least a $1-c \beta$ fraction of queries $\recquery \land \privquery \in \recqueryset \land \{\privquery\}$.  Since $D$ enables a reconstruction attack from $(\alpha, c\beta)$-accurate answers to $\queryset$, by Definition \ref{def:recattack}, $\recadv_{\recdb}((\answer_{q \land q'})_{q \in \queryset})$ recovers a vector $t_{\privquery} \in [0,1]^{\recrows}$ such that
$$
\frac{1}{ \recrows} \sum_{i = 1}^\recrows \left| t_{q', i} - q'(\privdb_i) \right| < \alpha'.
$$
%For the remaining $1/c$ fraction of queries for which the condition fails, $\recadv(\recdb, \answer^{\privquery})$ recovers an arbitrary vector $\hat{s}^{\privquery} \in \multsofm^{\recrows}$ that is not guaranteed to be close to $s^{\privquery}$.
Since this holds for every $q' \in \goodqueryset$, we have
\begin{align*}
&\Ex{q' \getsr \goodqueryset, i \getsr [\recrows]}{|t_{q', i} - q'(\privdb_i)|} \leq \alpha' \\
\Longrightarrow &\Prob{i \getsr [\recrows]}{ \Ex{q' \in \goodqueryset}{|t_{q',i} - q'(\privdb_i)|} \leq 6\alpha' } \geq 5/6 \tag{Markov} \\
\Longrightarrow &\Prob{i \getsr [\recrows]}{ \textrm{$|t_{q',i} - q'(\privdb_i)| \leq 6c\alpha'$ for at least a $1-1/c$ fraction of $q' \in \goodqueryset$} } \geq 5/6 \tag{Markov} \\
\Longrightarrow &\Prob{i \getsr [\recrows]}{ \textrm{$|t_{q',i} - q'(\privdb_i)| \leq 6c\alpha'$ for at least a $1-2/c$ fraction of $q' \in \privqueryset$} } \geq 5/6 \tag{since $|\goodqueryset| \geq (1-1/c)|\privqueryset|$}
\end{align*}
The statement inside the final probability is precisely that $(t_{q',i})_{q' \in \queryset'}$ is $(6c\alpha', 2/c)$-accurate for $\queryset'$ on $\privdb_i$.  This completes the proof of the lemma.
\end{proof}

We now explain how the main lemma allows us to prove a composition theorem for sample complexity lower bounds.  We start with a query family $\queryset$ on a database $\recdb \in \recuniv^\recrows$ that enables a reconstruction attack, and a distribution $\privdist$ over databases in $(\privuniv)^{\privrows}$ that is re-identifiable from answers to a family $\queryset'$. We show how to combine these objects to form a re-identifiable distribution $\compdist$ for queries $\queryset \land \queryset'$ over $(\recuniv \times \privuniv)^{\recrows \cdot \privrows}$, yielding a sample complexity lower bound of $\recrows \cdot \privrows$.

A sample from $\compdist$ consists of $\compdb = \recdb \times (\privdb_1,\dots,\privdb_{\recrows})$ where each subdatabase $\privdb_i$ is an independent sample from from $\privdist$. The main lemma above shows that if there is an algorithm $\san$ that is accurate for $\queryset \land \queryset'$ on $\compdb$, then an adversary can reconstruct accurate answers to $\queryset'$ on most of the subdatabases $\privdb_1, \dots, \privdb_{\recrows}$. Since these subdatabases are drawn from a re-identifiable distribution, the adversary can the re-identify a member of one of the subdatabases $\privdb_i$.  Since the identified member of $\privdb_i$ is also a member of $\compdb$, we will have a re-identification attack against $\compdb$ as well.

We are now ready to formalize our composition theorem.
\begin{theorem} \label{thm:composition}
Let $\recqueryset$ be a family of counting queries on $\recuniv$, and let $\privqueryset$ be a family of counting queries on $\privuniv$.  Let $\gamma, \sec, \alpha', \alpha, \beta \in [0,1]$ be parameters.  Assume that for some parameters $c > 1$, $\gamma, \sec, \alpha', \alpha, \beta \in [0,1]$, the following both hold:
\begin{enumerate}
\item There exists a database $\recdb \in \recuniv^{\recrows}$ that enables an $\alpha'$-reconstruction attack from $(\alpha, c \beta)$-accurate answers to $\recqueryset$.
\item There is a distribution $\privdist$ on databases $\db \in (\privuniv)^{\privrows}$ that is $(\gamma,\sec)$-re-identifiable from $(6c \alpha', 2/c)$-accurate answers to $\privqueryset$.
\end{enumerate}
Then there is a distribution on databases $\compdb \in (\recuniv \times \privuniv)^{\recrows \cdot \privrows}$ that is $(\gamma+1/6, \sec)$-re-identifiable from $(\alpha, \beta)$-accurate answers to $\recqueryset \land \privqueryset$.

%Let $\recadv : \R^{|\recqueryset|} \to [\privrows]^\recrows$ be a reconstruction attack for $\recqueryset$ with tolerance $\alpha$. Suppose $\privqueryset$ enables a $(\privrows, 1/4\recrows\privrows, 1/3)$-sanitization-resistant distribution. Then there is a $(\recrows\privrows, 1/4\recrows\privrows, \alpha)$-sanitization-resistant distribution in the worst-case for $\recqueryset \land \privqueryset$.  
\end{theorem}

\begin{proof}
Let $\recdb = (\row_1,\dots, \row_\rows) \in \recuniv^{\recrows}$ be the database that enables a reconstruction attack (Definition~\ref{def:recattack}).   Let $\privdist$ be the promised re-identifiable distribution on databases $D \in (\privuniv)^{\privrows}$ and $\reidadv \from (\privuniv)^{\privrows}\times \R^{|\privqueryset|} \to [\privrows] \cup \{\bot\}$ be the promised adversary (Definition~\ref{def:reidentifiabledist}).

%over pairs $(\privdb, j) \in \privuniv^\privrows \times [\privrows]$ be a $(\privrows, 1/4\privrows\recrows, \alpha)$-sanitization-resistant distribution for $\privqueryset$ with associated privacy adversary $\reidadv$.
%Let $\recdist$ be a distribution on databases $\recdb \in \recuniv^\recrows$ associated to the reconstruction attack $\recadv$, and 

In Figure~\ref{fig:compdist}, we define a distribution $\compdist$ on databases $\privdb \in (\recuniv \times \privuniv)^{\recrows \cdot \privrows}$.  In Figure~\ref{fig:compadv}, we define an adversary $\compadv \from (\recuniv \times \privuniv)^{\recrows \cdot \privrows} \times \R^{|\queryset \land \privqueryset|}$ for a re-identification attack.  The shared state of $\compdist$ and $\compadv$ will be the shared state of $\privdist$ and $\privadv$.  The next two claims show that $\compdist$ satisfies the two properties necessary to be a $(\gamma + 1/6, \sec)$-re-identifiable distribution (Definition \ref{def:reidentifiabledist}).
%(\compdb, (i, j))$ where the rows of $\compdb$ are themselves indexed by pairs of the form $(i, j)$, and a privacy adversary $\compadv$ over $\compdist$.
\begin{figure}[ht]
\begin{framed}
\begin{algorithmic}
\STATE{Let $\recdb = (\recrow_1, \dots, \recrow_\recrows) \in \recuniv^{\recrows}$ be a database that enables reconstruction.}
\STATE{Let $\privdist$ on $(\privuniv)^{\privrows}$ be a re-identifiable distribution.}
\STATE{For $i = 1, \dots, \recrows$, choose $\privdb_i \getsr \privdist$ (independently)}
\STATE{Output $\compdb = \recdb \times (\privdb_1,\dots,\privdb_\recrows) \in (\recuniv \times \privuniv)^{\recrows \cdot \privrows}$}
\end{algorithmic}
\end{framed}
\vspace{-6mm}
\caption{The new distribution $\compdist$.}
\label{fig:compdist}
\end{figure}

\begin{figure}[ht]
\begin{framed}
\begin{algorithmic}
\STATE{Let $\compdb = \recdb \times (\privdb_1,\dots,\privdb_{\recrows})$.}
\STATE{Run $\cR_{\recdb}^*(\san(\compdb))$ (Figure~\ref{fig:compadv0}) to reconstruct a set of approximate answers $(t_{\privquery, i})_{\privquery \in \privqueryset, i \in [\recrows]}$.}
\STATE{Choose a random $i \getsr [\recrows]$.}
\STATE{Output $\privadv(\privdb_i, (t_{\privquery, i})_{\privquery \in \privqueryset})$.}
\end{algorithmic}
\end{framed}
\vspace{-6mm}
\caption{The privacy adversary $\compadv(\compdb, \san(\compdb))$.}
\label{fig:compadv}
\end{figure}

\begin{claim} \label{clm:comp1}
$$
\Prob{\compdb \getsr \compdist \atop \textrm{coins$(\san)$}, \textrm{coins$(\compadv)$}}{(\compadv(\compdb, \san(\db^*)) = \bot) \land (\textrm{$\san(\compdb)$ is $(\alpha,\beta)$-accurate for $\queryset \land \queryset'$})} \leq \gamma + 1/6.
$$
\end{claim}
\begin{proof}[Proof of Claim~\ref{clm:comp1}]  Assume that $\san(\compdb)$ is $(\alpha, \beta)$-accurate for $\recqueryset \land \privqueryset$.  By Lemma~\ref{lem:comp0}, we have
\begin{align} \label{eq:comp2}
\Prob{i \getsr [n] \atop \textrm{coins$(\san)$}, \textrm{coins$(\compadv)$}}{(\textrm{$\san(\compdb)$ is $(\alpha,\beta)$-accurate for $\queryset \land \queryset'$}) \atop \land (\textrm{$(t_{q', i})_{q' \in \queryset'}$ is not $(6c\alpha', 2/c)$-accurate for $\queryset'$ on $D_i$})} \leq 1/6.
\end{align}
\newcommand{\cG}{\mathcal{G}}

By construction of $\compadv$,
\begin{align}
&\Prob{\compdb \getsr \compdist}{(\compadv(\compdb, \san(\db^*)) = \bot) \land (\textrm{$\san(\compdb)$ is $(\alpha,\beta)$-accurate for $\queryset \land \queryset'$})} \notag \\
={} &\Prob{\compdb \getsr \compdist \atop i \getsr [\recrows]}{(\privadv(\privdb_i, (t_{q',i})_{q' \in \privqueryset}) = \bot) \land (\textrm{$\san(\compdb)$ is $(\alpha,\beta)$-accurate for $\queryset \land \queryset'$})} \notag \\
\leq{} &\Prob{\compdb \getsr \compdist \atop i \getsr [\recrows]}{(\privadv(\privdb_i, (t_{q',i})_{q' \in \privqueryset}) = \bot) \land (\textrm{$(t_{q',i})$ is $(6c \alpha',2/c)$-accurate for $\queryset'$})} + \frac{1}{6} \label{eq:comp3}
\end{align}
where the last inequality is by~\eqref{eq:comp2}.  Thus, it suffices to prove that
\begin{equation} \label{eq:comp1}
\Prob{\compdb \getsr \compdist \atop i \getsr [\recrows]}{(\privadv(\privdb_i, (t_{q',i})_{q' \in \privqueryset}) = \bot) \land (\textrm{$(t_{q',i})$ is $(6c \alpha',2/c)$-accurate for $\queryset'$})}  \leq \gamma
\end{equation}
We prove this inequality by giving a reduction to the re-identifiability of $\privdist$. Consider the following sanitizer $\san'$:  On input $\privdb \getsr \privdist$, $\san'$ first chooses a random index $i^* \getsr [\recrows]$. Next, it samples $\privdb_1,\dots,\privdb_{i^*-1}, \privdb_{i^*+1},\dots,\privdb_{\recrows} \getsr \privdist$ independently, and sets $\privdb_{i^*} = \privdb$.  Finally, it runs $\san$ on $\compdb = \recdb \times (\privdb_1,\dots,\privdb_{\recrows})$ and then runs the reconstruction attack $\cR^*$ to recover answers $(t_{\privquery, i})_{\privquery \in \privqueryset, i \in [\recrows]}$ and outputs $(t_{\privquery, i^*})_{\privquery \in \privqueryset}$.

Notice that since $\privdb_1,\dots,\privdb_{\recrows}$ are all i.i.d. samples from $\privdist$, their joint distribution is independent of the choice of $i^*$.  Specifically, in the view of $\compadv$, we could have chosen $i^*$ after seeing its output on $\compdb$.  Therefore, the following random variables are identically distributed:
\begin{enumerate}
\item $(t_{q', i})_{q' \in \queryset'}$, where $(t_{q', i})_{q' \in \queryset', i \in [n]}$ is the output of $\cR_{\recdb}^*(\san(\compdb))$ on $\compdb \getsr \compdist$, and $i \getsr [\recrows]$.
\item $\san'(D')$ where $D' \getsr \privdist$.
\end{enumerate}
Thus we have
\begin{align*}
&\Prob{\compdb \getsr \compdist \atop i \getsr [\recrows]}{(\privadv(\privdb_i, (t_{q',i})_{q' \in \privqueryset}) = \bot) \land (\textrm{$(t_{q',i})$ is $(6c \alpha',2/c)$-accurate for $\queryset'$})} \\
={} &\Prob{\privdb \getsr \privdist}{(\privadv(\privdb, \san'(\privdb)) = \bot) \land (\textrm{$\san'(\privdb)$ is $(6c \alpha',2/c)$-accurate for $\queryset'$})} \leq \gamma
\end{align*}
where the last inequality follows because $\privdist$ is a $(\gamma, \sec)$-re-identifiable from $(6c \alpha', 2/c)$-accurate answers to $\privqueryset$.  Thus we have established~\eqref{eq:comp1}.  Combining~\eqref{eq:comp3} and~\eqref{eq:comp1} completes the proof of the claim.

\end{proof}

The next claim follows directly from the definition of $\compadv$ and the fact that $\privdist$ is $(\gamma,\sec)$-re-identifiable.
\begin{claim} \label{clm:comp2}
For every $(i,j) \in [\recrows] \times [\privrows]$,
$$
\Prob{\db \getsr \compdist}{\compadv(\db, \san(\db_{-(i,j)})) = (i,j)} \leq \sec.
$$
\end{claim}

Combining Claims~\ref{clm:comp1} and~\ref{clm:comp2} suffices to prove that $\compdist$ is $(\gamma + 1/6, \sec)$-re-identifiable from $(\alpha, \beta)$-accurate answers to $\recqueryset \land \privqueryset$, completing the proof of the theorem.
\end{proof}

%%%Some old equations we don't use.  Couldn't bear to delete them.
\begin{comment}
For the remainder of the proof, we will show that~\eqref{eq:comp1} is bounded by $\gamma + 1/6$, which will be sufficient to prove the claim.
We can split~\eqref{eq:comp1} based on whether or not $i \in \cG$.
\begin{align*}
\eqref{eq:comp1} = 
&\Prob{\compdb \getsr \compdist \atop i \getsr [\recrows]}{(\privadv(\privdb_i, (t_{q',i})_{q' \in \privqueryset}) = \bot) \land (\textrm{$\san(\compdb)$ is $(\alpha,\beta)$-accurate for $\queryset \land \queryset'$}) \land (i \in \cG)} \\
&+ \Prob{\compdb \getsr \compdist \atop i \getsr [\recrows]}{(\privadv(\privdb_i, (t_{q',i})_{q' \in \privqueryset}) = \bot) \land (\textrm{$\san(\compdb)$ is $(\alpha,\beta)$-accurate for $\queryset \land \queryset'$}) \land (i \not\in \cG)}
\end{align*}
The second term we can bound by $1/6$ because $|\cG| \geq 5\recrows / 6$.  For the first term, we can apply~\eqref{eq:comp2}.
\begin{align*}
&\Prob{\compdb \getsr \compdist \atop i \getsr [\recrows]}{(\privadv(\privdb_i, (t_{q',i})_{q' \in \privqueryset}) = \bot) \land (\textrm{$\san(\compdb)$ is $(\alpha,\beta)$-accurate for $\queryset \land \queryset'$}) \land (i \in \cG)} \\
\leq{} &\Prob{\compdb \getsr \compdist \atop i \getsr [\recrows]}{(\privadv(\privdb_i, (t_{q',i})_{q' \in \privqueryset}) = \bot) \land (\textrm{$(t_{q',i})$ is $(6c \alpha',2/c)$-accurate for $\queryset'$})} \\
\end{align*}
If we can show that the final expression is at most $\gamma$, then we will prove that $\eqref{eq:comp1} \leq \gamma + 1/6$.
\end{comment}

The proof of Theorem \ref{thm:composition} also yields a composition theorem for generalized fingerprinting codes. Specifically, Theorem \ref{thm:composition-fpc} below shows how to combine a reconstruction attack for a query family $\recqueryset$ on a database $D \in \recuniv^\recrows$ with a $(\privrows, \privqueryset)$-generalized fingerprinting code to obtain a $(\recrows \cdot \privrows, \recqueryset \land \privqueryset)$-generalized fingerprinting code.

\begin{theorem} \label{thm:composition-fpc}
Let $\recqueryset$ be a family of counting queries on $\recuniv$, and let $\privqueryset$ be a family of counting queries on $\privuniv$.  Let $\gamma, \sec, \alpha', \alpha, \beta \in [0,1]$ be parameters.  Assume that for some parameters $c > 1$, $\gamma, \sec, \alpha', \alpha, \beta \in [0,1]$, the following both hold:
\begin{enumerate}
\item There exists a database $\recdb \in \recuniv^{\recrows}$ that enables an $\alpha'$-reconstruction attack from $(\alpha, c \beta)$-accurate answers to $\recqueryset$.
\item There exists a $(\privrows, \privqueryset)$-generalized fingerprinting code for $(6c\alpha', 2/c)$-accuracy with security $(\gamma, \sec)$.
%There is a distribution $\privdist$ on databases $\db \in (\privuniv)^{\privrows}$ that is $(\gamma,\sec)$-re-identifiable from $(6c \alpha', 2/c)$-accurate answers to $\privqueryset$.
\end{enumerate}
Then there is a $(\recrows \cdot \privrows, \recqueryset \land \privqueryset)$-generalized fingerprinting code for $(\alpha, \beta)$-accuracy with security $(\gamma + 1/6, \sec)$.
%Then there is a distribution on databases $\compdb \in (\recuniv \times \privuniv)^{\recrows \cdot \privrows}$ that is $(\gamma+1/6, \sec)$-re-identifiable from $(\alpha, \beta)$-accurate answers to $\queryset$.
\end{theorem}

\section{Applications of the Composition Theorem} \label{sec:applications}

In this section we show how to use our composition theorem (Section~\ref{sec:comp}) to combine our new lower bounds for $1$-way marginal queries from Section~\ref{sec:1way} with (variants of) known lower bounds from the literature to obtain our main results.  In Section~\ref{sec:lbkway} we prove a lower bound for $k$-way marginal queries when $\alpha$ is not too small (at least inverse polynomial in $d$), thereby proving Theorem~\ref{thm:main1} in the introduction.  Then in Section~\ref{sec:lbarbitrary} we obtain a similar lower bound for arbitrary counting queries that allows $\alpha$ to take a wider range of parameters.. 

\subsection{Lower Bounds for $k$-Way Marginals} \label{sec:lbkway}
In this section, we carry out the composition of sample complexity lower bounds for $k$-way marginals as described in the introduction (Theorem~\ref{thm:main1}).  Recall that we obtain our new $\tilde{\Omega}(k\sqrt{d}/\alpha^2)$ lower bound by combining three lower bounds:
\begin{enumerate}
\item Our re-identification based $\tilde{\Omega}(\sqrt{d})$ lower bound for $1$-way marginals (Section~\ref{sec:lb1way}),
\item A known reconstruction-based lower bound of $\Omega(k)$ for $k$-way marginals.
\item A known reconstruction-based lower bound of $\Omega(1/\alpha^2)$ for $k$-way marginals.
\end{enumerate}
The lower bound of $\Omega(k)$ for $k$-way marginals is a special case of a lower bound of $\Omega(\vcdim(\queryset))$ due to~\cite{Roth10} and based on~\cite{DinurNi03}, where $\vcdim(\queryset)$ is the \emph{Vapnik-Chervonenkis (VC) dimension} of $\queryset$.  The lower bound of $\Omega(1/\alpha^2)$ for $k$-way marginals is due to~\cite{KasiviswanathanRuSmUl10, De12}.

To apply our composition theorem, we need to formulate these reconstruction attack in the language of Definition \ref{def:recattack}. In particular, we observe that these reconstruction attacks readily generalize to allow us to reconstruct fractional vectors $s \in [0, 1]^n$, instead of just boolean vectors as in \cite{DinurNi03, Roth10}.

\subsubsection{The $\Omega(k)$ Lower Bound}
First we state and prove that the linear dependence on $k$ is necessary.
\begin{definition}[VC Dimension of Counting Queries] \label{def:vcdim}
Let $\queryset$ be a collection of counting queries over a data universe $\univ$. We say a set $\{x_1, \dots, x_k\} \subseteq \univ$ is \emph{shattered} by $\queryset$ if for every string $v \in \{0, 1\}^k$, there exists a query $\query \in \queryset$ such that $(q(x_1), \dots, q(x_k)) = (v_1, \dots, v_k)$. The \emph{VC-Dimension} of $\queryset$ denoted $\vcdim(\queryset)$ is the cardinality of the largest subset of $\univ$ that is shattered by $\queryset$.
\end{definition}
\begin{fact} \label{fact:vcconj}
The set of $k$-way conjunctions $\kdconj{k}$ over any data universe $\bits^d$ with $d \ge k$ has VC-dimension $\vcdim(\kdconj{k}) \ge k$.\footnote{More precisely, $\vcdim(\kdconj{k}) \geq k \log_2 (\lfloor d/k \rfloor),$ but we use the simpler bound $\vcdim(\kdconj{k}) \geq k$ to simplify calculation, since our ultimate lower bounds are already suboptimal by $\polylog(d)$ factors for other reasons.}
\end{fact}
\begin{proof}
For each $i = 1, \dots, k$, let $x_i = (1, 1, \dots, 0, \dots, 1)$ where the zero is at the $i$-th index. We will show that $\{x_1, \dots, x_k\}$ is shattered by $\kdconj{k}$. For a string $v \in \bits^k$, let the query $q_v(x)$ take the conjunction of the bits of $x$ at indices set to $0$ in $v$. Then $q_v(x_i) = 1$ iff $v_i = 1$, so $(q_v(x_1), \dots, q_v(x_k)) = (v_1, \dots, v_k)$.
\end{proof}
\begin{comment}
\paragraph{Remark.}
The VC-Dimension of $k$-way conjunctions is actually known to be $\Omega(k \log (d/k))$. However, the additional logarithmic factor does not qualitatively strengthen our results.
\end{comment}
\begin{lemma}[Variant of~\cite{DinurNi03,Roth10}] \label{lem:vcrec}
Let $\queryset$ be a collection of counting queries over a data universe $\univ$ and let $n = \vcdim(\queryset)$. Then there is a database $D \in \univ^n$ which enables a $4\alpha$-reconstruction attack from $(\alpha, 0)$-accurate answers to $\cQ$.
\end{lemma}
\begin{proof}
Let $\{x_1, \dots, x_n\}$ be shattered by $\queryset$, and consider the database $D = (x_1, \dots, x_n)$. Let $s \in [0,1]^n$ be an arbitrary string to be reconstructed and let $\answer = (\answer_{\query})_{\query \in \queryset}$ be $(\alpha, 0)$-accurate answers. That is, for every $\query \in \queryset$
$$
\left| \answer_\query - \frac{1}{\rows} \sum_{i=1}^{\rows} q(x_i) s_i \right| \leq \alpha
$$
Consider the brute-force reconstruction attack $\recadv$ defined in Figure \ref{fig:vcrec}. Notice that, since $\answer$ is $(\alpha, 0)$-accurate, $\recadv$ always finds a suitable vector $t$. Namely, the original database $s$ satisfies the constraints.
\begin{figure}[ht]
\begin{framed}
\begin{algorithmic}
\STATE{\textbf{Input:} Queries $\cQ$, and $(a_q)_{q \in \cQ}$ that are $(\alpha,0)$-accurate for $s$.}
\STATE{Find any $t \in [0,1]^n$ such that
$$
\left| a_q - \frac{1}{n} \sum_{i=1}^n \query(\row_i)t_i\right| \leq \alpha \quad \forall \query \in \cQ.
$$ }
\STATE{\textbf{Output:} $t$.}
\end{algorithmic}
\end{framed}
\vspace{-6mm}
\caption{The reconstruction adversary $\recadv(\db, \answer)$.}
\label{fig:vcrec}
\end{figure}
We will show that the reconstructed vector $t$ satisfies
$$
\frac{1}{n} \sum_{i=1}^{n} |t_i - s_i| \le 4\alpha.
$$
Let $T$ be the set of coordinates on which $t_i > s_i$ and let $S$ be the set of coordinates where $s_i > t_i$. Note that
\[\sum_{i=1}^{n} |t_i - s_i| = \sum_{i \in T} (t_i - s_i) + \sum_{i \in S} (s_i - t_i).\]
We will show that absolute values of the sums over $T$ and $S$ are each at most $2\alpha$. Since $\{x_1, \dots, x_n\}$ is shattered by $\queryset$, there is a query $\query \in \queryset$ such that $\query(x_i) = 1$ iff $i \in T$. Therefore, by the definitions of $t$ and $(\alpha, 0)$-accuracy,
\[\left|a_\query - \frac{1}{n} \sum_{i = 1}^n q(x_i)t_i\right| = \left|a_\query - \frac{1}{n} \sum_{i \in T} t_i \right|\le \alpha \quad \text{and} \quad \left|a_\query - \frac{1}{n} \sum_{i \in T} s_i \right|\le \alpha,\]
so by the triangle inequality, $\frac{1}{n}\sum_{i \in T} (t_i - s_i) \le 2\alpha$. An identical argument shows that $\frac{1}{n}\sum_{i \in S} (s_i - t_i) \le 2\alpha$, proving that $t$ is an accurate reconstruction.
\end{proof}

\subsubsection{The $\Omega(1/\alpha^2)$ Lower Bound for $k$-Way Marginals}
We can now state in our terminology the lower bound of De from~\cite{De12} (building on~\cite{KasiviswanathanRuSmUl10}) showing that the inverse-quadratic dependence on $\alpha$ is necessary.
\jnote{Being forced to use some constant $k$ is a bit annoying since it adds lots of quantifiers.  It would be nicer if we could say that it holds for some fixed function $k = \omega_d(1).$}
\begin{theorem} [Restatement of~\cite{De12}] \label{thm:conjrec}
Let $k$ be any constant, $d \geq k$ be any integer, and let $\alpha \geq 1/d^{.499k}$ be a sufficiently small parameter\footnote{The constant $.499$ was chosen for simplicity, and can be replaced with any constant strictly smaller than $.5$.} (i.e.~bounded by an absolute constant).  There exists a constant $\beta = \beta(k) > 0$ such that for every $\alpha' > 0,$ there exists a database $D \in (\bits^d)^n$ with $n = \Omega_{\alpha', k}(1/\alpha^2)$ such that $D$ enables an $\alpha'$-reconstruction attack from $(\alpha, \beta)$-accurate answers to the $k$-way marginals $\kdconj{k}$.
\end{theorem}

Although the above theorem is a simple extension of De's lower bound, we sketch a proof for completeness, and refer the interested reader to~\cite{De12} for a more detailed analysis.
\begin{proof}[Proof Sketch]
The reconstruction attack uses the ``$\ell_1$-minimization'' algorithm, which is shown in Figure~\ref{fig:l1rec}.
\begin{figure}[h!]
\begin{framed}
\begin{algorithmic}
\STATE{\textbf{Input:} Queries $\cQ$, $D = (x_1,\dots,x_n) \in \bits^{n \times d},$ and $a = (a_q)_{q \in \cQ}.$}
\STATE{Let $t \in [0,1]^n$ be
$$
\argmin_{t \in [0,1]^n} \sum_{q \in \cQ} \left| a_q - \frac{1}{n} \sum_{i=1}^{n} q(x_i) t_i \right| 
$$}
\STATE{\textbf{Output:} $t$.}
\end{algorithmic}
\end{framed}
\vspace{-6mm}
\caption{The reconstruction adversary $\recadv_{\cQ}(\db, \answer)$.}
\label{fig:l1rec}
\end{figure}
To prove that the reconstruction attack succeeds, we will show that there exists a database $D = (x_1,\dots,x_n) \in \bits^{n \times d}$ such that for any $s \in [0,1]^n$, if $\answer$ satisfies
$$
\Prob{q \in \kdconj{k}}{\left|a_q - \frac{1}{n}\sum_{i=1}^{n} q(x_i) s_i \right| \leq \alpha} \geq 1-\beta,
$$
(i.e.~$\answer$ has $(\alpha, \beta)$-accurate answers)
then $\recadv_{\kdconj{k}}(\db, \answer)$ returns a vector $t$ such that $\| t - s \|_1 \leq \alpha' \cdot n.$  Henceforth we refer to such an $a$ simply as $(\alpha, \beta)$-accurate for $\kdconj{k}$ on $(D, s),$ as a shorthand.  The above guarantee must hold for suitable choices of $n, \beta,$ and $\alpha'$ to satisfy the theorem.

%For technical reasons, we will prove something slightly different.  Specifically, instead of $\kdconj{k}$, we will use only a subset $\cQ$ of $\kdconj{k}$, and run the reconstruction attack $\recadv_{\cQ}(\db, \answer).$  We will choose $\cQ$ so that $|\cQ| \geq \kappa |\kdconj{k}|$ for some constant $\kappa > 0.$  This will ensure that $(\alpha, \kappa \beta)$-accurate answers to $\kdconj{k}$ contain $(\alpha, \beta)$-accurate answers to $\cQ$ as a subset.

We will argue that the reconstruction succeeds in two steps.  First, we show that reconstruction succeeds if $D$ is ``nice.''  Second, we show that there exists ``nice'' $D$ that has the dimensions promised by the theorem.

To explain what we mean by a ``nice'' database $D$, for any $D = (x_1,\dots,x_n) \in \bits^{n \times d}$ and family of queries $\cQ$ on $\bits^d$, we define the matrix $M = M_{D, \cQ} \in \bits^{n \times |\cQ|},$ as
$
M(i, q) = q(x_i).
$
\begin{comment}
Observe that in this notation, we have
$
\frac{1}{n}(s^{\top} M)_q = \frac{1}{n} \sum_{i=1}^{n} q(x_i) s_i,
$
and $(\alpha, \beta)$-accuracy of $a$ for $\cQ$ on $(\db, s)$ means that $\Prob{q \in \cQ}{|a_q - (s^{\top} M)_q| \leq \alpha} \geq 1-\beta.$  We can rewrite the reconstruction attack as finding
$$
\argmin_{t \in [0,1]^n} \| a - t^{\top} M_{D, \cQ} \|_1.
$$
\end{comment}

De analyzes this reconstruction attack in terms of certain properties of the matrix $M.$  Before stating the conclusion, we will need to define the notion of a Euclidean section.  Informally, a matrix $M$ is a Euclidean Section if its rowspace\footnote{For a matrix $M$ with rows $M_1,\dots,M_n$, the rowspace of $M$ is $\set{a = \sum_{i=1}^{n} c_i M_i \mid c_1,\dots,c_n \in \mathbb{R}}$.} contains only vectors that are ``spread out.''
\begin{definition}[Euclidean Section]
A matrix $M \in \bits^{n \times m}$ is a \emph{$\delta$-Euclidean section} if for every vector $a$ in the rowspace of $M$ we have $\sqrt{m} \cdot \| a \|_2 \geq \| a \|_1 \geq \delta \sqrt{m} \cdot \| a \|_2.$
\end{definition}

\begin{lemma}[\cite{De12}] \label{lem:l1analysis}
Let $D$ be a database and $\cQ$ be a set of queries such that $M_{D, \cQ} \in \bits^{n \times |\cQ|}$ is a $\delta$-Euclidean section and the least singular value of $M_{D, \cQ}$ is $\sigma$.  Let $s \in [0,1]^n$ be arbitrary.  There exists $\beta = \beta(\delta) > 0$ such that if $a$ are $(\alpha, \beta)$-accurate answers for $\cQ$ on $(D,s)$, and $t = \recadv_{\cQ}(\db, \answer)$, then $t$ satisfies
$$
\| s - t \|_1 \leq \gamma n
$$
for $\gamma = O( \alpha \sqrt{n |\cQ|} / \sigma).$  The constant hidden in the $O(\cdot)$ notation depends only on $\delta.$
\end{lemma}

Thus, it suffices to find database $D$ such that the matrix $M_{D, \kdconj{k}}$ is a Euclidean section (for some fixed constant $\delta > 0$) and has no ``small'' singular values.  A result of Rudelson~\cite{Rudelson11} (strengthening that of Kasiviswanathan et al.~\cite{KasiviswanathanRuSmUl10}) guarantees that such a database exists.
\begin{lemma}[\cite{Rudelson11}] \label{lem:hadamardanalysis}
Let $k \in \N$ be any constant.  Let $d, n \in \N$ be such that $d^k \geq n \log n.$  Let $D \in \bits^{n \times d}$ be a uniform random matrix.  Then with probability at least $9/10$, the matrix $M_{D, \kdconj{k}}$ defined above has least singular value at least $\sigma = \Omega(d^{k/2})$ (where the hidden constant in the $\Omega(\cdot)$ may depend on $k$) and is a $\delta$-Euclidean section for some constant $\delta > 0$ that depends only on $k$.\footnote{Rudelson actually proves these statements about a related matrix $M_{D, \cQ}$ where $\cQ \subseteq \kdconj{k}.$  Since, for the $\cQ$ he considers, $|\cQ| \geq |\kdconj{k}|/(2k)^k,$ these statements can easily be seen to hold for the matrix $M_{D, \kdconj{k}}$ itself.  Specifically, adding this many more columns to the matrix $M_{D, \cQ}$ cannot decrease its least singular value (since $M_{D, \cQ}$ already has more columns than rows), and can only decrease the Euclidean section parameter $\delta$ by a factor of at most $(2k)^{k}.$}

In particular, there exists a database $D \in \bits^{n \times d}$ such that the Hadamard product $M$ satisfies the two properties above.
\end{lemma}

Using the above lemma, we can now complete the proof.  Fix any constant $k \in \N$.  Let $\alpha, d, n$ be any parameters such that $d \geq k$, $\alpha \geq 1/d^{.499k}$, and $d^{k} \geq n \log n$.  The precise value of $n$ will be determined later.  Let $D \in \bits^{n \times d}$ be the database promised by Lemma~\ref{lem:hadamardanalysis}.  Let $\beta = \beta(k) > 0$ be a parameter to be chosen later.  Let $\alpha' > 0$ be the desired accuracy of the reconstruction attack.

Now fix any $s \in [0,1]^n$ and let $a \in [0,1]^{|\kdconj{k}|}$ be $(\alpha, \beta)$-accurate answers to $\kdconj{k}$ on $(D,s)$.  Now, if we let $t = \cB_{\kdconj{k}}(D, a)$, by Lemma~\ref{lem:l1analysis}, provided that $\beta$ is smaller than some constant that depends only on $\delta$, which in turn depends only on $k$, we will have $\| s - t \|_1 \leq \gamma \cdot n$ for 
$$
\gamma = O\Bigg( \frac{\alpha \sqrt{n |\cQ|}}{\sigma} \Bigg) = O\Bigg( \frac{\alpha \sqrt{n} (d/k)^{k/2}}{d^{k/2}} \Bigg) = O(\alpha \sqrt{n}).
$$
Note that by Lemma~\ref{lem:l1analysis}, the hidden constant in the $O(\cdot)$ notation depends only on the parameter $\delta$ such that $M_{D, \kdconj{k}}$ is a $\delta$-Euclidean section.  By Lemma~\ref{lem:hadamardanalysis}, the parameter $\delta$ depends only on $k$.  Thus $\gamma = O(\alpha \sqrt{n})$ where the hidden constant depends only on $k$.  Now, we can choose $n = \Omega(1/\alpha^2)$ such that $\gamma \leq \alpha'.$ The hidden constant in the $\Omega(\cdot)$ will depend only on $k$ and $\alpha',$ as required by the theorem.  Note that, since we have assumed $\alpha \geq 1/d^{.499k}$, we have $n \log n = \tilde{O}(d^{.998k})$, and so we can define $n = \Omega_{k, \alpha'}(1/\alpha^2)$ while ensuring that $d^{k} \geq n \log n.$  Similarly, we required that $\beta$ is smaller than some constant that depends only on $\delta$, which in turn depends only on $k$.  Thus, we can set $\beta = \beta(k) > 0$ to be some sufficiently small constant depending only on $k$, as required by the theorem.  This completes our sketch of the proof.\end{proof}

\subsubsection{Putting Together the Lower Bound} \label{sec:puttingtogetherconj}
Now we show how to combine the various attacks to prove Theorem~\ref{thm:main1} in the introduction.  We obtain our lower bound by applying two rounds of composition. In the first round, we compose the reconstruction attack of Theorem~\ref{thm:conjrec} described above with the re-identifiable distribution for $1$-way marginals. We then take the resulting re-identifiable distribution and apply a second round of composition using the reconstruction attack based on the VC-dimension of $k$-way marginals.

We remark that it is necessary to apply the two rounds of composition in this order. In particular, we cannot prove Theorem~\ref{thm:main2} by composing first with the VC-dimension-based reconstruction attack.  Our composition theorem requires a re-identifiable distribution from $(\alpha, \beta)$-accurate answers for $\beta > 0$, whereas the reconstruction attack described in Lemma~\ref{lem:vcrec} requires $(\alpha, 0)$-accurate answers, and the reconstruction can fail if some queries have error much larger than $\alpha$.  The resulting re-identifiable distribution obtained from composing with this reconstruction attack will also require $(\alpha, 0)$-accurate answers, and thus cannot be composed further.

This limitation of Lemma~\ref{lem:vcrec} is inherent, because a sample complexity upper bound of $\tilde{O}(\sqrt{d} / \alpha^2)$ can be achieved for answering any family of queries $\cQ$ with $(\alpha, \beta)$-accuracy (for any constant $\beta > 0$).  Notice that this sample complexity is independent of $\mathit{VC}(\cQ)$.

We can now formally state and prove our sample-complexity lower bound for $k$-way marginals, thereby establishing Theorem~\ref{thm:main2} in the introduction.
\jnote{This magic parameter $\ell$ is a bit annoying.}
\begin{theorem} \label{thm:conjlb}
For every constant $\ell \in \N$, every $k, d \in \N$, $\ell + 2 \le k \le d$, and every sufficiently small (i.e.~bounded by an absolute constant) $\alpha \geq 1/d^{.499\ell}$, there is an
$$
n = n(k, d, \alpha) = \tilde{\Omega}\Bigg( \frac{k \sqrt{d}}{\alpha^2} \Bigg)
$$
such that there exists a distribution on $\rows$-row databases $D \in (\bits^d)^\rows$ that is $(1/2, o(1/\rows))$-re-identifiable from $(\alpha, 0)$-accurate answers to the $k$-way marginals $\kdconj{k}$.

\begin{comment}
There exists a universal constant $\alpha_0 > 0$ such that for every $k,d \in \N$, $k \le d$, there is an
$$
n = n(k,d) = \tilde{\Omega}(k\sqrt{d})
$$
such that there exists a distribution on $\rows$-row databases $D \in (\bits^d)^\rows$ that is $(1/3, o(1/\rows))$-re-identifiable from $(\alpha_0, 0)$-accurate answers to the $k$-way marginals $\kdconj{k}$.
\end{comment}
\end{theorem}

\begin{proof}
We begin with the following two attacks:
\begin{enumerate}
\item By combining Theorem~\ref{thm:fpctolb} and Theorem~\ref{thm:rfpc}, there exists a distribution on databases $\privdb \in (\bits^{d/3})^{n_{d}}$ that is $(\gamma = 1/6, \sec = o(1/n_{d}n_{\alpha}n_{k}))$-re-identifiable from $(6c\alpha' = 1/3, 2/c = 1/75)$ accurate answers to the $1$-way marginals $\mathcal{M}_{1,d/3}$ for $n_d = \tilde{\Omega}(\sqrt{d} / \log (n_{d} n_{\alpha} n_{k})).$  Here $n_{\alpha}$ and $n_{k}$ are set below (the subscript corresponds to the primary parameter that each of the $n$'s will depend on).
\item By Theorem~\ref{thm:conjrec} (with $\alpha' = 1/2700$ and $k = \ell$), there is a constant $\beta > 0$ such that for any $7200 \alpha / \beta  \geq 1/d^{.499\ell}$ there exists a database $\recdb \in (\bits^{d/3})^{n_\alpha}$, for $n_{\alpha} = \tilde{\Omega}(1/\alpha^2)$ that enables a $(1/2700)$-reconstruction attack from $(7200 \alpha / \beta, \beta)$-accurate answers to $\mathcal{M}_{\ell, d/3}.$
\end{enumerate}
Applying Theorem~\ref{thm:composition} (with parameter $c=150$), we obtain item 1' below.  We then bring in another reconstruction attack for the composition theorem.
\begin{enumerate}[label=\arabic*'.]
\item There exists a distribution on databases in $(\bits^{2d/3})^{n_{d} n_{\alpha}}$ that is $(1/3, o(1/n_{d} n_{\alpha} n_{k}))$-re-identifiable from $(6c'\alpha' = 7200 \alpha / \beta, 2/c' = \beta/150)$-accurate answers to $\mathcal{M}_{\ell, d/3} \land \mathcal{M}_{1,d/3} \subset \mathcal{M}_{\ell + 1, 2d/3}$ (By applying Theorem~\ref{thm:composition} to 1 and 2 above.)
\item By Lemma~\ref{lem:vcrec} and Fact~\ref{fact:vcconj}, there exists a database $\recdb \in (\bits^{d/3})^{n_{k}}$, for $n_{k} = k - \ell - 1$, that enables an $(\alpha' = 4\alpha)$-reconstruction attack from $(\alpha, 0)$-accurate answers to the $(k-\ell-1)$-way marginals $\mathcal{M}_{k-\ell-1,d/3}$.  Note that $(k - \ell - 1) \geq 1$, since we have assumed $k \geq \ell + 2$.
\end{enumerate}
We can then apply Theorem~\ref{thm:composition} to 1' and 2' (with parameter $c' = 300/\beta$).  Thereby we obtain a distribution $\cD$ on databases $D \in (\bits^{d/3} \times \bits^{d/3} \times \bits^{d/3})^{n_{d} n_{\alpha} n_{k}}$ that is $(1/2, \sec)$-re-identifiable from $(\alpha, 0)$-accurate answers to $\mathcal{M}_{k-\ell-1,d/3} \land \mathcal{M}_{\ell, d/3} \land \mathcal{M}_{1, d/3} \subset \mathcal{M}_{k,d}.$

To complete the theorem, first note that $(\alpha, 0)$-accurate answers to $\mathcal{M}_{k,d}$ imply $(\alpha, 0)$-accurate answers to any subset of $\mathcal{M}_{k,d}$.  So our lower bound for the subset $\mathcal{M}_{k-\ell-1,d/3} \land \mathcal{M}_{\ell, d/3} \land \mathcal{M}_{1, d/3}$ is sufficient to obtain the desired lower bound.  Finally, note that
$$
n = n_{d} n_{\alpha} n_{k} = \tilde{\Omega}\left( \frac{k \sqrt{d} }{ \alpha^2} \right),
$$
as desired.  This completes the proof.
\end{proof}
\begin{comment}
\begin{proof}
The previous results will imply the existence of two privacy attacks, and we obtain the result by applying the composition theorem (Theorem~\ref{thm:composition}) to them.
\begin{enumerate}
\item By combining Theorem~\ref{thm:fpctolb} and Theorem~\ref{thm:rfpc}, there exists a distribution on databases $\privdb \in (\bits^{d/2})^{\privrows}$ that is $(\gamma = 1/6, \sec = o(1/\privrows k))$-re-identifiable from $(6c\alpha' = 1/3, 2/c = 1/75)$ accurate answers to the $1$-way marginals $\mathcal{M}_{1,d/2}$ for $\privrows = \tilde{\Omega}(\sqrt{d} / \log (dk))$,
\item By Lemma~\ref{lem:vcrec} and Fact~\ref{fact:vcconj}, there exists a database $\recdb \in (\bits^{d/2})^{k-1}$ that enables a $(\alpha' = 4\alpha_0)$-reconstruction attack from $(\alpha = \alpha_0, \beta = 0)$-accurate answers to the $(k-1)$-way marginals $\mathcal{M}_{k-1,d/2}$ for any $\alpha_0$.
\end{enumerate}
By applying Theorem~\ref{thm:composition} (with parameter $c = 150$) to these two distributions, we obtain a new distribution on $(\bits^{d})^{\privrows (k-1)}$ that is $(1/3, o(\privrows k))$-re-identifiable from $(\alpha_0, 0)$-accurate answers to $\mathcal{M}_{k-1,d/2} \land \mathcal{M}_{1,d/2}$ on $\bits^{d/2} \times \bits^{d/2}$. Note that this family of queries is a subset of $\kdconj{k}$ on $\bits^d$, but $(\alpha_0,0)$-accuracy for $\kdconj{k}$, (accuracy for all queries in in $\kdconj{k}$) implies $(\alpha_0,0)$-accuracy for any subset of $\kdconj{k}$.
\end{proof}
\end{comment}

Using the composition Theorem \ref{thm:composition-fpc} in place of Theorem \ref{thm:composition}, we obtain a version of Theorem \ref{thm:conjlb} in the language of generalized fingerprinting codes.

\begin{theorem} 
For every constant $\ell \in \N$, every $k, d \in \N$, $\ell + 2 \le k \le d$, and every sufficiently small (i.e.~bounded by an absolute constant) $\alpha \geq 1/d^{.499\ell}$, there is an
$$
n = n(k, d, \alpha) = \tilde{\Omega}\Bigg( \frac{k \sqrt{d}}{\alpha^2} \Bigg)
$$
such that there exists a $(\rows, \kdconj{k})$-generalized fingerprinting code with security $(1/2, o(1/\rows))$ for $(\alpha, 0)$-accuracy.
\end{theorem}

\subsubsection{A Tight Lower Bound for 2-Way Marginals}
Theorem~\ref{thm:conjlb} does not give any non-trivial lower bound for $2$-way marginals.  Intuitively, the problem is that the proof uses two rounds of composition, and thus if we try to instantiate the proof for $2$-way marginals, one of the three lower bounds being composed will have to be trivial (i.e.~will be a lower bound for $0$-way marginals).  However, a simple modification of the proof yields a tight lower bound for $2$-way marginals that holds even for $(\alpha, \beta)$-accuracy.

\begin{theorem}
For every $k, d \in \N$, and every sufficiently small (i.e.~bounded by an absolute constant) $\alpha \geq 1/d^{.499}$, there is a constant $\beta > 0$ and an
$$
n = n(d, \alpha) = \tilde{\Omega}\big( \sqrt{d} / \alpha^2 \big)
$$
such that there exists a distribution on $\rows$-row databases $D \in (\bits^d)^\rows$ that is $(1/2, o(1/\rows))$-re-identifiable from $(\alpha, \beta)$-accurate answers to the $2$-way marginals $\kdconj{2}$.
\end{theorem}

\begin{proof}
We begin with the following two attacks:
\begin{enumerate}
\item By combining Theorem~\ref{thm:fpctolb} and Theorem~\ref{thm:rfpc}, there exists a distribution on databases $\privdb \in (\bits^{d/2})^{n_{d}}$ that is $(\gamma = 1/6, \sec = o(1/n_{d}n_{\alpha}))$-re-identifiable from $(6c\alpha' = 1/3, 2/c = 1/75)$ accurate answers to the $1$-way marginals $\mathcal{M}_{1,d/2}$ for $n_d = \tilde{\Omega}(\sqrt{d} / \log (n_{d} n_{\alpha})).$  $n_{\alpha}$ is set below.
\item By Theorem~\ref{thm:conjrec} (with $\alpha' = 1/2700$ and $k = 1$), there is a constant $\beta > 0$ such that for any $2700\alpha / \beta \geq 1/d^{.499}$ there exists a database $\recdb \in (\bits^{d/2})^{n_\alpha}$, for $n_{\alpha} = \tilde{\Omega}(1/\alpha^2)$ that enables a $(1/2700)$-reconstruction attack from $(2700 \alpha, 600\beta)$-accurate answers to $\mathcal{M}_{1, d/2}.$
\end{enumerate}
Applying Theorem~\ref{thm:composition} (with parameter $c=150$), we obtain the following: 
There exists a distribution on databases in $(\bits^{d})^{n_{d} n_{\alpha}}$ that is $(1/3, o(1/n_{d} n_{\alpha}))$-re-identifiable from $(\alpha, 4\beta)$-accurate answers to $\mathcal{M}_{1, d/2} \land \mathcal{M}_{1,d/2} \subset \mathcal{M}_{2,d}$.

To complete the theorem, note that $\mathcal{M}_{1,d/2} \land \mathcal{M}_{1,d/2}$ contains exactly $1/4$ of all the queries in $\mathcal{M}_{2,d}$, so $(\alpha, \beta)$-accurate answers to $\mathcal{M}_{2,d}$ contain $(\alpha, 4\beta)$-accurate answers to the subset $\mathcal{M}_{1,d/2} \land \mathcal{M}_{1,d/2}$.  So our lower bound for the subset $\mathcal{M}_{1,d/2} \land \mathcal{M}_{1, d/2}$ is sufficient to obtain the desired lower bound.  Finally, note that
$$
n = n_{d} n_{\alpha} = \tilde{\Omega}\big(\sqrt{d} / \alpha^2 \big),
$$
as desired.  This completes the proof.
\end{proof}

\subsection{Lower Bounds for Arbitrary Queries} \label{sec:lbarbitrary}
\begin{comment}
\begin{figure}[ht]
\centering
\includegraphics[scale=.45]{productdb.png}
\vspace{-20mm}
\caption{A pictorial representation of a product database $D^* = D \times (D'_1,\dots,D'_n)$.}
\end{figure}
\end{comment}

Using our composition theorem, we can also prove a nearly-optimal sample complexity lower bound as a function of the $|\cQ|, d,$ and $\alpha$ and establish Theorem~\ref{thm:main2} in the introduction.  

As was the case in the previous section, the main result of this section will follow from three lower bounds: the $\tilde{\Omega}(\sqrt{d})$ lower bound for $1$-way marginals and the $\Omega(\vcdim(\queryset))$ bound that we have already discussed, a lower bound of $\Omega(1/\alpha^2)$ for worst-case queries, which is a simple variant of the seminal reconstruction attack of Dinur and Nissim~\cite{DinurNi03}, and related attacks such as~\cite{DworkMcTa07,DworkYe08}.  Although we already proved a $\Omega(1/\alpha^2)$ lower bound for the simpler family of $k$-way marginals in the previous section, the lower bound in this section will hold for a much wider range of $\alpha$ than what is known for $k$-way marginals (roughly $\alpha \geq 2^{-d}$ for arbitrary queries, whereas for $k$-way marginals we require $\alpha \geq 1/d^{\ell}$ for some constant $\ell$).

\subsubsection{The $\Omega(1/\alpha^2)$ Lower Bound for Arbitrary Queries}
Roughly, the results of~\cite{DinurNi03} can be interpreted in our framework as showing that there is an $\Omega(1/\alpha^2)$-row database that enables a $1/100$-reconstruction attack from $(\alpha, 0)$-accurate answers to some family of queries $\queryset$, but only when the vector to be reconstructed is Boolean. That is, the attack reconstructs a bit vector accurately provided that every query in $\queryset$ is answered correctly. Dwork et al. \cite{DworkMcTa07,DworkYe08} generalized this attack to only require $(\alpha, \beta)$-accuracy for some constant $\beta > 0$, and we will make use of this extension (although we do not require computational efficiency, which was a focus of those works).  Finally, we need an extension to the case of fractional vectors $s \in [0,1]^n$, instead of Boolean vectors $s \in \{0, 1\}^n$.

The extension is fairly simple and the proof follows the same outline of the original reconstruction attack from~\cite{DinurNi03}.  We are given accurate answers to queries in $\queryset$, which we interpret as approximate ``subset-sums'' of the vector $s \in [0, 1]^n$ that we wish to reconstruct.  The reconstruction attack will output any vector $t$ from a discretization $\multsofm^n$ of the unit interval that is ``consistent'' with these subset-sums.  The main lemma we need is an ``elimination lemma'' that says that if $\|t -s\|_1$ is sufficiently large, then for a random subset $T \subseteq [n]$,
$$
\frac{1}{n}\left| \sum_{i \in T} (t_i - s_i) \right| > 3\alpha
$$
with suitable large constant probability.
For $m=1$ this lemma can be established via combinatorial arguments, whereas for the $m > 1$ case we establish it via the Berry-Ess\'een Theorem.  The lemma is used to argue that for every $t$ that is sufficiently far from $s$, a large fraction of the subset-sum queries will witness the fact that $t$ is far from $s$, and ensure that $t$ is not chosen as the output.

First we state and prove the lemma that we just described, and then we will verify that it indeed leads to a reconstruction attack.
\begin{lemma} \label{lem:elimination}
Let $\kappa > 0$ be a constant, let $\alpha > 0$ be a parameter with $\alpha \le \kappa^2/240$, and let $n = 1/576\kappa^2\alpha^2$.  Then for every $r \in [-1,1]^n$ such that $\frac{1}{n} \sum_{i=1}^{n} |r_i| > \kappa$, and a randomly chosen $q \subseteq [n]$,
$$
\Prob{q \subseteq [n]}{\left|\frac{1}{n} \sum_{i \in q} r_i \right| > 3\alpha} \geq \frac{3}{5}.
$$
\end{lemma}
\begin{proof} [Proof of Lemma~\ref{lem:elimination}]

Let $r$ be as in the statement of the lemma. Define a random variable
\begin{equation*}
Q_i =
\begin{cases}
r_i/2 & \textrm{if $i \in q$} \\
-r_i/2 & \textrm{if $i \notin q$}
\end{cases}
\end{equation*}
By construction, we have
$$
\frac{1}{n} \sum_{i \in q} r_i = \frac{1}{n}\sum_{i=1}^n \left(Q_i + \frac{r_i}{2}\right),
$$
Thus,
$$
\left| \frac{1}{n} \sum_{i \in q} r_i\right| \leq 3 \alpha \Longleftrightarrow \sum_{i=1}^n Q_i \in \left[ -3\alpha n - \frac{1}{2} \sum_{i=1}^n r_i, 3\alpha n - \frac{1}{2} \sum_{i=1}^n r_i \right].
$$
The condition on the right-hand side says that $\sum_{i} Q_i$ is in some interval of width $6\alpha n$.  Since the random variables $Q_i$ are independent, as $q$ is a randomly chosen subset, we will use the Berry-Ess\'een Theorem (Theorem~\ref{thm:berryesseen}) to conclude that this sum does not fall in any interval of this width too often. Establishing the next claim suffices to prove Lemma~\ref{lem:elimination}.

\begin{claim} \label{clm:elimination}
For any interval $I \subseteq \R$ of width $6 \alpha n$,
$$
\prob{\sum_{i} Q_i \not\in I} \geq \frac{3}{5}.
$$
\end{claim}
\begin{proof}[Proof of Claim~\ref{clm:elimination}]
We use the Berry-Ess\'een Theorem:

\begin{theorem}[Berry-Ess\'een Theorem] \label{thm:berryesseen}
Let $X_1,\dots,X_n$ be independent random variables such that $\ex{X_i} = 0$, $\sum_{i} \ex{X_i^2} = \sigma^2$, and $\sum_{i} \ex{|X_i|^3} = \gamma$.  Let $X = (X_1 + \dots + X_n) / \sigma$ and let $Y$ be a normal random variable with mean $0$ and variance $1$.  Then,
$$
\sup_{z, z' \in \R} \left| \prob{X \in [z,z']} - \prob{Y \in [z,z']} \right| \leq \frac{2\gamma}{\sigma^3}.
$$
\end{theorem}

In order to apply Theorem~\ref{thm:berryesseen} with $X_i = Q_i$, we need to analyze the moments of the random variables $Q_i$.
The following bounds can be verified from the definition of $Q_i$ and the assumption that $\|r\|_1 \geq \kappa n$.
\begin{enumerate}
\item $\ex{Q_i} = 0$.
\item $\sigma^2 = \sum_{i} \ex{Q_i^2} \ge \kappa^2n/4$.
\item $\gamma =  \sum_{i} \ex{|Q_i|^3} \leq \frac{n}{8}$.
\end{enumerate}
Thus, by Theorem~\ref{thm:berryesseen} we have
$$
\sup_{z,z' \in \R} \left| \prob{\frac{Q_1 + \dots + Q_n}{\sigma} \in [z,z']} - \prob{Y \in [z,z']}  \right| \leq \frac{2\gamma}{\sigma^3} \le \frac{2}{\kappa^3\sqrt{n}} \leq \frac{1}{5},
$$
where the final inequality holds because $n = 1/576\kappa^2\alpha^2 \geq 100/\kappa^6$.  It can be verified that for a standard normal random variable $Y$, and every interval $I \subset \R$ of width $1/2$, it holds that $\prob{Y \not\in I} \geq 4/5$.  Thus, for every such interval $I$,
\begin{align*}
&\prob{ \frac{Q_1 + \dots + Q_n}{\sigma} \not\in I} \geq \frac{4}{5} - \frac{1}{5}\\
\Longrightarrow{} &\prob{ Q_1 + \dots + Q_n \not\in \sigma I} \geq \frac{3}{5}
\end{align*}
where $\sigma I$ is an interval of width $\sigma/2$.  Thus we have obtained that $\sum_{i} Q_i$ falls outside of any interval of width $\sigma / 2$ with probability at least $3/5$.  In order to establish the claim, we simply observe that
$$
\frac{\sigma}{2} \ge \frac{ \kappa\sqrt{n}}{4} \ge 6 \alpha n 
$$
when $n = 1/576\kappa^2\alpha^2$. Thus, the probability of falling outside an interval of width $6 \alpha n $ is only larger than the probability of falling outside an interval of width $\sigma/2$.
\end{proof}

Establishing Claim~\ref{clm:elimination} completes the proof of Lemma~\ref{lem:elimination}.
\end{proof}

\begin{theorem} \label{thm:arbitraryrec}
Let $\alpha' \in (0, 1]$ be a constant, let $\alpha > 0$ be a parameter with $\alpha \le (\alpha')^2/960$, and let $n = 1/144(\alpha')^2\alpha^2$.  For any data universe $\univ = \{x_1, \dots, x_n\}$ of size $n$, there is a set of counting queries $\cQ$ over $\univ$ of size at most $O(n \log (1/\alpha))$ such that the database $D = (x_1, \dots, x_n)$ enables a $\alpha'$-reconstruction attack from $(\alpha, 1/3)$-accurate answers to $\cQ$.
\end{theorem}

\begin{proof}
First we will give a reconstruction algorithm $\adv$ for an arbitrary family of queries.  We will then show that for a random set of queries $\cQ$ of the appropriate size, the reconstruction attack succeeds for every $s \in [0,1]^n$ with non-zero probability, which implies that there exists a set of queries satisfying the conclusion of the theorem. We will use the shorthand
\[\corr{q}{s} = \frac{1}{n} \sum_{i=1}^n q(x_i)s_i\]
for vectors $s \in [0,1]^n$.
\begin{figure}[ht]
\begin{framed}
\begin{algorithmic}
\STATE{\textbf{Input:}  Queries $\cQ$, and $(a_q)_{q \in \cQ}$ that are $(\alpha, 1/3)$-accurate for $s$.}
\STATE{Let $m = \lceil\frac{1}{\alpha}\rceil$}
\STATE{Find any $t \in \multsofm^n$ such that
$$
\Prob{q \getsr \cQ}{ | \corr{q}{t} - a_q | < 2\alpha} > \frac{5}{6}.
$$ }
\STATE{\textbf{Output:} $t$.}
\end{algorithmic}
\end{framed}
\vspace{-6mm}
\caption{The reconstruction adversary $\recadv$.}
\end{figure}
\label{fig:arbqueries}

In order to show that the reconstruction attack $\adv$ from Figure~\ref{fig:arbqueries} succeeds, we must show that
$
\frac{1}{n} \sum_{i=1}^n |t_i - s_i| \leq \alpha'.
$
Let $s \in [0, 1]^n$, and let $s' \in \multsofm^n$ be the vector obtained by rounding each entry of $s$ to the nearest $1/m$. Then 
$$
\frac{1}{n} \sum_{i=1}^n |s'_i - s_i| \leq \frac{\alpha}{2} \le \frac{\alpha'}{2},
$$
so it is enough to show that the reconstruction attack outputs a vector close to $s'$. Observe that the vector $s'$ itself satisfies
$$|\corr{q}{s'} - a_q| \le |\corr{q}{s} - a_q| + |\corr{q}{s' - s}| \le 2\alpha$$
for any subset-sum query $q$, so the reconstruction attack always finds some vector $t$. To show that the reconstruction is successful, fix any $t \in \multsofm^n$ such that 
$
\frac{1}{n} \sum_{i=1}^{n} | t_i -s'_i | > \frac{\alpha'}{2}.
$  
If we write $r =  s' - t \in \{-1, \dots,-1/m,0,1/m,\dots, 1\}^n$, then 
$
\frac{1}{n} \sum_{i=1}^{n} |r_i| > \frac{\alpha'}{2}
$  and
$
\corr{q}{r} = \corr{q}{t} - \corr{q}{s'}
$.
In order to show that no $t$ that is far from $s'$ can be output by $\recadv$, we will show that for any $r \in \{-1,\dots,-1/m,0,1/m,\dots,1\}$ with $\frac{1}{n}\sum_{i=1}^{n} |r| > \frac{\alpha'}{2}$,
$$
\Prob{q \getsr \cQ}{|\corr{q}{r}| > 3\alpha} \geq \frac{1}{2}.
$$

To prove this, we first observe by Lemma~\ref{lem:elimination} (setting $\kappa = \frac{1}{2}\alpha'$) that for a randomly chosen query $q$ defined on $\univ$,
$$
\Prob{q}{|\corr{q}{r}| > 3 \alpha} \geq \frac{3}{5}.
$$
The lemma applies because $\corr{q}{r} = \frac{1}{n} \sum_{i=1}^{n} q(x_i) r_i$ is a random subset-sum of the entries of $r$.  

Next, we apply a concentration bound to show that if the set $\queryset$ of queries is a sufficiently large random set, then for every vector $r$ the fraction of queries for which $|\corr{q}{r}|$ is large will be close to the expected number, which we have just established is at least $3|\cQ|/5$.  We use the following version of the Chernoff bound. %to show that for a large enough set of random queries, a significant fraction satisfy $|q(r)| > 2\alpha$ for every $r$ with $\|r\|_1 > \alpha'mn$.

\begin{theorem}[Chernoff Bound] \label{thm:chernoff}
Let $X_1, \dots, X_N$ be a sequence of independent random variables taking values in $[0, 1]$. If $X = \sum_{i=1}^N X_i$ and $\mu = \ex{X}$, then
\[\prob{X \le \mu - \eps} \le e^{-2\eps^2/N}.\]
\end{theorem}

Consider a set of randomly chosen queries $\cQ$.  By the above, we have that for every $r \in \{-1, \dots,-1/m,0,1/m,\dots, 1\}^n$ such that $\frac{1}{n} \sum_{i=1}^{n} |r| > \frac{\alpha'}{2}$, 
$$
\Ex{\cQ}{\left|\left\{ q \in \cQ \mid |\corr{q}{r}| > 3\alpha \right\}\right|} \geq \frac{3|\cQ|}{5}.
$$
Since the queries are chosen independently, by the Chernoff bound we have
\begin{align*}
\Prob{\cQ}{\left|\left\{ q \in \cQ \mid |\corr{q}{r}| > 3 \alpha \right\}\right| \leq \frac{|\cQ|}{2}} \leq e^{-|\cQ|/50}.
\end{align*}
Thus, we can choose $|\cQ| = O(n \log m)$ to obtain
\begin{align*}
&\Prob{\cQ}{ \exists r \in \set{-1,\dots,-1/m,0,1/m,\dots,1}^n, \; \frac{1}{n} \sum_{i=1}^{n} |r_i| > \frac{\alpha'}{2}, \quad \left|\left\{ q \in \cQ \mid |\corr{q}{y}| > 3 \alpha  \right\}\right| \leq \frac{|\cQ|}{2}} \\
&\quad < (2m+1)^n e^{-|\cQ|/50} \leq \frac{1}{2}.
\end{align*}

Thus, we have established that there exists a family of queries $\cQ$ such that for every $s, t$ such that $\frac{1}{n} \sum_{i=1}^{n} |t_i - s_i| > \alpha'$,
$$
\Prob{q \getsr \cQ}{|\corr{q}{s} - \corr{q}{t}| > 3\alpha} \geq \frac{1}{2}.
$$
Moreover, by $(\alpha, 1/3)$-accuracy, we have
$$
\Prob{q \getsr \cQ}{|a_q - \corr{q}{s}| > \alpha} \leq \frac{1}{3}.
$$
Applying a triangle inequality, we can conclude
$$
\Prob{q \getsr \cQ}{|a_q - \corr{q}{t}| > 2\alpha} \geq \frac{1}{2} - \frac{1}{3} \geq \frac{1}{6},
$$
which implies that $t$ cannot be the output of $\recadv$.  This completes the proof.

\end{proof}

\subsubsection{Putting Together the Lower Bound}

Now we show how to combine the various attacks to prove Theorem~\ref{thm:main1} in the introduction.  We obtain our lower bound by applying two rounds of composition. In the first round, we compose the reconstruction attack described above with the re-identifiable distribution for $1$-way marginals. We then take the resulting re-identifiable distribution and apply a second round of composition using the reconstruction attack for query families of high VC-dimension.

Just like our lower bound for $k$-way marginal queries, we remark that it is necessary to apply the two rounds of composition in this order.  See Section~\ref{sec:puttingtogetherconj} for a discussion of this issue.

\begin{theorem} \label{lem:arbitrarylb}
For all $d \in \N$, all sufficiently small (i.e.~bounded by an absolute constant) $\alpha > 2^{-d/6}$, and all $h \leq 2^{d/3}$, there exists a family of queries $\cQ$ of size $O(h d \log (1/\alpha)/ \alpha^2)$ and an
$$
n = n(h,d,\alpha) = \tilde{\Omega}\Bigg( \frac{\sqrt{d} \log h}{\alpha^2}\Bigg)
$$
such that there exists a distribution on $\rows$-row databases $D \in (\bits^{d})^{n}$ that is $(1/2, o(1/\rows))$-re-identifiable from $(\alpha, 0)$-accurate answers to $\cQ$.
\end{theorem}

\begin{proof}
We begin with the following two attacks:
\begin{enumerate}
\item By Theorem~\ref{thm:fpctolb} and Theorem~\ref{thm:rfpc}, there exists a distribution on databases in $(\bits^{d/3})^{m}$ that is $(1/6, o(1/m\ell\log h))$-re-identifiable from $(1/3, 1/75)$ accurate answers to $\mathcal{M}_{1,d/3}$ for $m =  \tilde{\Omega}(\sqrt{d} / \log (m\ell\log h))$.  Here $\ell$ and $h$ are parameters we set below.
\item For some $\ell = \Omega(1/\alpha^2)$, by Theorem~\ref{thm:arbitraryrec}, there exists a database $\recdb \in (\bits^{d/3})^{\ell}$ that enables a $\alpha'$-reconstruction attack from $(6c'\alpha, 1/3)$-accurate answers to some $\queryset_{rec}$ of size \linebreak$O((\log (1/\alpha))/\alpha^2)$. Here $\alpha'$ is a constant with $6c\alpha' = 1/3$ for a composition parameter $c$ set below, and $c'$ is a constant composition parameter set when we apply the second round of composition.
\end{enumerate}
Applying Theorem~\ref{thm:composition} (with parameter $c=150$), we obtain item 1' below.  We then bring in another reconstruction attack for the composition theorem.
\begin{enumerate}[label=\arabic*'.]
\item There exists a distribution on databases in $(\bits^{2d/3})^{m\ell}$ that is $(1/3, o(1/m\ell\log h))$-re-identifiable from $(6c' \alpha, 1/450)$-accurate answers to $\queryset_{rec} \land \mathcal{M}_{1,d/3}$ (By applying Theorem~\ref{thm:composition} to 1 and 2 above.)
\item By Lemma \ref{lem:vcrec}, there exists a database $\db \in (\bits^{d/3})^{\log h}$ that enables a $(4\alpha)$-reconstruction attack from $(\alpha, 0)$-accurate answers to some $\queryset_{vc}$ of size $h$.  (In particular, the family of queries can be all $(\log h)$-way marginals on the first $\log h$ bits of the data universe items.)
\end{enumerate}
We can then apply Theorem~\ref{thm:composition} to 1' and 2' (with parameter $c' = 900$).  Thereby we obtain a distribution $\cD$ on databases $D \in (\bits^{d/3} \times \bits^{d/3} \times \bits^{d/3})^{m \ell \log h}$ that is $(1/2, \sec)$-re-identifiable from $(\alpha, 0)$-accurate answers to $\cQ = \cQ_{vc} \land \cQ_{rec} \land \mathcal{M}_{1,d/3}$.

To complete the theorem we first set 
$$
n = m \ell \log h = \tilde{\Omega}(\sqrt{d} \log h/ \alpha^2).
$$
and then observe that 
$$
|\cQ_{vc} \land \cQ_{rec} \land \mathcal{M}_{1,d/3}| = h \cdot O(\ell \log (1/\alpha)/\alpha^2) \cdot d/3 =  O(h d \log (1/\alpha) / \alpha^2).
$$
This completes the proof.
\end{proof}

Again, Theorem \label{lem:arbitrarylb} has a corresponding statement in terms of generalized fingerprinting codes.

\begin{theorem} \label{lem:arbitrarylb}
For all $d \in \N$, all sufficiently small (i.e. bounded by an absolute constant) $\alpha > 2^{-d/6}$, and all $h \leq 2^{d/3}$, there exists a family of queries $\cQ$ of size $O(h d \log (1/\alpha)/ \alpha^2)$ and an
$$
n = n(h,d,\alpha) = \tilde{\Omega}\Bigg( \frac{\sqrt{d} \log h}{\alpha^2}\Bigg)
$$
such that there exists a $(\rows, \cQ)$-generalized fingerprinting code with security $(1/2, o(1/\rows))$ for $(\alpha, 0)$-accuracy.
\end{theorem}

\section{Constructing Error-Robust Fingerprinting Codes} \label{chap:rfpc}

\newcommand{\marked}{m}
\newcommand{\interword}{\tilde{\codeword}}
\newcommand{\diffword}{\overline{\codeword}}

In this section, we show how to construct fingerprinting codes that are robust to a constant fraction of errors, which will establish Theorem~\ref{thm:rfpc}.  Our codes are based on the fingerprinting code of Tardos~\cite{Tardos08}, which has a nearly optimal number of users, but is not robust to any constant fraction of errors.  The number of users in our code is only a constant factor smaller than that of Tardos, and thus our codes also have a nearly optimal number of users.

To motivate our approach, it is useful to see why the Tardos code (and all other fingerprinting codes we are aware of) are not robust to a constant fraction of errors.  The reason is that the the only way to introduce an error is to put a $0$ in a column containing only $1$'s or vice versa (recall that the set of codewords, $\codebook \in \bits^{\users \times \len}$, can be viewed as an $\users \times \len$ matrix).  We call such columns ``marked columns.''  Thus, if the adversary is allowed to introduce $\geq m$ errors where $m$ is the number of marked columns then he can simply ignore the codewords and output either the all-$0$ or all-$1$ codeword, which cannot be traced.  Thus, in order to tolerate a $\beta$ fraction of errors, it is necessary that $m \geq \beta d$ where $d$ is the length of the codeword, and this is not satisfied by any construction we know of (when $\beta > 0$ is a constant).  However, Tardos' construction can be shown to remain secure if the adversary is allowed to introduce $\beta m$ errors, rather than $\beta d$ errors, for some constant $\beta > 0$.  We demonstrate this formally in Section~\ref{sec:weakrobusttardos}.  In addition, we show how to take a fingerprinting code that tolerates $\beta m$ errors and modify it so that it can tolerate about $\beta d / 3$ errors.  This reduction is formalized in Section~\ref{sec:weaktostrong}.  Combining these two results will give us a robust fingerprinting code.

We remark that prior work~\cite{BonehNa08, BonehKiMo10} has shown how to construct fingerprinting codes satisfying a weaker robustness property.  Specifically, their codes allow the adversary to introduce a special ``?'' symbol in a large fraction of coordinates, but still require that any coordinate that is not a ``?'' satisfies the feasibility constraint.

Before proceeding with the construction and analysis, we restate some terminology and notation from Section~\ref{sec:1way}.  Recall that a fingerprinting code is a pair of algorithms $(\gen, \trace)$, where $\gen$ specifies a distribution over codebooks $\codebook \in \bits^{\users \times \len}$ consiting of $\users$ codewords $(\codewordi{1},\dots,\codewordi{\users})$, and $\trace(\codebook, \pirateword)$ either outputs the identity $i \in [\users]$ of an accused user or outputs $\bot$.  Recall that $\gen$ and $\trace$ share a common state.  For a coalition $S \subseteq [\users]$, we write $\codebookS{S} \in \bits^{|S| \times \len}$ to denote the subset of codewords belonging to users in $S$.

Every codebook $\codebook$, coalition $S$, and robustness parameter $\rob \in [0,1]$ defines a feasible set of combined codewords,
$$
\mathit{F}_{\rob}(\codebookS{S}) = \set{\pirateword \in \bits^{\len} \mid \Prob{j \getsr [\len]}{\exists i \in S, \piratewordj{j} = \codewordij{i}{j}} \geq 1-\rob}.
$$

We now recall the definition of an error-robust fingerprinting code from Section~\ref{sec:fpcs}.
\begin{definition}[Error-Robust Fingerprinting Codes (Restatement of Definition \ref{def:rfpc})] \label{def:restaterfpc}
For any $\users, \len \in \N$, $\sec, \rob \in [0,1]$, a pair of algorithms $(\gen, \trace)$ is an \emph{$(\users,\len)$-fingerprinting code with security $\sec$ robust to a $\rob$ fraction of errors} if $\gen$ outputs a codebook $\codebook \in \bits^{\users \times \len}$ and for every (possibly randomized) adversary $\fpadv$, and every coalition $S \subseteq [\users]$, if we set $\pirateword \getsr \fpadv(\codebookS{S})$, then
\begin{enumerate}
\item
$
\prob{(\trace(\codebook, \pirateword) = \bot) \land (\pirateword \in F_{\rob}(\codebookS{S}))} \leq \sec,
$
\item
$
\prob{\trace(\codebook, \pirateword) \in [\users] \setminus S} \leq \sec,
$
\end{enumerate}
where the probability is taken over the coins of $\gen, \trace$, and $\fpadv$.  The algorithms $\gen$ and $\trace$ may share a common state.
\end{definition}

The main result of this section is a construction of fingerprinting codes satisfying Definition~\ref{def:restaterfpc}
\begin{theorem}[Restated from Section~\ref{sec:fpcs}] \label{thm:rfpc0}
For every $\users \in \N$ and $\sec \in (0,1]$, there exists an $(\users, \len)$ fingerprinting code with security $\sec$ robust to a $1/75$ fraction of errors for
$$
\len = \len(\users, \sec) = \tilde{O}(\users^2 \log(1/\sec)).
$$
Equivalently, for every $\len \in \N$, and $\sec \in (0,1]$, there exists an $(\users, \len)$-fingerprinting code with security $\sec$ robust to a $1/75$ fraction of errors for
$$
\users = \users(\len, \sec) = \tilde{\Omega}(\sqrt{\len / \log(1/\sec)}).
$$
\end{theorem}

We remark that we have made no attempt to optimize the fraction of errors to which our code is robust.  We leave it as an interesting open problem to construct a robust fingerprinting code for a nearly-optimal number of users that is robust to a fraction of errors arbitrarily close to $1/2$.

\subsection{From Weak Error Robustness to Strong Error Robustness} \label{sec:weaktostrong}

A key step in our construction is a reduction from constructing error-robust fingerprinting codes to constructing a weaker object, which we call a weakly-robust fingerprinting code. The difference between a weakly-robust fingerprinting code and an error-robust fingerprinting code of the previous section is that we now demand that only a $\beta$ fraction of the \emph{marked} positions can have errors, rather than a $\beta$ fraction of all positions.

 In order to formally define weakly-robust fingerprinting codes, we introduce some terminology.  If $\codebook \in \bits^{\users \times \len}$ is a codebook, then for $b \in \bits$, we say that position $j \in [\len]$ is \emph{$b$-marked in $\codebook$} if $\codewordij{i}{j} = b$ for every $i \in [\users]$.  That is, $j$ is $b$-marked if every user has the symbol $b$ in the $j$-th position of their codeword.  The set $F_{\rob}(\codebook)$ consists of all codewords $\pirateword$ such that for a $1-\rob$ fraction of positions $j$, either $j$ is not marked, or $j$ is $b$-marked and $\piratewordj{j} = b$.  Notice that this constraint is vacuous if fewer than a $\rob$ fraction of positions are marked.

For a weakly-robust fingerprinting code, we will define a more constrained feasible set.  Intuitively, a codeword $\pirateword$ is feasible if for a $1-\rob$ fraction of positions that are marked, $\piratewordj{j}$ is set appropriately.  Note that this condition is meaningful even when the fraction of marked positions is much smaller than $\rob$.  More formally, we define
$$
\mathit{WF}_{\rob}(\codebookS{S}) = \set{\pirateword \in \bits^{\len} \mid \Prob{j \getsr [\len]}{\piratewordj{j} = b \mid \textrm{$j$ is $b$-marked in $\codebookS{S}$ for some $b \in \{0,1\}$}} \geq 1-\rob}.
$$

\begin{definition}[Weakly-Robust Fingerprinting Codes] \label{def:wrfpc}
For any $\users, \len \in \N$ and $\sec, \rob \in [0,1]$, a pair of algorithms $(\gen, \trace)$ is an \emph{$(\users,\len)$-weakly-robust fingerprinting code with security $\sec$ weakly-robust to a $\rob$ fraction of errors} if $(\gen, \trace)$ satisfy the conditions of a robust fingerprinting code (for the same parameters) with $\mathit{WF}_{\rob}$ in place of $\mathit{F}_{\rob}$.
\end{definition}

The next theorem states that if we have an $(n,d)$-fingerprinting code that is weakly-robust to a $\beta$ fraction of errors and satisfies a mild technical condition, then we obtain an $(n, O(d))$-fingerprinting code that is robust to an $\Omega(\beta)$ fraction of errors with a similar level of security.

\begin{lemma} \label{lem:fpcreduction}
For any $\users, \len \in \N$, $\sec, \rob \in [0,1]$, and $\marked \in \N$, suppose there is a pair of algorithms $(\gen,\trace)$ which
\begin{enumerate}
\item are a $(\users,\len)$-fingerprinting code with security $\sec$ weakly-robust to a $\rob$ fraction of errors, and
\item with probability at least $1-\sec$ over $\codebook \getsr \gen$, produce $\codebook$ that has at least $m$ $0$-marked columns and $m$ $1$-marked columns.
\end{enumerate}
Then there is a pair of algorithms $(\gen', \trace')$ that are a $(\users, \len')$-fingerprinting code with security $\sec'$ robust to a $\rob/3$ fraction of errors, where
$$
\len' = 5\len \quad\textrm{ and }\quad \sec' = \sec + 2 \exp\left(-\Omega(\rob m^2/d)\right).
$$
\end{lemma}

\begin{proof}
The reduction is given in Figure~\ref{alg:weaktostrong}. Recall that $\gen'$ and $\trace'$ may share state, so $\pi$ and the shared state of $\gen$ and $\trace$ is known to $\trace'$.
\begin{figure}[ht]
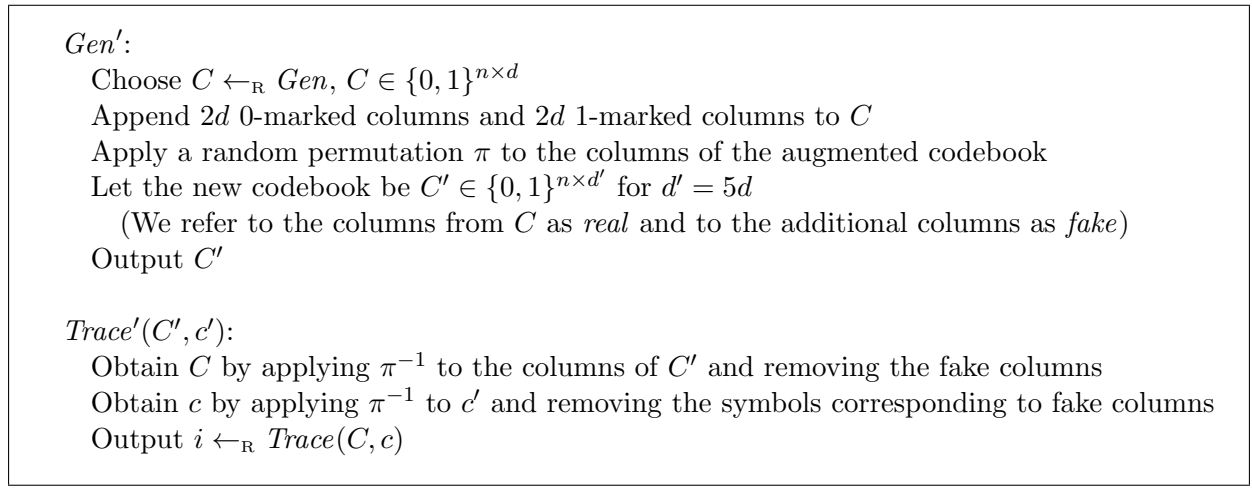

\begin{framed}
\begin{algorithmic}
\STATE{$\gen'$:}
\INDSTATE[1]{Choose $\codebook \getsr \gen$, $\codebook \in \bits^{\users \times \len}$}
\INDSTATE[1]{Append $2\len$ $0$-marked columns and $2\len$ $1$-marked columns to $\codebook$}
\INDSTATE[1]{Apply a random permutation $\pi$ to the columns of the augmented codebook}
\INDSTATE[1]{Let the new codebook be $\codebook' \in \bits^{\users \times \len'}$ for $\len' = 5\len$}
\INDSTATE[2]{(We refer to the columns from $\codebook$ as \emph{real} and to the additional columns as \emph{fake})}
\INDSTATE[1]{Output $\codebook'$}
\STATE{}
\STATE{$\trace'(\codebook', \pirateword)$:}
\INDSTATE[1]{Obtain $\codebook$ by applying $\pi^{-1}$ to the columns of $\codebook'$ and removing the fake columns}
\INDSTATE[1]{Obtain $\codeword$ by applying $\pi^{-1}$ to $\pirateword$ and removing the symbols corresponding to fake columns}
\INDSTATE[1]{Output $i \getsr \trace(\codebook, \codeword)$}
\end{algorithmic}
\end{framed}
\vspace{-6mm}
\caption{Reducing robustness to weak robustness.}
\label{alg:weaktostrong}
\end{figure}

Fix a coalition $S \subseteq [\users]$.  Let $\fpadv'$ be an adversary.  Sample $\codebook' \getsr \gen'$ and let $\pirateword = \fpadv'(\codebook')$. We will show that the reduction is successful by proving that if $\pirateword \in \mathit{F}_{\rob/3}(\codebook')$, then the modified string $\codeword \in \mathit{WF}_{\rob}(\codebook)$ with probability $1- \exp(-\Omega(\rob m^2/d))$.  The reason is that an adversary who is given (a subset of the rows of) $\codebook'$ cannot distinguish real columns that are marked from fake columns.  Therefore, the fraction of errors in the real marked columns should be close to the fraction of errors that are either real and marked or fake.  Since the total fraction of errors in the entire codebook is at most $\rob/3$, we know that the fraction of errors in real marked columns is not much larger than $\rob/3$.  Thus the fraction of errors in the real marked columns will be at most $\rob$ with high probability.  We formalize this argument in the following claim.

\begin{claim} \label{clm:weaktostrong0}
$$
\Prob{\pi}{(\pirateword \in \mathit{F}_{\rob/3}(\codebook')) \land (\codeword \in \mathit{WF}_{\rob}(\codebook))} \leq 2 \exp(-\Omega(\beta m^2/d))
$$
\end{claim}
\begin{proof} [Proof of Claim~\ref{clm:weaktostrong0}]
Our analysis will handle $0$-marked and $1$-marked columns separately. Assume that $\pirateword \in \mathit{F}_{\rob/3}(\codebook')$ and that the adversary has introduced $k \le \rob d' / 3$ errors to $0$-marked columns. Let $m_0 \geq m$ be the number of $0$-marked columns.  Let $R_0$ be a random variable denoting the number of columns that are both real and $0$-marked in which the adversary introduces an error. Since real $0$-marked columns are indistinguishable from fake $0$-marked columns, $R_0$ has a hypergeometric distribution on $k$ draws from a population of size $N = m_0 + 2d$ with $m_0$ successes. In other words, we can think of an urn with $N$ balls, $m_0$ of which are labeled ``real'' and $2d$ of which are labeled ``fake.'' We draw $k$ balls without replacement, and $R_0$ is the number that are labeled ``real.''  This distribution has $\ex{R_0} = km_0/N = km_0/(m_0+2d)$.  Moreover, as shown in~\cite[Section 7.1]{DubhashiSe01}), it satisfies the concentration inequality
\[\Pr[|R_0 - \ex{R_0}| > t] \le \exp\left(\frac{-2(N-1)t^2}{(N-k)(k-1)}\right) \le \exp(-\Omega(t^2/k))\]
since $k \le 5N/6$. Thus
\begin{align*}
\Pr[R_0 > \beta m_0]  &\le \Pr[|R_0 - \ex{R_0}| > \beta m_0 - \ex{R_0}] \\
&\le \exp\bigg(-\Omega\bigg(\frac{(\beta m_0 - km_0/N)^2}{k^2}\bigg)\bigg) \\
&\le \exp\bigg(-\Omega\bigg(\frac{(\rob m_0)^2(1 - d'/6d)^2}{(\rob d'/3)^2}\bigg)\bigg) \\
&\le \exp\bigg(-\Omega\bigg(\frac{\rob m_0^2}{d}\bigg)\bigg)
\end{align*}
for any choice of $k$. An identical argument bounds the probability that the number of errors in real $1$-marked columns is more than $\beta m_1$. Therefore, the probability that more than a $\beta$ fraction of marked columns have errors is at most $2 \exp(-\Omega(\beta m^2/d))$.
\end{proof}

\begin{comment}
Suppose there were an adversary $\fpadv'$ that could cause $\trace'$ to output $\bot$ with too high a probability.  That is,
$$
\Prob{\codebook' \getsr \gen'}{(\trace'(\codebook', \fpadv'(\codebookS{S})) = \bot) \land (\fpadv'(\codebookS{S}) \in \mathit{F}_{\rob/3}(\codebookS{S}))} > \sec',
$$
for some $S \subseteq [\users]$.  Then we will construct an adversary $\fpadv$ that causes $\trace$ to output $\bot$ with too high a probability and obtain a contradiction.  That is,
$$
\Prob{\codebook \getsr \gen}{(\trace(\codebook, \fpadv(\codebookS{S})) = \bot \land (\fpadv(\codebookS{S}) \in \mathit{WF}_{\rob}(\codebookS{S}))} > \sec.
$$
\end{comment}
Now define an adversary $\fpadv$ that takes $\codebookS{S}$ as input, simulates $\gen'$ by appending marked columns to $\codebook_S$ and applying a random permutation $\pi$, and then applies $\fpadv'$ to the resulting codebook $\codebook'_{S}$.  Then it takes $\fpadv'(\codebook'_{S})$, applies $\pi^{-1}$, removes the fake columns, and outputs the result.  Notice that $\trace'$ applies $\trace$ to a codebook and codeword generated by exactly the same procedure.  If we assume that $\fpadv'(\codebook'_{S})$ is feasible with parameter $\rob/3$, then by the analysis above, with probability at least $1 - \sec - \exp(-\Omega(\rob m^2/d))$, $\fpadv(\codebookS{S})$ is weakly feasible with parameter $\rob$.  Thus,
\begin{align*}
&\Prob{\codebook' \getsr \gen'}{(\trace'(\codebook', \fpadv'(\codebookS{S})) = \bot) \land (\fpadv'(\codebookS{S}) \in \mathit{F}_{\rob/3}(\codebookS{S}))} \\
\leq{} &\Prob{\codebook \getsr \gen}{(\trace(\codebook, \fpadv(\codebookS{S})) = \bot \land (\fpadv(\codebookS{S}) \in \mathit{WF}_{\rob}(\codebookS{S}))} + 2 \exp(-\Omega(\beta m^2/d)) \\
\leq{} &\sec + 2 \exp(-\Omega(\beta m^2/d)),
\end{align*} 
where the first inequality is by Claim~\ref{clm:weaktostrong0} and the second inequality is by $\sec$-security of $\trace$.

Since $\trace$ does not accuse a user outside of $S$ (except with probability at most $\sec$) regardless of whether or not that adversary's codeword is feasible, it is immediate that $\trace'$ also does not accuse a user outside of $S$ (except with probability at most $\sec$).
\end{proof}

\subsection{Weak Robustness of Tardos' Fingerprinting Code} \label{sec:weakrobusttardos}
In this section we show that Tardos' fingerprinting code is weakly robust to a $\rob$ fraction of errors for $\rob \geq 1/25$.  Specifically we prove the following:
\begin{lemma} \label{lem:weakrobusttardos}
For every $\users \in \N$ and $\sec \in (0,1]$, there exists an $(\users, \len)$ fingerprinting code with security $\sec$ weakly robust to a $1/25$ fraction of errors for
$$
\len = \len(\users, \sec) = \tilde{O}(\users^2 \log(1/\sec)).
$$
Equivalently, for every $\len \in \N$, and $\sec \in (0,1]$, there exists an $(\users, \len)$-fingerprinting code with security $\sec$ weakly robust to a $1/25$ fraction of errors for
$$
\users = \users(\len, \sec) = \tilde{\Omega}(\sqrt{\len / \log(1/\sec)}).
$$
\end{lemma}
Tardos' fingerprinting code is described in Figure~\ref{alg:tardoscode}. Note that the shared state of $\gen$ and $\trace$ will include $p_1, \dots, p_d$.

\begin{figure}[ht]
\begin{framed}
\begin{algorithmic}
\STATE{$\gen$:}
\INDSTATE[1]{Let $\len = 100\users^2 \log(\users/\sec)$ be the length of the code.}
\INDSTATE[1]{Let $t = 1/300\users$ be a parameter and let $t'$ be such that $\sin^2 t' = t$.}
\INDSTATE[1]{For $j = 1,\dots,\len$:}
\INDSTATE[2]{Choose $r_j \getsr [t', \pi/2 - t']$ and let $p_j = \sin^2 r_j$.  Note that $p_j \in [t, 1-t]$.}
\INDSTATE[2]{For each $i = 1,\dots,\users$, set $\codebook_{ij} = 1$ with probability $p_j$, independently.}
\INDSTATE[1]{Output $\codebook$.}
\STATE{}
\STATE{$\trace(\codebook, \pirateword)$:}
\INDSTATE[1]{Let $Z = 20\users \log(\users/\sec)$ be a parameter.}
\INDSTATE[1]{For each $j = 1,\dots,\len$, let $q_j = \sqrt{(1-p_j)/p_j}$.}
\INDSTATE[1]{For each $j = 1,\dots,\len$, and each $i = 1,\dots,\users$, let 
	$$U_{ij} = \begin{cases} q_j & \textrm{if $\codebook_{ij} = 1$} \\ -1/q_j & \textrm{if $\codebook_{ij} = 0$} \end{cases}$$}
\INDSTATE[1]{For each $i = 1,\dots,\users$:}
\INDSTATE[2]{Let
	$$S_i(\pirateword) = \sum_{j=1}^{\len} \piratewordj{j} U_{ij}$$}
\INDSTATE[2]{If $S_i(\pirateword) \geq Z/2$, output $i$}
\INDSTATE[1]{If $S_i(\pirateword) < Z/2$ for every $i = 1,\dots,\users$, output $\bot$.}
\end{algorithmic}
\end{framed}
\vspace{-6mm}
\caption{The Tardos Fingerprinting Code~\cite{Tardos08}}
\label{alg:tardoscode}
\end{figure}

Tardos' proof that no user is falsely accused (except with probability $\sec$) holds for every adversary, regardless of whether or not the adversary's output is feasible, therefore it holds without modification even when we allow the adversary to introduce errors.  So we will state the following lemma from~\cite[Section 3]{Tardos08} without proof.
\begin{lemma}[Restated from \cite{Tardos08}] \label{lem:robsoundness}
Let $(\gen, \trace)$ be the fingerprinting code defined in Algorithm~\ref{alg:tardoscode}.  Then for every adversary $\fpadv$, and every $S \subseteq [\users]$,
$$
\prob{\trace(\codebook, \fpadv(\codebookS{S})) \in [n] \setminus S} \leq \sec,
$$
where the probability is taken over the choice of $\codebook \getsr \gen$ and the coins of $\fpadv$.
\end{lemma}

Most of the remainder of this section is devoted to proving that any adversary who introduces errors into at most a $1/25$ fraction of the marked columns can be traced successfully.
\begin{lemma} \label{lem:robcompleteness}
Let $(\gen, \trace)$ be the fingerprinting code defined in Algorithm~\ref{alg:tardoscode}.  Then for every adversary $\fpadv$, and every $S \subseteq [\users]$,
$$
\prob{(\trace(\codebook, \fpadv(\codebookS{S})) = \bot) \land (\fpadv(\codebookS{S}) \in \mathit{WF}_{1/25}(\codebookS{S}))} \leq \sec,
$$
where the probability is taken over the choice of $\codebook \getsr \gen$ and the coins of $\fpadv$.
\end{lemma}

Before giving the proof, we briefly give a high-level roadmap.  Recall that in the construction there is a ``score'' function $S_i(\pirateword)$ that is computed for each user, and $\trace$ will output some user whose score is larger than the threshold $Z/2$, if such a user exists.  Tardos shows that the sum of the scores over all users is at least $nZ/2$, which demonstrates that there exists a user whose score is above the threshold.  His argument works by balancing two contributions to the score: 1) the contribution from $1$-marked columns $j$, which will always be positive due to the fact that $\pirateword_j = 1$, and 2) the potentially negative contribution from columns that are not $1$-marked.  Conceptually, he shows that the contribution from the $1$-marked columns is larger in expectation than the negative contribution from the other columns, so the expected score is significantly above the threshold.  He then applies a Chernoff-type bound to show that the score will be above the threshold with high probability.  When the adversary is allowed to introduce errors so that there may be some $1$-marked columns $j$ such that $\pirateword_j = 0$, these errors will contribute negatively to the score.  The new ingredient in our argument is essentially to bound the negative contribution from these errors.  We are able to get a sufficiently good bound to tolerate errors in $1/25$ of the coordinates.  We expect that a tighter analysis and more careful tuning of the parameters can improve the fraction of errors that can be tolerated.

\begin{proof}[Proof of Lemma~\ref{lem:robcompleteness}]
We will write $S = [\users]$.  Doing so is without loss of generality as users outside of $S$ are irrelevant.  We will use $\beta = 1/25$ to denote the allowable fraction of errors.  Fix an adversary $\adv$.  Sample $\codebook \getsr \gen$ and let $\pirateword = \adv(\codebook)$.  Assume $\pirateword \in \mathit{WF}_{\beta}(C)$.  In order to prove that some user is traced, we will  bound the quantity
$$
S(\pirateword) = \sum_{i=1}^{\users} S_i(\pirateword) = \sum_{j = 1}^{\len} \piratewordj{j} \left( x_j q_j - \frac{n - x_j}{q_j}\right)
$$
where $x_j = \sum_{i=1}^{\users} \codebook_{ij}$ is defined to be the number of codewords $\codewordi{i}$ such that $\codewordij{i}{j} = 1$.  Our goal is to show that this quantity is at least $\users Z/2$ with high probability.  If we can do so, then there must exist a user $i \in [\users]$ such that $S_i(\pirateword) \geq Z/2$, in which case $\trace(\codebook, \pirateword) \neq \bot$.

\begin{comment}
We will decompose $\adv$ into a pair of algorithms $\adv_0$ and $\adv_{1}$, where $\interword = \adv_0(\codebook) \in \mathit{WF}_0(\codebook)$ ($\interword$ has no errors) and $\pirateword = \adv_1(\interword, \codebook)$ introduces errors into at most a $\rob$ fraction of the marked coordinates.  We may assume $\adv_0$ is deterministic, since we can derandomize $\adv_0$ by letting it output the feasible codeword $\interword$ that maximizes the probability that no user is traced. However, it will be convenient in our analysis to continue regarding $\adv_1$ as randomized. Thus we can write $\adv(\codebook) = \adv_1(\adv_0(\codebook), \codebook)$.  Let $\diffword = \pirateword - \interword$.
\end{comment}

We may decompose an output $\pirateword$ of $\adv(\codebook)$ into a the sum of a codeword $\interword \in \mathit{WF}_0(C)$ with no errors, and a string $\diffword$ that captures errors introduced into at most a $\rob$ fraction of the marked coordinates. Each codeword $\codeword$ has a unique such decomposition if we assume the following constraints on $\diffword$.
\begin{enumerate}
\item If $j$ is unmarked, then $\diffword_j = 0$.
\item If $j$ is $0$-marked, then $\diffword_j \in \{0, 1\}$.
\item If $j$ is $1$-marked, then $\diffword_j \in \{-1, 0\}$.
\item The number of nonzero coordinates of $\diffword$ is at most $\rob m$, where $m$ is the number of marked columns of $\codeword$.
\end{enumerate}
We call a $\diffword$ satisfying the above constraints \emph{valid}. By the linearity of $S(\cdot)$, we can write
$$
S(\pirateword) = S(\interword) + S(\diffword).
$$
Tardos' analysis of the error-free case proves that $S(\interword)$ is large.  In our language, he proves
\begin{claim}[Restated from \cite{Tardos08}] \label{clm:errorfreecase}
For every adversary $\adv$, if $\codebook \getsr \gen$, $\pirateword \getsr \adv(\codebook)$, and $\pirateword = \tilde{c} + \overline{c}$ as above, then
$$
\prob{(S(\interword) < nZ) \land (\interword \in \mathit{WF}_{0}(\codebook))} \leq \sec^{\sqrt{\users}/4}.
$$
\end{claim}

Although $S(\diffword)$ will be negative, and thus $S(\pirateword) \leq S(\interword)$, we will show that $S(\diffword)$ is not too negative.  That is, introducing errors into a $\rob$ fraction of the marked columns in $\pirateword$ cannot reduce $S(\pirateword)$ by too much. 

\begin{comment}
 First, we state some simple facts about $\diffword$.
\begin{fact} \label{fact:difffacts}
For $\pirateword, \interword, \diffword$ as above, if $\pirateword \in \mathit{WF}_{\beta}(\codebook)$, then the following all hold.
\begin{enumerate}
\item If $j$ is $0$-marked, then $\diffword_j \in \{0, 1\}$. If $j$ is $1$-marked, then $\diffword_j \in \{-1, 0\}$. This is because $\pirateword, \interword \in \{0, 1\}$ where $\interword \in \mathit{WF}_0(\codebook)$, and $\diffword = \pirateword - \interword$. %$\diffword \in \{-1,0,1\}^\len$, since $\pirateword, \interword \in \bits^\len$ and $\diffword = \pirateword - \interword$.
\item By assumption, the number of non-zero coordinates of $\diffword$ is at most $\rob m$, where $m$ is the number of marked columns of $\codebook$, and 
\item By assumption, $\diffword_j \neq 0$ only if $j$ is a marked column of $\codebook$.
\end{enumerate}
\end{fact}
\end{comment}

We will now establish the following claim.
\begin{claim} \label{clm:errorcase}
For any adversary $\adv$, if $\codebook \getsr \gen$, $\codeword' \getsr \adv(\codebook)$, and $\pirateword = \tilde{c} + \overline{c}$ as above, then
$$
\prob{(S(\diffword) < -nZ/2) \land (\textrm{$\diffword$  is valid} )} \leq \sec/2.
$$
\end{claim}

\begin{proof} [Proof of Claim~\ref{clm:errorcase}]
We start by making an observation about the distribution of $S(\diffword) = S(\diffword)|_{C, \diffword}$, which denotes $S(\diffword)$ when we condition on a fixed choice of a codebook $\codebook$ and a valid choice of $\diffword$.  Because the non-zero coordinates of $\diffword$ are only in marked columns of $\codebook$ (those in which $x_j = 0$ or $x_j = \users$), the distribution of
$$
S(\diffword)|_{C, \diffword} = \sum_{j=1}^{\len} \diffword_j \left( x_j q_j - \frac{\users - x_j}{q_j}\right)
$$
depends only on the number of non-zero coordinates of $\diffword$, and not on their location.  To see that this is the case, consider a $0$-marked coordinate $j$ on which $\diffword_j = 1$.  The contribution of $j$ to $S(\diffword)$ is exactly $-\users / q_j$.  Similarly, for a $1$-marked coordinate $j$ on which $\diffword_j = -1$, the contribution of $j$ to $S(\diffword)$ is exactly $-\users q_j$.  Thus we can write
\begin{align}
S(\diffword) 
&={} \sum_{j=1}^{\len} \diffword_j \left( x_j q_j - \frac{\users - x_j}{q_j}\right) \notag \\
&={} - \left(\sum_{j \in [d] : \textrm{$j$ is $0$-marked and $\diffword_j = 1$}} n / q_j  + \sum_{j \in [d] : \textrm{$j$ is $1$-marked and $\diffword_j = -1$}} n q_j \right) \label{eq:scoreassum}
\end{align}
Each term in the first sum (resp.~second sum) is a random variable that depends only on the distribution of $q_j$ conditioned on the the $j$-th column being $0$-marked (resp.~$1$-marked).  Recall that $q_j$ is determined by $p_j$.  Moreover, conditioned on a fixed $\codebook$, the $p_j$'s are independent. To see this, let $\codebook_j$ denote the $j$th column of the codebook $\codebook$. Recall that each column $\codebook_j$ is generated independently using $p_j$, and the $p_j$'s themselves are chosen independently. Letting $f_X$ denote the density function of a random variable $X$, this means that the joint density
\begin{align*}
f_{p_1, \dots, p_d}(x_1, \dots, x_d \mid C_1, \dots, C_d) &= \frac{\Pr[C_1, \dots, C_d \mid x_1, \dots, x_d] f_{p_1, \dots, p_d}(x_1, \dots, x_d)}{\Pr[C_1, \dots, C_d]} \tag{Bayes' rule}\\
&= \frac{\Pr[C_1 \mid x_1]f_{p_1}(x_1)}{\Pr[C_1]} \cdot \ldots \cdot \frac{\Pr[C_d \mid x_d]f_{p_d}(x_d)}{\Pr[C_d]}\\
&= f_{p_1}(x_1 \mid C_1) \cdot \ldots \cdot f_{p_d}(x_d \mid C_d).
\end{align*} 
This shows that the conditional random variables $p_j|_{C_j}$ are independent. Moreover, since $\diffword$ only depends on the codebook $\codebook$ and coins of the adversary $\adv$, the $p_j$'s are still independent when we also condition on $\diffword$. In fact, the following holds: %Using the fact that $p_j$ is symmetric about $1/2$ and $q_j = \sqrt{(1-p_j)/p_j}$, one can very the following fact:
\begin{claim} \label{clm:iid}
Conditioned on any fixed choice of $\codebook$ and $\diffword$, the following distributions are all identical, independent, and non-negative: 1) $(n/q_j \mid \textrm{$j$ is $0$-marked})$ for $j \in [d]$, and 2) $(nq_j \mid \textrm{$j$ is $1$-marked})$.
\end{claim}
\begin{proof}[Proof of Claim~\ref{clm:iid}]
By the discussion above, we know that these random variables are independent. To see that they are identicially distributed, note that the distribution $p_j$ used to generate the $j$th column of $\codebook$ is symmetric about $1/2$. Therefore, the probability that column $j$ is $0$-marked when its entries are sampled according to $p_j$ is the same as the probability that $j$ is $1$-marked when its entries are sampled according to $1-p_j$. Applying Bayes' rule, again using the fact that $p_j$ and $1-p_j$ have the same distribution, we see that the random variables $(p_j \mid \textrm{$j$ is $0$-marked})$ and $(1-p_j \mid \textrm{$j$ is $1$-marked})$ are identically distributed. The claim follows since $q_j = \sqrt{(1-p_j)/p_j}$.
\end{proof}
In light of this fact, we can see that the conditional random variable $S(\diffword)|_{\codebook, \diffword}$ is a sum of i.i.d. random variables and the number of these variables in the sum is exactly the number of marked columns $j$ on which $\diffword_j$ is non-zero.  For any $t \in \N$ and any non-negative random variable $Q$, the sum of $t+1$ independent draws from $Q$ stochastically dominates\footnote{For random variables $X$ and $Y$ over $\R$, $X$ \emph{stochastically dominates} $Y$ if for every $z \in \R$, $\prob{X \geq r} \geq \prob{Y \geq r}$.} the sum of $t$ independent draws from $Q$.  Recall that $S(\diffword)$ will be negative and we want its magnitude not to be too large.  Equivalently, we want the positive sum in~\eqref{eq:scoreassum} not to be too large.  Therefore, the ``worst-case'' for the sum~\eqref{eq:scoreassum} is when $\diffword$ has the largest possible number of non-zero coordinates.  Recall that the number of non-zero coordinates of $\diffword$ is exactly the number of errors introduced by the adversary.  Thus, the ``worst-case'' adversary $\adv^*$ is the one that chooses a random set of exactly $\rob m$ marked columns and for the chosen columns $j$ that are $0$-marked, sets $\diffword_j = 1$ and for those that are $1$-marked, sets $\diffword_j = -1$.  In summary, it suffices to consider only the single adversary $\adv^*(\codebook)$ that constructs a feasible $\interword$ and introduces errors in a random set of $\rob m$ of the marked coordinates in $\codebook$.  

Now we proceed to analyzing $\adv^*$.  We follow Tardos' approach to analyzing $S$.  A key step in his analysis is to show that the optimal adversary (for the error-free case) chooses the $j$-th coordinate of $\pirateword$ based only on the $j$-th column of $\codebook$.  In our case, the optimal adversary $\adv^*$ introduces errors in a random set of exactly $\rob m$ marked columns, which does not satisfy this independence condition.  So instead, we will analyze an adversary $\hat{\adv}^*$ that introduces an error in each marked column independently with probability $\rob$.  This adversary may fail to introduce errors in exactly $\rob m$ random columns, and thus it is not immediately sufficient to bound $\prob{S(\diffword) < -nZ/2}$ for $\pirateword \getsr \hat{\adv}^*(\codebook)$.  However, a standard analysis of the binomial distribution shows that this adversary introduces errors in exactly $\rob m$ marked columns with probability at least
$$
1/2\sqrt{m} \geq 1/2 \sqrt{d} = 1/\poly(\users, \log(1/\sec)),
$$ 
and conditioned on having $\rob m$ errors, those errors occur on a uniformly random set of marked columns.  Thus, if we can show that 
$$
\Prob{\pirateword \getsr \hat{\adv}^*(\codebook)}{S(\diffword) < -nZ/2} < \sec^{\sqrt{n}/4},
$$
we must also have 
$$
\Prob{\pirateword \getsr \adv^*(\codebook)}{S(\diffword) < -nZ/2} \leq \poly(\users, \log(1/\sec)) \cdot \sec^{\sqrt{n}/4} \leq \sec/2,
$$
provided $\users, 1/\sec$ are sufficiently large.

For the remainder of the proof, we will show that indeed $\prob{S(\diffword) < -nZ/2} < \sec^{\sqrt{n}/4}$ for $\pirateword \getsr \adv^*(\codebook)$.  We do so by bounding the quantity $\Ex{\overline{p},\codebook}{e^{-\alpha S}}$ for a suitable $\alpha > 0$ that we will choose later, and then by applying Markov's inequality.  Note that the expectation is taken over both the parameters $\overline{p} = (p_1, \dots, p_d)$ and the randomness of the adversary.
\begin{align*}
\Ex{\overline{p},\codebook}{e^{-\alpha S}}
&={} \sum_{\codebook} \Ex{\overline{p}}{e^{-\alpha S} \prod_{j=1}^{\len} p_j^{x_j} (1-p_j)^{n-x_j}} \\
&={} \sum_{\codebook} \Ex{\overline{p}}{\prod_{j=1}^{\len} p_j^{x_j} (1-p_j)^{\users-x_j} e^{-\alpha \diffword_j \left(x_jq_j - \frac{\users-x_j}{q_j}\right)}} \\
&={} \sum_{\codebook} \prod_{j=1}^{\len} \Ex{p}{p^{x_j} (1-p)^{\users-x_j} e^{-\alpha \diffword_j \left(x_jq_j - \frac{\users-x_j}{q_j}\right)}} \\
\end{align*}
The first two equalities are by definition.  The third equality follows from observing that for fixed $\codebook$, each term in the product depends only on the (independent) choice of $p_j$ and the adversary's choice of $\diffword_j$, and are thus independent by our choice of adversary $\tilde{\adv}^*$.  This step is the sole reason why it was helpful to consider an adversarial strategy that treats columns independently.  Now we want to interchange the sum and product to obtain a product of identical terms, so we can analyze the contribution of an individual term to the product.
\begin{align}
\Ex{\overline{p},\codebook}{e^{-\alpha S}}
&={} \sum_{\codebook} \prod_{j=1}^{\len} \Ex{p}{p^{x_j} (1-p)^{\users-x_j} e^{-\alpha \diffword_j \left(x_jq_j - \frac{\users-x_j}{q_j}\right)}} \notag \\
&={} \left( \sum_{x = 0}^{\users} \binom{\users}{x} \Ex{p}{p^{x}(1-p)^{\users-x} e^{-\alpha \diffword \left(xq - \frac{\users-x}{q}\right)}} \right)^{\len} \tag{independence of $\overline{c}_j$'s} \\
&={} \left( \sum_{x = 0}^{\users} \binom{\users}{x} A_{x} \right)^\len \notag
\end{align}
where
\begin{equation*}
A_{x} =
\begin{cases}
(1-\rob) \Ex{p}{(1-p)^\users} + \rob \Ex{p}{(1-p)^\users e^{\alpha n / q}} &\textrm{if $x = 0$} \\
\Ex{p}{p^{x}(1-p)^{n-x}} &\textrm{if $1 \leq x \leq n-1$} \\
(1-\rob) \Ex{p}{p^n} + \rob \Ex{p}{p^n e^{\alpha n q}} &\textrm{if $x = n$}
\end{cases}
\end{equation*}

First, observe that, since the distribution of $p$ is symmetric about $1/2$, $A_0 = A_n$.  Second, if we let
\begin{equation*}
B_{x} = \Ex{p}{p^x (1-p)^{n-x}}
\end{equation*}
for every $x = 0,1,\dots,n$, then we have
\begin{align*}
\sum_{x = 0}^{n} \binom{n}{x} A_{x} 
&={} \left(\sum_{x = 0}^{n} \binom{n}{x} B_{x} \right) + 2(A_n - B_n) \\
&={} 1 + 2(A_n - B_n)
\end{align*}
In order to obtain a strong enough bound, we need to show that $A_n - B_n = O(\rob \alpha)$.  We can calculate
\begin{align*}
A_n - B_n
&={} (1-\rob) \Ex{p}{p^n} + \rob \Ex{p}{p^n e^{\alpha n q}} - \Ex{p}{p^n} \\
&={} \rob \Ex{p}{p^n e^{\alpha n q}} - \rob\Ex{p}{p^n}
\end{align*}
Now we apply the approximation $e^{u} \leq 1+2u$, which holds for $0 \leq u \leq 1$.  To do so, we choose $\alpha = \sqrt{t}/n$.  Since $q = \sqrt{(1-p)/p}$ and $p \geq t$, we have $\alpha n q \leq 1$ for this choice of $\alpha$.  Thus we have
\begin{align*}
A_n - B_n
&= {}\rob \Ex{p}{p^n e^{\alpha n q}}  - \rob \Ex{p}{p^n} \\
&\leq{} \rob \Ex{p}{p^n (1+2\alpha n q)} - \rob \Ex{p}{p^n} \\
&={} 2\rob \alpha \Ex{p}{p^n n q}
\end{align*}
Now, to show that $A_n - B_n = O(\rob \alpha)$, we simply want to show that $\Ex{p}{p^n n q} = O(1)$, which we do by direct calculation.
\begin{align*}
\Ex{p}{p^n n \sqrt{\frac{1-p}{p}}}
&={} n \int_{t'}^{\pi/2 - t'} \frac{\sin^{2n} r \sqrt{\frac{1-\sin^2 r}{\sin^2 r}}}{\pi/2 - 2t'} dr 
%&={} n \int_{t'}^{\pi/2 - t'} \frac{\sin^{2n} r  \cot r}{\pi/2 - 2t'} dr \\
%&={} \left[ \frac{\sin^{2n} r}{\pi - 4t'} \right]_{t'}^{\pi/2 - t'} \\
={} \frac{\sin^{2n} (\pi/2 - t') - \sin^{2n} (t')}{\pi - 4t'} \\
&={} \frac{(1-t)^{n} - t^{n}}{\pi - 4t'}
={} \frac{(1-1/300n)^{n} - (1/300n)^n}{\pi - 4t'}
\leq{} \frac{1}{\pi}
\end{align*}
The final inequality holds as long as $n$ is larger than some absolute constant.  (To see that this is the case, recall that $t' = \arcsin(\sqrt{t}) = \arcsin(\sqrt{1/300n}) = \Theta(1/\sqrt{n})$, whereas $(1-1/300n)^n = 1 - \Omega(1)$.) So we have established
$$
A_n - B_n \leq \frac{2\rob \alpha}{\pi}.
$$
Plugging this fact into the analysis above, we have
\begin{align*}
\Ex{\overline{p},\codebook}{e^{-\alpha S}} 
&={} \left( \sum_{x=0}^{n} \binom{n}{x} A_x \right)^\len \\
&={} \left( 1 + 2(A_n - B_n) \right)^d \\
&\leq{} \left(1 + \frac{4 \rob \alpha}{\pi}\right)^\len \leq e^{4\rob \alpha \len / \pi}
\end{align*}
Now all that remains is to apply Markov's inequality to bound this quantity by $\sec^{\sqrt{n}/4}$.
\begin{align*}
\prob{S < -nZ/2}
&={} \prob{-\alpha S > \alpha nZ/2} \\
&={} \prob{e^{-\alpha S} > e^{\alpha nZ/2}}
\leq{} \frac{\ex{e^{-\alpha S}}}{e^{\alpha n Z / 2}}
\leq{} \frac{e^{4\beta \alpha \len / \pi}}{e^{\alpha n Z / 2}} \\
&={} e^{4\beta \alpha \len / \pi - \alpha n Z / 2}
\end{align*}
To get the desired upper bound, it is sufficient to show $$\frac{\alpha n Z }{ 2}  -\frac{4 \rob \alpha \len }{ \pi} \geq \frac{\sqrt{n}\log(1/\sec)}{4}.$$
We calculate
\begin{align*}
\frac{\alpha n Z}{2} - \frac{4 \rob \alpha \len }{ \pi}
&={} 10 \sqrt{t} n \log(n/\sec) - \frac{400 \rob}{\pi} \sqrt{t} n \log(n/\sec) \\
&={} \left(10 - \frac{400 \rob}{\pi}\right) \left( \sqrt{t}{n} \log(n/\sec) \right) \\
&\geq{} \left(10 - \frac{400 \rob}{\pi}\right) \frac{\sqrt{n} \log(n/\sec)}{18} \\
&\geq{} \frac{\sqrt{n} \log(1/\sec)}{4}
\end{align*}
where the last inequality holds when $\rob < 1/25$.  This is sufficient to complete the proof of Claim \ref{clm:errorcase}.
\end{proof}
Combining Claims \ref{clm:errorfreecase} and \ref{clm:errorcase} yields Lemma \ref{lem:robcompleteness} as follows. If $S(\pirateword) < nZ/2$, then either $S(\interword) < nZ$ or $S(\diffword) < nZ/2$. Moreover, if $\pirateword \in \mathit{WF}_{1/25}(\codebook)$, we must have both $\interword \in \mathit{WF}_0(\codebook)$ and a valid $\diffword$. A union bound thereby gives us Lemma \ref{lem:robcompleteness}.
\end{proof}

Lemma~\ref{lem:robsoundness} and~\ref{lem:robcompleteness} are sufficient to imply Lemma~\ref{lem:weakrobusttardos}, that Tardos' fingerprinting code is weakly robust.  In order to apply our reduction from full robustness to weak robustness (Lemma~\ref{lem:fpcreduction}), we need to also establish that with high probability there are many marked columns in the matrix $\codebook \getsr \gen$ for Tardos' fingerprinting code.

\begin{lemma} \label{lem:manymarkedcols}
With probability at least $1-\sec$ over the choice of $\codebook \getsr \gen$, it holds that the number of $0$-marked columns $m_0$ and the number of $1$-marked columns $m_1$ are both larger than $m = 5\users^{3/2} \log(\users/\sec)$.
\end{lemma}
\begin{proof} [Proof of Lemma~\ref{lem:manymarkedcols}]
To estimate the number of marked columns, define for each $j = 1, \dots, d$ an indicator random variable $D_j$ for whether column $j$ is $0$-marked. The $D_j$'s are i.i.d., and have expectation at least
\[\ex{D_j | p_j < 1/n} \Pr[p_j < 1/n] > \left(1 - \frac{1}{n}\right)^n \Pr[r_j < \arcsin (1/\sqrt{n})] \ge \frac{1}{6\sqrt{n}}.\]
Let $D = \sum_{j=1}^d D_j$ be the total number of $0$- marked columns. Then $\ex{D} \ge 10n\sqrt{n}\log(n/\sec)$, so by the additive Chernoff bound (Theorem \ref{thm:chernoff}),
\[\Pr[D < 5n\sqrt{n}\log(n/\sec)] < \exp\left(\frac{-2(5n\sqrt{n}\log(n/\sec))^2}{d}\right) < \sec/2.\]
A similar argument holds for $1$-marked columns. Thus letting $m = 5n\sqrt{n}\log(n/\sec)$, the codebook $C$ has at least $m$ $0$-marked columns and $m$ $1$-marked columns with probability at least $1 - \sec$. Now observe that
\[\exp(-\Omega(\beta m^2 / d)) < \exp(-\Omega(\beta n\log(n/\sec))) < \sec\]
for $n$ larger than some absolute constant. 
\end{proof}

Combining Lemma~\ref{lem:fpcreduction} (reduction from robustness to weak robustness), Lemma~\ref{lem:weakrobusttardos} (weak robustness of Tardos' code), and Lemma~\ref{lem:manymarkedcols} (Tardos' code has many marked columns), suffices to prove Theorem~\ref{thm:rfpc0}.

\section*{Acknowledgements}

We thank Kobbi Nissim for drawing our attention to the question of sample complexity and for many helpful discussions.  We thank Adam Smith for suggesting that we use the Gaussian mechanism to provide a new proof of the lower bound on the length of fingerprinting codes.  Finally, we thank the anonymous reviewers for their helpful comments.

\addcontentsline{toc}{section}{References}
\bibliographystyle{alpha}
\bibliography{references}

\newcommand{\etalchar}[1]{$^{#1}$}
\begin{thebibliography}{DKM{\etalchar{+}}06}

\bibitem[AB09]{AnthonyBa09}
Martin Anthony and Peter~L. Bartlett.
\newblock {\em Neural Network Learning: Theoretical Foundations}.
\newblock Cambridge University Press, New York, NY, USA, 1st edition, 2009.

\bibitem[BCD{\etalchar{+}}07]{BarakChDwKaMcTa07}
Boaz Barak, Kamalika Chaudhuri, Cynthia Dwork, Satyen Kale, Frank McSherry, and
  Kunal Talwar.
\newblock Privacy, accuracy, and consistency too: a holistic solution to
  contingency table release.
\newblock In {\em PODS}, pages 273--282, June 11--13 2007.

\bibitem[BDMN05]{BlumDwMcNi05}
Avrim Blum, Cynthia Dwork, Frank McSherry, and Kobbi Nissim.
\newblock Practical privacy: the {SuLQ} framework.
\newblock In {\em PODS}, pages 128--138. {ACM}, June 13--15 2005.

\bibitem[BKM10]{BonehKiMo10}
Dan Boneh, Aggelos Kiayias, and Hart~William Montgomery.
\newblock Robust fingerprinting codes: a near optimal construction.
\newblock In {\em Digital Rights Management Workshop}, pages 3--12. {ACM}, Oct
  4 2010.

\bibitem[BKN10]{BeimelKaNi10}
Amos Beimel, Shiva~Prasad Kasiviswanathan, and Kobbi Nissim.
\newblock Bounds on the sample complexity for private learning and private data
  release.
\newblock In {\em TCC}, pages 437--454. Springer, Feb 9--11 2010.

\bibitem[BLR08]{BlumLiRo08}
Avrim Blum, Katrina Ligett, and Aaron Roth.
\newblock A learning theory approach to non-interactive database privacy.
\newblock In {\em STOC}. {ACM}, May 17--20 2008.

\bibitem[BN08]{BonehNa08}
Dan Boneh and Moni Naor.
\newblock Traitor tracing with constant size ciphertext.
\newblock In {\em CCS}, pages 501--510. ACM, 2008.

\bibitem[BNS13a]{BeimelNiSt13a}
Amos Beimel, Kobbi Nissim, and Uri Stemmer.
\newblock Characterizing the sample complexity of private learners.
\newblock In {\em ITCS}, pages 97--110. {ACM}, Jan 9--12 2013.

\bibitem[BNS13b]{BeimelNiSt13b}
Amos Beimel, Kobbi Nissim, and Uri Stemmer.
\newblock Private learning and sanitization: Pure vs. approximate differential
  privacy.
\newblock In {\em APPROX-RANDOM}, pages 363--378. Springer, Aug 21--23 2013.

\bibitem[BNSV15]{BunNiStVa15}
Mark Bun, Kobbi Nissim, Uri Stemmer, and Salil~P. Vadhan.
\newblock Differentially private release and learning of threshold functions.
\newblock In {\em FOCS}, 2015.

\bibitem[BS98]{BonehSh98}
Dan Boneh and James Shaw.
\newblock Collusion-secure fingerprinting for digital data.
\newblock {\em IEEE Transactions on Information Theory}, 44(5):1897--1905,
  1998.

\bibitem[BSSU15]{BassilySmStUl15}
Raef Bassily, Adam Smith, Thomas Steinke, and Jonathan Ullman.
\newblock More general queries and less generalization error in adaptive data
  analysis.
\newblock {\em CoRR}, abs/1503.04843, 2015.

\bibitem[BST14]{BassilySmTh14}
Raef Bassily, Adam Smith, and Abhradeep Thakurta.
\newblock Private empirical risk minimization: Efficient algorithms and tight
  error bounds.
\newblock In {\em {FOCS}}, pages 464--473. {IEEE}, October 18--21 2014.

\bibitem[CTUW14]{ChandrasekaranThUlWa13}
Karthekeyan Chandrasekaran, Justin Thaler, Jonathan Ullman, and Andrew Wan.
\newblock Faster private release of marginals on small databases.
\newblock {\em ITCS 2014 (to appear)}, 2014.

\bibitem[De12]{De12}
Anindya De.
\newblock Lower bounds in differential privacy.
\newblock In {\em TCC}, pages 321--338, 2012.

\bibitem[DFH{\etalchar{+}}15]{DworkFeHaPiReRo15}
Cynthia Dwork, Vitaly Feldman, Moritz Hardt, Toniann Pitassi, Omer Reingold,
  and Aaron~Leon Roth.
\newblock Preserving statistical validity in adaptive data analysis.
\newblock In {\em {STOC}}, pages 117--126. {ACM}, 14--17 Jun 2015.

\bibitem[DJW13]{DuchiJoWa13}
John~C. Duchi, Michael~I. Jordan, and Martin~J. Wainwright.
\newblock Local privacy and statistical minimax rates.
\newblock In {\em 54th Annual {IEEE} Symposium on Foundations of Computer
  Science, {FOCS} 2013, 26-29 October, 2013, Berkeley, CA, {USA}}, pages
  429--438, 2013.

\bibitem[DKM{\etalchar{+}}06]{DworkKeMcMiNa06}
Cynthia Dwork, Krishnaram Kenthapadi, Frank McSherry, Ilya Mironov, and Moni
  Naor.
\newblock Our data, ourselves: Privacy via distributed noise generation.
\newblock In {\em EUROCRYPT}, pages 486--503. Springer, May 28--June 1 2006.

\bibitem[DMNS06]{DworkMcNiSm06}
Cynthia Dwork, Frank McSherry, Kobbi Nissim, and Adam Smith.
\newblock Calibrating noise to sensitivity in private data analysis.
\newblock In {\em TCC}, pages 265--284. Springer, Mar 4--7 2006.

\bibitem[DMT07]{DworkMcTa07}
Cynthia Dwork, Frank McSherry, and Kunal Talwar.
\newblock The price of privacy and the limits of lp decoding.
\newblock In {\em STOC}, pages 85--94, 2007.

\bibitem[DN03]{DinurNi03}
Irit Dinur and Kobbi Nissim.
\newblock Revealing information while preserving privacy.
\newblock In {\em PODS}, pages 202--210. {ACM}, June 9--12 2003.

\bibitem[DN04]{DworkNi04}
Cynthia Dwork and Kobbi Nissim.
\newblock Privacy-preserving datamining on vertically partitioned databases.
\newblock In {\em CRYPTO}, pages 528--544, Aug 15--19 2004.

\bibitem[DNR{\etalchar{+}}09]{DworkNaReRoVa09}
Cynthia Dwork, Moni Naor, Omer Reingold, Guy~N. Rothblum, and Salil~P. Vadhan.
\newblock On the complexity of differentially private data release: efficient
  algorithms and hardness results.
\newblock In {\em {STOC}}, pages 381--390, 2009.

\bibitem[DNT13]{DworkNiTa13}
Cynthia Dwork, Aleksandar Nikolov, and Kunal Talwar.
\newblock Efficient algorithms for privately releasing marginals via convex
  programming.
\newblock {\em Manuscript}, 2013.

\bibitem[DNV12]{DworkNaVa12}
Cynthia Dwork, Moni Naor, and Salil~P. Vadhan.
\newblock The privacy of the analyst and the power of the state.
\newblock In {\em FOCS}, pages 400--409. IEEE Computer Society, 2012.

\bibitem[DRV10]{DworkRoVa10}
Cynthia Dwork, Guy~N. Rothblum, and Salil~P. Vadhan.
\newblock Boosting and differential privacy.
\newblock In {\em FOCS}, pages 51--60, Oct 23--26 2010.

\bibitem[DS01]{DubhashiSe01}
Devdatt~P. Dubhashi and Sandeep Sen.
\newblock Concentration of measure for randomized algorithms: techniques and
  applications.
\newblock In Handbook of Randomized Algorithms, 2001.

\bibitem[DSS{\etalchar{+}}15]{DworkSmStUlVa15}
Cynthia Dwork, Adam Smith, Thomas Steinke, Jonathan Ullman, and Salil Vadhan.
\newblock Robust traceability from trace amounts.
\newblock In {\em {FOCS}}. {IEEE}, Oct 17--20 2015.

\bibitem[DTTZ14]{DworkTaThZh14}
Cynthia Dwork, Kunal Talwar, Abhradeep Thakurta, and Li~Zhang.
\newblock Analyze gauss: optimal bounds for privacy-preserving principal
  component analysis.
\newblock In {\em Symposium on Theory of Computing {STOC}}, pages 11--20.
  {ACM}, May 31--June 3 2014.

\bibitem[DY08]{DworkYe08}
Cynthia Dwork and Sergey Yekhanin.
\newblock New efficient attacks on statistical disclosure control mechanisms.
\newblock In {\em CRYPTO}, pages 469--480, 2008.

\bibitem[GHRU11]{GuptaHaRoUl11}
Anupam Gupta, Moritz Hardt, Aaron Roth, and Jonathan Ullman.
\newblock Privately releasing conjunctions and the statistical query barrier.
\newblock In {\em STOC}, pages 803--812. {ACM}, 2011.

\bibitem[GRU12]{GuptaRoUl12}
Anupam Gupta, Aaron Roth, and Jonathan Ullman.
\newblock Iterative constructions and private data release.
\newblock In {\em TCC}, pages 339--356, 2012.

\bibitem[Har11]{HardtThesis}
Moritz Hardt.
\newblock {\em A Study in Privacy and Fairness in Sensitive Data Analysis}.
\newblock PhD thesis, Princeton University, 2011.

\bibitem[HLM12]{HardtLiMc12}
Moritz Hardt, Katrina Ligett, and Frank McSherry.
\newblock A simple and practical algorithm for differentially private data
  release.
\newblock In {\em NIPS}, 2012.

\bibitem[HR10]{HardtRo10}
Moritz Hardt and Guy~N. Rothblum.
\newblock A multiplicative weights mechanism for privacy-preserving data
  analysis.
\newblock In {\em FOCS}, pages 61--70. {IEEE}, Oct 23--26 2010.

\bibitem[HSR{\etalchar{+}}08]{Homer+08}
Nils Homer, Szabolcs Szelinger, Margot Redman, David Duggan, Waibhav Tembe,
  Jill Muehling, John~V Pearson, Dietrich~A Stephan, Stanley~F Nelson, and
  David~W Craig.
\newblock Resolving individuals contributing trace amounts of dna to highly
  complex mixtures using high-density snp genotyping microarrays.
\newblock {\em PLoS genetics}, 2008.

\bibitem[HT10]{HardtTa10}
Moritz Hardt and Kunal Talwar.
\newblock On the geometry of differential privacy.
\newblock In {\em STOC}, pages 705--714, 2010.

\bibitem[HU14]{HardtUl14}
Moritz Hardt and Jonathan Ullman.
\newblock Preventing false discovery in interactive data analysis is hard.
\newblock In {\em {FOCS}}. IEEE, October 19-21 2014.

\bibitem[KLN{\etalchar{+}}11]{KasiviswanathanLeNiRaSm07}
Shiva~Prasad Kasiviswanathan, Homin~K. Lee, Kobbi Nissim, Sofya Raskhodnikova,
  and Adam Smith.
\newblock What can we learn privately?
\newblock {\em SIAM J. Comput.}, 40(3):793--826, 2011.

\bibitem[KRSU10]{KasiviswanathanRuSmUl10}
Shiva~Prasad Kasiviswanathan, Mark Rudelson, Adam Smith, and Jonathan Ullman.
\newblock The price of privately releasing contingency tables and the spectra
  of random matrices with correlated rows.
\newblock In {\em STOC}, pages 775--784, 2010.

\bibitem[LMTU14]{LibertyMiThUl14}
Edo Liberty, Michael Mitzenmacher, Justin Thaler, and Jonathan Ullman.
\newblock Space lower bounds for itemset frequency sketches.
\newblock {\em CoRR}, abs/1407.3740, 2014.

\bibitem[NTZ13]{NikolovTaZh13}
Aleksandar Nikolov, Kunal Talwar, and Li~Zhang.
\newblock The geometry of differential privacy: the sparse and approximate
  cases.
\newblock In {\em STOC}, pages 351--360, 2013.

\bibitem[Rot10]{Roth10}
Aaron Roth.
\newblock Differential privacy and the fat-shattering dimension of linear
  queries.
\newblock In {\em APPROX-RANDOM}, pages 683--695, 2010.

\bibitem[RR10]{RothRo10}
Aaron Roth and Tim Roughgarden.
\newblock Interactive privacy via the median mechanism.
\newblock In {\em STOC}, pages 765--774. {ACM}, 2010.

\bibitem[Rud12]{Rudelson11}
Mark Rudelson.
\newblock Row products of random matrices.
\newblock {\em Advances in Mathematics}, 231(6):3199--3231, 2012.

\bibitem[SOJH09]{SankararamonObJoHa09}
Sriram Sankararaman, Guillaume Obozinski, Michael~I Jordan, and Eran Halperin.
\newblock Genomic privacy and limits of individual detection in a pool.
\newblock {\em Nature genetics}, 41(9):965--967, 2009.

\bibitem[SU15a]{SteinkeUl15a}
Thomas Steinke and Jonathan Ullman.
\newblock Between pure and approximate differential privacy.
\newblock {\em CoRR}, abs/1501.06095, 2015.

\bibitem[SU15b]{SteinkeUl15b}
Thomas Steinke and Jonathan Ullman.
\newblock Preventing false discovery in interactive data analysis is hard.
\newblock In {\em {COLT}}. {JMLR}.org, July 3--6 2015.

\bibitem[Tar08]{Tardos08}
G{\'a}bor Tardos.
\newblock Optimal probabilistic fingerprint codes.
\newblock {\em J. ACM}, 55(2), 2008.

\bibitem[TUV12]{ThalerUlVa12}
Justin Thaler, Jonathan Ullman, and Salil~P. Vadhan.
\newblock Faster algorithms for privately releasing marginals.
\newblock In {\em ICALP (1)}, pages 810--821, 2012.

\bibitem[Ull13]{Ullman13}
Jonathan Ullman.
\newblock Answering $n^{2+o(1)}$ counting queries with differential privacy is
  hard.
\newblock In {\em STOC}, pages 361--370, 2013.

\bibitem[UV11]{UllmanVa11}
Jonathan Ullman and Salil~P. Vadhan.
\newblock {PCP}s and the hardness of generating private synthetic data.
\newblock In {\em TCC}, pages 400--416, 2011.

\bibitem[Vad16]{Vadhan16}
Salil Vadhan.
\newblock The complexity of differential privacy, 2016.

\end{thebibliography}

\appendix

\section{Lower Bounds on Fingerprinting Codes via Differential Privacy}\label{sec:gaussianmech}
By the contrapositive of Theorem~\ref{thm:fpctolb}, upper bounds on the sample complexity of answering $1$-way marginals with differential privacy imply a lower bound on the length $d$ of any fingerprinting code with a given number of users $n$.  As pointed out to us by Adam Smith, this yields a particularly simple, self-contained proof of Tardos'~\cite{Tardos08} optimal lower bound on the length of fingerprinting codes.  Specifically, using the well known Gaussian mechanism for achieving differential privacy, we can design a simple adversary $\fpadv$ that violates the security of any traitor tracing scheme with length $d = o(n^2)$.

\begin{theorem}
There is a function $n = n(d) = \tilde{O}(\sqrt{d})$ such that for every $d$, there is no $(n, d)$-fingerprinting code with security $\sec < 1/6en$.
\end{theorem}

\begin{proof}
Before diving into the proof, we will state the following elementary fact about Gaussian random variables.  The fact simply says that a Gaussian random variable with suitable variance is ``close'' to a shifted version of itself in a particular sense.  This same fact is used to show that adding Gaussian noise of suitable variance provides differential privacy.
\begin{fact} \label{fact:gaussiandp}
Let $c, c' \in \R^d$ satisfy $\| c - c' \|_2 \leq \sqrt{d}/n,$ $\delta > 0$ be a parameter, and let $\sigma^2 = 2d\ln(1/\delta)/n^2.$  Let $z \in \R^d$ be a random vector where each coordinate is an independent draw from a Gaussian distribution with mean $0$ and variance $\sigma^2.$  Then for any (measurable) set $T \subseteq \R^d$.
$$
\Prob{z}{c + z \in T} \geq (1/e) \Prob{z}{c' + z \in T} - \delta.
$$
\end{fact}

Now we proceed with the proof.
Fix any choice of $d$.  Assume towards a contradiction that there is an $(n, d)$-fingerprinting code $(\gen, \trace)$ with security $\sec < 1/6en$ for $n = \left\lceil\sqrt{18 d \ln(6en) \ln(3d/2)}\right\rceil$.  Observe that $n = n(d) = \tilde{O}(\sqrt{d})$ as promised in the theorem.

Let $\fpadv(\codebook_{S})$ be the following adversary.  Define the vector $\overline{c} \in [0,1]^d$ as
$$\overline{c} = \frac{1}{n} \sum_{i \in S} c_{i}.$$
Now, let $z \in \R^d$ be a $d$-dimensional Gaussian where every coordinate is independent with mean $0$ and variance $\sigma^2 = 2 d \ln(1/\delta) / n^2,$ for $\delta = 1/6en.$ Finally, let $c'$ be $\hat{c}$ with each coordinate rounded to $\bits,$ and output the pirated codeword $c'.$

First we claim that $\fpadv$ outputs feasible codewords with at least constant probability.
\begin{claim} \label{clm:gaussianisfeasible}
For every $S$ such that $|S| \geq n-1,$ and every codebook $C = (c_{ij}) \in \bits^{n \times d},$
$$
\Prob{c' \getsr \fpadv(C_{S})}{c' \in F(C_{S})} \geq 2/3.
$$
\end{claim}
\begin{proof}[Proof of Claim~\ref{clm:gaussianisfeasible}]
By a standard tail bound for the Gaussian, we have
$$
\Prob{}{\forall \, j, \; |z_j| < \sigma \sqrt{\ln (3d/2)}} \geq 2/3.
$$
Thus, by our choice of $\sigma$ and $n \geq \sqrt{18 d \ln(1/\delta) \ln(3d/2)}$ we have
$
\Prob{}{\forall \, j, \; |z_j| < 1/3} \geq 2/3.
$
Now the claim follows easily.  Specifically, if $c_{ij} = 1$ for every $i \in S$, then $(1/n) \sum_{i \in S} c_{ij} \geq 1-1/n$, so $\hat{c}_j > 2/3 - 1/n$ and $c'_j = 1$.  A similar argument applies if $c_{ij} = 0$ for every $i \in S$.
\end{proof}

Now it remains to show that $\fpadv$ cannot be traced successfully.  By assumption $(\gen, \trace)$ has security $\sec < 1/6en <  1/3.$  Then we have in particular
$$
\Prob{\codebook \getsr \gen \atop c' \getsr \fpadv(\codebook)}{\pirateword \in F(\codebook) \land \trace(\codebook, \pirateword) = \bot} < \sec.
$$
Combining with Claim~\ref{clm:gaussianisfeasible} we have
$$
\Prob{\codebook \getsr \gen \atop c' \getsr \fpadv(\codebook)}{\trace(\codebook, c') \in [n]} > 1 - 1/3 - \sec > 1/3.
$$
Therefore, there exists $i^* \in [n]$ such that
\begin{equation} \label{eq:usetracing}
\Prob{\codebook \getsr \gen \atop c' \getsr \fpadv(\codebook)}{\trace(\codebook, c') = i^*} > 1/3n.
\end{equation}
To complete the proof, it now suffices to show that if $S = [n] \setminus \set{i^*}$, then
$$
\Prob{\codebook \getsr \gen \atop c' \getsr \fpadv(\codebook_{S})}{\trace(\codebook, c') = i^*} \geq 1/6en > \sec,
$$
which will contradict the security of the fingerprinting code.

To do so, first observe that if 
$$
\overline{c} = \frac{1}{n} \sum_{i \in [n]} c_{i}, \qquad \textrm{and} \qquad \overline{c}^{S} = \frac{1}{n} \sum_{i \in S} c_{i},
$$
then $\| \overline{c}_j - \overline{c}^{S}_j \|_2 \leq \sqrt{d}/n.$  Now, in case the tracing algorithm is randomized, let $\trace_r$ denote the tracing algorithm when run with its random coins fixed to $r$.  For any string of random coins $r$, define the set $T_{r} = \{t \in \R^d \mid \trace_r(C, \mathrm{round}(t)) = i^*\}.$  Here, $\mathrm{round}(\cdot)$ is the function that rounds each entry of its input to $\bits.$\footnote{Note, for completeness, that $T_{r}$ is measurable, since the set of $c' \in \bits^d$ such that $\trace_r(C, c') = i^*$ is finite (for every fixed $n, d$) and for every $c',$ $\set{t \mid \mathrm{round}(t) = c'}$ is a hypercube, so $T_{r}$ is a union of finitely many hypercubes.}

By Fact~\ref{fact:gaussiandp} (with $\delta = 1/6en > \sec$), for every $r$,
$$
\Prob{z}{\overline{c}^{S} + z \in T_{r}} \geq (1/e) \Prob{z}{\overline{c} + z \in T_{r}} - \sec.
$$  Applying~\eqref{eq:usetracing}, and averaging over $\codebook \getsr \gen$ and $r$, we have
\begin{align*}
\Prob{\codebook \getsr \gen \atop c' \getsr \fpadv(\codebook_{S})}{\trace(\codebook, c') = i^*} \geq (1/e)(1/3n) - 1/6en = 1/6en > \sec,
\end{align*}
which is the desired contradiction.  This completes the proof.
\end{proof}

\end{document}